\documentclass[oneside,11pt]{article}
\usepackage[utf8]{inputenc}
\usepackage[T1]{fontenc}

\usepackage{amsthm}
\usepackage{amsmath}
\usepackage{amsfonts}
\usepackage{amssymb}

\usepackage{graphicx} 
\usepackage{graphbox}
\usepackage{float}
\usepackage{booktabs}
\usepackage{multirow}
\usepackage{rotating}
\usepackage{array}
\usepackage{xcolor}
\usepackage{subcaption}
\usepackage{tabularx}
\usepackage{caption}
\usepackage{bm}
\usepackage{algpseudocode,algorithmicx,algorithm}

\usepackage{comment}

\newcolumntype{C}[1]{>{\centering\arraybackslash}p{#1}}
\newcolumntype{R}[1]{>{\raggedleft\let\newline\\\arraybackslash\hspace{0pt}}m{#1}}
\newcolumntype{L}[1]{>{\raggedright\let\newline\\\arraybackslash\hspace{0pt}}m{#1}}

\usepackage{breakcites}

\usepackage{nowidow}

\usepackage{enumitem}

\usepackage[top=1.5cm,bottom=2cm,right=2.5cm,left=2.5cm]{geometry}
\usepackage{titling}

\usepackage{natbib}

\usepackage{ifplatform}

\usepackage{textcomp}
\ifwindows
\fi

\allowdisplaybreaks

\newcommand{\Var}{\operatorname{\mathbb{V}ar}}
\newcommand{\Cov}{\operatorname{\mathbb{C}ov}}

\DeclareMathAlphabet\mathbfcal{OMS}{cmsy}{b}{n}

\DeclareFontFamily{OT1}{pzc}{}
\DeclareFontShape{OT1}{pzc}{m}{it}{<-> s * [1.10] pzcmi7t}{}
\DeclareMathAlphabet{\mathpzc}{OT1}{pzc}{m}{it}

\newtheorem{theorem}{Theorem}[section]

\newtheorem{remark}{Remark}[section]
\newtheorem{proposition}{Proposition}[section]

\usepackage[pdftex,final,bookmarksnumbered,bookmarksopen=false,breaklinks,colorlinks]{hyperref}
\hypersetup{
	final,
	bookmarksnumbered=true,
	bookmarksopen=true,
	bookmarksopenlevel=0,
	unicode=false,
	pdftoolbar=true,
	pdfmenubar=true,
	pdffitwindow=false,
	pdftitle={Blessing of dimensionality in cross-validated bandwidth selection on the sphere},
	pdfdisplaydoctitle=true,
	pdftoolbar=true,
	pdfmenubar=true,
	pdflang={English},
	pdfauthor={Jose E. Chacon, Eduardo Garcia-Portugues, Andrea Meilan-Vila},
	pdfsubject={arXiv paper},
	pdfcreator={José E. Chacón},
	pdfproducer={José E. Chacón},
	pdfkeywords={Directional data}{High-dimensional data}{Nonparametric statistics}{Smoothing},
	pdfnewwindow=true,
	breaklinks=true,
	hidelinks,
	linkcolor=black,
	citecolor=black,
	filecolor=magenta,
	urlcolor=cyan
}

\newif\ifmain
\maintrue
\newif\ifsupplement
\supplementtrue

\newif\iffigstabs
\figstabstrue

\begin{document}

\ifmain

%-----------------------------------------------%
\title{Analyzing directional errors in spatial orientation using nonparametric circular regression with mixed covariates}
\setlength{\droptitle}{-1cm}
\predate{}%
\postdate{}%
\date{}
%-----------------------------------------------%

%-----------------------------------------------%
\author{Mario Francisco-Fernández$^{1,3}$ and Andrea Meil\'an-Vila$^{2}$}
\footnotetext[1]{Department of Mathematics, Universidade da Coruña (Spain).}
\footnotetext[2]{Department of Statistics, Universidad Carlos III de Madrid (Spain).}
\footnotetext[3]{Corresponding author. e-mail: \href{mailto:mario.francisco@udc.es}{mario.francisco@udc.es}.}
\maketitle
%-----------------------------------------------%

\begin{abstract}
Spatial orientation is a fundamental cognitive skill that relies on sensory information to update perceived direction.
Understanding how sensory conditions influence directional accuracy is important for both cognitive science and
the design of assistive technologies. We analyze experimental data in which blind, low-vision, and sighted
participants performed spatial updating tasks under five sensory conditions, with signed angular error as the
response.
To model these data, we propose a nonparametric circular regression framework that accommodates both
continuous and categorical predictors via a product-kernel estimator. Bandwidth selection is crucial in this
setting, yet developing practical data-driven methods remains
challenging. We derive asymptotic bias and variance expressions for the estimator, though these results do not
directly lead to a feasible plug-in bandwidth selector. To address this, we develop a bootstrap bandwidth selection criterion tailored to the
cosine loss and compare it with cross-validation and rule-of-thumb approaches in simulation studies.
Applied to the spatial updating data, the proposed framework reveals nonlinear, condition-specific patterns and
quantifies uncertainty via simultaneous bootstrap confidence bands. Across the scenarios considered, the proposed
bootstrap selector achieves a favorable bias-variance trade-off and yields stable inference relative to the
competing methods. An implementation is available in the \textsf{R} package \texttt{circMixedReg}.
\end{abstract}
\begin{flushleft}
	\small\textbf{Keywords:} Kernel smoothing; Bootstrap methods; Bandwidth selection; Circular data analysis; Human orientation.
\end{flushleft}

\section{Introduction}
\label{sec:intro}

Spatial orientation enables humans to locate themselves in their environment and navigate effectively by integrating multiple sensory inputs while moving. Accurate estimation of direction (and distance) is crucial for everyday activities, from moving safely through familiar spaces to exploring novel ones. These processes can become particularly challenging for individuals with visual impairments or in situations where sensory cues are limited, degraded, or conflicting, leading to systematic directional biases and increased uncertainty. Quantifying how sensory conditions affect spatial updating is therefore relevant both for advancing our understanding of cognitive and neural mechanisms and for informing the design of accessible environments, assistive technologies, and rehabilitation strategies.

A growing body of research suggests that spatial representations can be preserved and updated without vision, albeit with increased cognitive demands and modality-specific limitations. For instance, \citet{klatzky2006cognitive} showed that navigation without visual input imposes substantially greater cognitive demands, even in familiar settings. Complementarily, \citet{giudice2011functional} provided evidence that spatial images derived from touch and vision can function equivalently in supporting spatial cognition. Non-visual modalities such as audition and proprioception may therefore partially compensate for the absence of vision: \citet{papadopoulos2012working} highlighted the role of auditory spatial working memory and proprioception for navigation in blind individuals, and \citet{schinazi2016navigation} reviewed findings indicating that congenitally blind individuals can navigate effectively using non-visual cues, although certain aspects of spatial learning may be limited without visual experience. At the same time, early visual input appears to contribute to the calibration and refinement of spatial representations; for example, \citet{gori2014impairment} reported deficits in auditory spatial localization among congenitally blind individuals. Taken together, these results support a multimodal, experience-dependent view of spatial orientation and motivate statistical tools that can quantify condition-specific directional error patterns while accommodating heterogeneous participant groups.

% \begin{figure}[!htb]
%   \centering
%   \includegraphics[width=0.8\textwidth]{legge_combined_figure_rec.pdf}
%   \caption{Spatial-updating task used by \citet{legge2016indoor}. \textbf{(A)} Plan view of a three-segment path with two turns. The target is a beanbag dropped at the first waypoint, and the start location is the doorway. \textbf{(B)} Sensory conditions considered in the experiment, indicating which cues were available (visual preview, vision during walking, auditory information) and whether looking back was allowed. After completing the path, participants reported distance and direction to both the start and the target.}
%   \label{fig:data}
% \end{figure}

In the present paper, we analyze data from a spatial-updating experiment conducted by \citet{legge2016indoor}.
In that study, participants traversed predefined indoor routes composed of three straight segments with two turns and, at the end of each route, indicated the direction (and distance) to a previously viewed target (Figure~\ref{fig:data}). The design was intended to emulate indoor wayfinding demands under reduced sensory input. The dataset comprises 679 trials from 64 participants (32 sighted, 16 low-vision, and 16 blind), with each participant completing multiple trials under five sensory conditions that manipulate the availability of visual and/or auditory information: \textit{Control} (full vision and hearing), \textit{Preview} (a brief visual preview of the target followed by blindfolding during locomotion and estimation), \textit{Forward Facing} (vision available during locomotion but restricted during the pointing/estimation phase), \textit{Auditory} (blindfolded, relying primarily on self-generated auditory cues such as footsteps and echoes), and \textit{Deprivation} (blindfolded with noise-canceling headphones). For each trial, the dataset records the true and reported target direction and distance, along with trial- and participant-level identifiers. This experiment is part of a broader line of work on indoor navigation under reduced sensory input \citep[see also][]{legge2016plos}.

 \begin{figure}[!htb]
     \centering
     \includegraphics[width=0.75\linewidth]{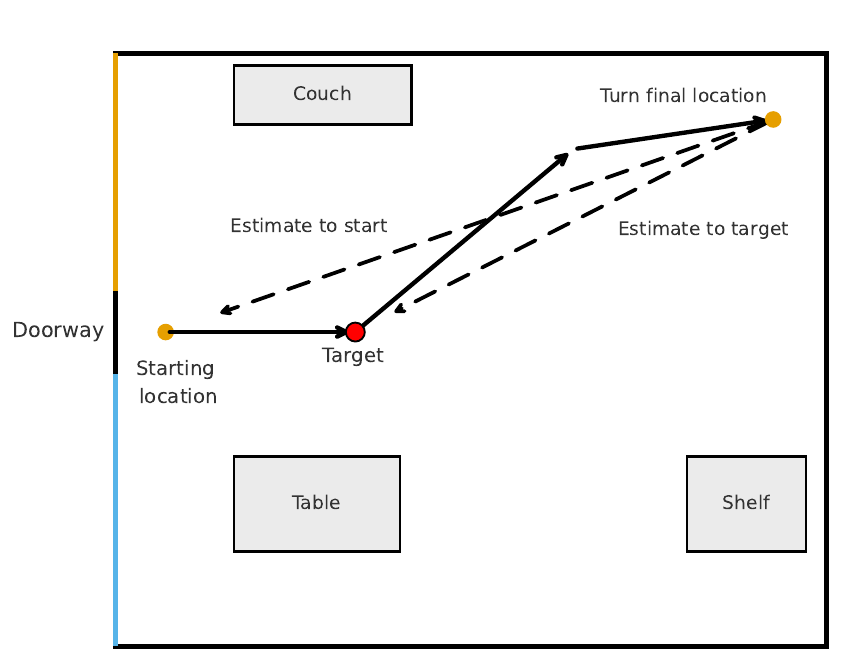}
     \caption{Spatial-updating task used by \citet{legge2016indoor}. Participants are guided along a three-segment path with two turns to a final location. The target is a beanbag dropped at the first waypoint, and the start location is the doorway. After completing the path, participants reported distance and direction to both the start and the target.
     }
     \label{fig:data}
 \end{figure}

%  \begin{table}[!htb]
% \centering
% \small
% \begin{threeparttable}
% \caption{Sensory conditions in the spatial-updating task. Availability of visual and auditory cues and whether looking back was allowed.}
% \label{tab:sens_cond}

% \setlength{\tabcolsep}{6pt}
% \renewcommand{\arraystretch}{1.15}
% \begin{tabular}{lccccc}
% \toprule
%  & Control & Preview & Forward Facing & Auditory & Deprivation \\
% \midrule
% Vision during walk            & \cmark & \xmark & \cmark & \xmark & \xmark \\
% 10\,s preview at doorway      & \xmark & \cmark & \xmark & \xmark & \xmark \\
% Audio available (not masked)  & \cmark & \xmark & \cmark & \cmark & \xmark \\
% Looking back permitted           & \cmark & \xmark & \xmark & \xmark & \xmark \\
% \bottomrule
% \end{tabular}

% \begin{tablenotes}[flushleft]
% \footnotesize
% \item \textit{Note:} ``Audio available (not masked)'' indicates whether participants could use ambient/self-generated auditory cues.
% In \textit{Preview} and \textit{Deprivation}, auditory cues were masked (earmuffs + white noise). In \textit{Auditory}, participants were blindfolded but auditory cues were available.
% \end{tablenotes}
% \end{threeparttable}
% \end{table}

In our analysis, the primary outcome is the \emph{signed} circular error, defined as the angular difference between the reported and true target directions, wrapped to $(-\pi,\pi]$. This choice contrasts with the \emph{absolute} directional error summaries reported in \citet{legge2016indoor}. Keeping the sign preserves information about systematic directional bias (clockwise vs.\ counterclockwise), which is essential for modeling the conditional mean direction in a circular regression framework. We examine how directional error varies with task demand and sensory restriction using sensory condition as a categorical predictor and target distance as a continuous covariate. In additional analyses, we also consider the distance-estimation error (reported minus true distance), a signed measure that distinguishes under- from overestimation either as an additional covariate or as an alternative continuous predictor. We use the signed version \citep[rather than the absolute summaries emphasized in][]{legge2016indoor} to retain information about systematic bias. Because each participant contributes multiple trials, we conduct the analysis at the trial level, with interest in the population-averaged relationship between directional error and the covariates across trials. We interpret the results accordingly.

Modeling circular deviations presents unique challenges because the response lies on the circle,
making standard linear regression inappropriate \citep{fisher1995statistical,jammalamadaka2001topics}.
Moreover, in the spatial-updating setting described above, the predictors naturally combine a
continuous measure of task demand with a categorical sensory-condition factor, calling for regression
methods that accommodate mixed covariates. Circular nonparametric regression provides a flexible
alternative, enabling smooth estimation of the relationship between a circular response and covariates
without restrictive parametric assumptions. In practice, local smoothing estimators such as
Nadaraya--Watson (NW) and local linear (LL) methods have been studied for a single continuous covariate
\citep{dimarzio2013nonparametric} and for multiple continuous covariates \citep{meilan2021nonparametricTEST},
with extensions to functional predictors, time series, and spatially correlated errors
\citep{meilan2024functional,di2012non,meilan2021trendSURF}; see \citet{alonso2025review} for an overview.
However, to the best of our knowledge, practical methodology for circular regression with mixed covariates, particularly data-driven bandwidth selection and uncertainty quantification, remains largely unexplored.

In this paper, we develop a nonparametric circular regression framework for mixed covariates, accommodating both
continuous and categorical predictors via a product-kernel NW estimator, motivated by
applications such as the spatial-updating data described above. The approach is inspired by mixed-covariate
smoothing in Euclidean settings \citep{racine2004nonparametric} and adapted to circular responses. Our main
contributions are: (i) derivation of the asymptotic bias and variance of the  proposed estimator under general
regularity conditions; (ii) development and comparison of three bandwidth selection strategies: a simple, fully explicit rule-of-thumb, in the spirit of classical plug-in rule, a leave-one-out cross-validation criterion, and a bootstrap criterion
based on a cosine loss \citep{kim2017multivariate}, with the bootstrap approach offering improved stability
empirically; (iii) construction of simultaneous confidence bands via bootstrap, including an iterative
calibration procedure to improve simultaneous coverage and compatible with any data-driven bandwidth choice; and
(iv) an application to the spatial orientation data described above \citep{legge2016indoor}, illustrating
nonlinear, condition-specific directional error patterns. All methods are implemented in the \textsf{R} package \texttt{circMixedReg} \citep{circMixedReg2025}.

The rest of the paper is organized as follows. Section~\ref{sec:statistical_model} introduces the circular regression model, the proposed product-kernel estimator, and its main asymptotic properties. Section~\ref{sec:bw} presents the bandwidth selection strategies, including the bootstrap-based method for circular responses, as well as cross-validation and rule-of-thumb approaches. Section~\ref{sec:bands} describes the construction of simultaneous confidence bands for the circular regression function. Section~\ref{sec:simulations} reports a simulation study assessing the performance of the competing selectors across several scenarios. Section~\ref{sec:results} presents the real-data analysis of spatial orientation under different sensory conditions, and Section~\ref{sec:discussion} concludes with a discussion and directions for future work. %Proofs, technical assumptions, and additional real-data results are provided in the Supplementary Material (SM).
Proofs, technical assumptions, and additional real-data results are provided in the Supplementary Material.

\section{The circular regression model and its estimator}
\label{sec:statistical_model}

Let \( \{(\bm{X}_i, \bm{Z}_i, \Theta_i)\}_{i=1}^n \) be a random sample of \((\bm{X}, \bm{Z}, \Theta)\), where $\Theta$ is a circular random variable taking values on $[0, 2\pi)$. The vector $\bm{X}_i = (X_{i1}, \ldots, X_{ik}) \in \mathbb{R}^k$ represents continuous covariates, and $\bm{Z}_i = (Z_{i1}, \ldots, Z_{ip}) \in \mathcal{Z}_1 \times \cdots \times \mathcal{Z}_p$ contains categorical covariates, with each \(\mathcal{Z}_l\) a finite set. We assume that the explanatory variables \((\bm{X}, \bm{Z})\) have a joint distribution with density \(f(\bm{x}, \bm{z})\). Although \(\bm{Z}\) is discrete, we follow standard convention in the literature \citep[see, e.g.,][]{racine2004nonparametric} and refer to \(f(\bm{x}, \bm{z})\) as a \textit{density}. 

We consider the following circular regression model:
\begin{equation}
\Theta_i = [m(\bm{X}_i, \bm{Z}_i) + \varepsilon_i] \quad (\operatorname{mod}\, 2\pi), \quad i = 1, \ldots, n,
\label{eq:model}
\end{equation}
where \( \varepsilon_i \), $i=1, \ldots, n$ are independent realizations of a circular variable $\varepsilon$ satisfying
$\mathbb{E}[\sin(\varepsilon)\mid \bm{X}=\bm{x},\bm{Z}=\bm{z}] = 0$ and
$\mathbb{E}[\cos(\varepsilon)\mid \bm{X}=\bm{x},\bm{Z}=\bm{z}]>0$.

Given the circular regression model defined in Equation~\eqref{eq:model}, our objective is to estimate the circular regression function $m$, which is defined as the minimizer of the angular risk:
\begin{equation}
 \mathbb{E}\left\{1 - \cos\left[\Theta - m(\bm{X}, \bm{Z})\right] \mid \bm{X} = \bm{x}, \bm{Z} = \bm{z}\right\}.
\label{eq:risk}
\end{equation}
 The solution to this optimization problem is given by \citep[see][]{jammalamadaka2001topics}:
\begin{equation}\label{eq:reg}
m(\bm{x}, \bm{z}) =\operatorname{atan2}\left[ m_1(\bm{x}, \bm{z}), m_2(\bm{x}, \bm{z})\right],
\end{equation}
where $m_1(\bm{x}, \bm{z})=\mathbb{E}[\sin(\Theta) \mid \bm{X} = \bm{x}, \bm{Z} = \bm{z}] $, $m_2(\bm{x}, \bm{z})=\mathbb{E}[\cos(\Theta) \mid \bm{X} = \bm{x}, \bm{Z} = \bm{z}]$ and the function $\operatorname{atan2}(y,x)$ returns the principal angle between the positive $x$-axis and the vector from the origin to $(x,y)$ in the Cartesian plane. 
We assume that $(m_1(\bm{x},\bm{z}),m_2(\bm{x},\bm{z}))\neq (0,0)$ so that $m(\bm{x},\bm{z})$ is well-defined.

These functions can be seen as conditional expectations in two auxiliary models with responses \( \sin(\Theta) \) and \( \cos(\Theta) \), respectively,
\begin{align}
\sin(\Theta_i) &= m_1(\bm{X}_i, \bm{Z}_i) + \xi_i, \quad i=1,\dots,n, \label{model1} \\
\cos(\Theta_i) &= m_2(\bm{X}_i, \bm{Z}_i) + \zeta_i, \quad i=1,\dots,n, \label{model2}
\end{align}
where \( \xi_i \) and \( \zeta_i \) are bounded error terms with conditional mean zero (since $\sin(\Theta)$ and $\cos(\Theta)$ take values in $[-1,1]$). Their conditional variances and covariance are defined as \( s_1^2(\bm{x}, \bm{z}) =\Var(\xi \mid \bm{X} = \bm{x}, \bm{Z} = \bm{z}) \), \( s_2^2(\bm{x}, \bm{z}) =\Var(\zeta \mid \bm{X} = \bm{x}, \bm{Z} = \bm{z}) \), \( c(\bm{x}, \bm{z}) = \mathbb{E}(\xi \zeta \mid \bm{X} = \bm{x}, \bm{Z} = \bm{z}) \).

A NW estimator for \( m(\bm{x}, \bm{z}) \)  in \eqref{eq:reg} can be constructed by replacing the unknown functions $m_1(\bm{x}, \bm{z})$ and $m_2(\bm{x}, \bm{z})$ in its expression with suitable estimators, as follows:
\begin{equation}
\hat{m}_{\bm{H}}(\bm{x}, \bm{z}) =\operatorname{atan2}\left[ \hat{m}_{1,\bm{H}}(\bm{x}, \bm{z}), \hat{m}_{2,\bm{H}}(\bm{x}, \bm{z})  \right],
\label{eq:est}
\end{equation}
where
\begin{equation}
\hat{m}_{1,\bm{H}}(\bm{x}, \bm{z})=\sum_{i=1}^n w^{i}_{\bm{H}}(\bm{x}, \bm{z}) \sin (\Theta_i)
\label{est_m1}
\end{equation} 
and
\begin{equation}
\hat{m}_{2,\bm{H}}(\bm{x}, \bm{z})=\sum_{i=1}^n w^{i}_{\bm{H}}(\bm{x}, \bm{z}) \cos (\Theta_i),
\label{est_m2}
\end{equation}
with
\begin{equation}
w^{i}_{\bm{H}}(\bm{x}, \bm{z}) = \frac{K_{\bm{h}}(\bm{x}, \bm{X}_i)\,L_{\boldsymbol{\lambda}}(\bm{z}, \bm{Z}_i)}{\sum_{j=1}^n K_{\bm{h}}(\bm{x}, \bm{X}_j)\,L_{\boldsymbol{\lambda}}(\bm{z}, \bm{Z}_j)}.
\label{eq:pesos}
\end{equation}

Here, $\bm{H}= (\bm{h}, \boldsymbol{\lambda})$, where \(\bm{h} = (h_1, \ldots, h_k)\) and \(\boldsymbol{\lambda} = (\lambda_1, \ldots, \lambda_p)\) are vectors of bandwidths for the continuous and categorical covariates, respectively. The kernel product is defined by:
\[
K_{\bm{h}}(\bm{x}, \bm{X}_i) = \prod_{j=1}^k h_j^{-1} K_{j}[(x_j - X_{ij})/h_j], \quad L_{\boldsymbol{\lambda}}(\bm{z}, \bm{Z}_i) = \prod_{l=1}^p L_l(z_l, Z_{il};\lambda_l),
\]
where \(K_{j}\) are univariate kernels for Euclidean data (e.g., Gaussian or Epanechnikov), and \(L_{l}\) are categorical kernels such as Aitchison--Aitken or order-sensitive ones \citep{racine2004nonparametric}. This structure follows the product-kernel methodology proposed in \cite{racine2004nonparametric} and extends circular response models, such as those in \cite{dimarzio2013nonparametric} and \cite{meilan2021nonparametricTEST}, to accommodate both continuous and categorical covariates.

%It is important to note that adopting a product-kernel structure for the weights does not entail assuming independence between the covariates, and potential interactions between the continuous and categorical components in the underlying data-generating mechanism are allowed. This structure rather provides a convenient way to separate the smoothing across different types of predictors. The resulting estimator remains consistent under general conditions, even when the covariates exhibit dependence, as established by \citet{racine2004nonparametric}.

It is important to note that using a product-kernel structure for the weights does not imply any independence assumption on the covariates. The product form is simply a convenient way to perform joint smoothing over predictors of different types (continuous and categorical), while the regression function $m(\bm{x},\bm{z})$ remains fully nonparametric and is not restricted to be additive or separable in its arguments. In particular, the effect of a continuous covariate may vary across categories, and vice versa.

\begin{remark}
\label{rem:LL}
A LL version of the proposed estimator in~\eqref{eq:est} can be obtained by replacing the
local-constant fits of $m_1$ and $m_2$ with the fitted intercepts from weighted local-linear regressions in the
continuous covariates, using product-kernel weights
$K_{\bm h}(\bm x,\bm X_i)\,L_{\boldsymbol\lambda}(\bm z,\bm Z_i)$ in the corresponding weighted least-squares
criterion. In this formulation, local linearity is applied only to the continuous covariates, whereas the
categorical predictors enter solely through the multiplicative kernel component.
This LL construction parallels local-linear circular regression with continuous predictors
\citep{dimarzio2013nonparametric,meilan2021nonparametricTEST} while adopting the standard mixed-covariate
product-kernel weighting of \citet{racine2004nonparametric}. Although local linear smoothing can reduce bias,
in our simulations and real-data application the NW and LL fits were very similar, and NW was numerically more
stable, particularly for small bandwidths and sparse category combinations. For clarity, we adopt the NW estimator
as our default, while the \texttt{circMixedReg} package implements both NW and LL versions.
\end{remark}

We now derive the asymptotic properties of the NW estimator in~\eqref{eq:est} under regularity conditions (A1)--(A6), stated in Section~\ref{app:theory} of the Supplementary Material. Under these conditions, we first study the NW estimators of \(m_j(\bm{x},\bm{z})\), \(j=1,2\). Building on \citet{racine2004nonparametric}, we obtain expansions for the bias and variance of \(\hat m_{j,\bm H}(\bm{x},\bm{z})\), as well as the asymptotic covariance between \(\hat m_{1,\bm H}(\bm{x},\bm{z})\) and \(\hat m_{2,\bm H}(\bm{x},\bm{z})\). These results are collected in Proposition~\ref{prop1}. 
For simplicity in presentation and proof, and following \citet{racine2004nonparametric}, we consider common kernels and bandwidths across coordinates, namely \(K=K_1=\cdots=K_k\), \(L=L_1=\cdots=L_p\), \(\bm h=(h,\ldots,h)\), and \(\boldsymbol\lambda=(\lambda,\ldots,\lambda)\). This restriction is adopted for notational convenience and does not affect the applicability of the methods, which readily extend to coordinate-specific choices.

In what follows, \(\nabla_{\bm{x}} g(\bm{x},\bm{z})\) and \(\nabla_{\bm{x}}^2 g(\bm{x},\bm{z})\) denote the gradient and Hessian of a sufficiently smooth function \(g\) with respect to \(\bm{x}\). For any matrix \(\bm A\), \(\bm A^\top\) and \(\operatorname{tr}(\bm A)\) denote its transpose and trace, respectively.

\begin{proposition}
\label{prop1}
Let  \( \{(\bm{X}_i, \bm{Z}_i, \Theta_i)\}_{i=1}^n \) be a random sample from $(\bm X,\bm{Z}, \Theta)$.  Under assumptions  $({\rm A}1)$--$({\rm A}6)$, if $(\bm{x}, \bm{z})$ is an interior point of the support of $f$, then, for $j=1,2$,
\begin{align*}
\mathbb{E}[\hat m_{j, \bm{H}}(\bm{x},\bm{z})- m_j(\bm{x}, \bm{z})]
=&\,\mu_2(K)h^2 \left\{
\frac{\nabla_{\bm{x}} f(\bm{x}, \bm{z})^\top \nabla_{\bm{x}} m_j(\bm{x}, \bm{z})}{f(\bm{x}, \bm{z})}
+ \frac{1}{2} \operatorname{tr} \left[ \nabla_{\bm{x}}^2 m_j(\bm{x}, \bm{z}) \right]
\right\} \\
&\quad+\lambda \sum_{\substack{\tilde{\bm{z}} \in \mathcal{Z} \\ d_H(\tilde{\bm{z}}, \bm{z}) = 1}}
\left[ m_j(\bm{x}, \tilde{\bm{z}}) - m_j(\bm{x}, \bm{z}) \right]
 \frac{f(\bm{x}, \tilde{\bm{z}})}{f(\bm{x}, \bm{z})}
+ o(h^2 + \lambda),\\
\Var[\hat m_{j, \bm{H}}(\bm{x}, \bm{z})]
=&\,\frac{R(K)s_j^2(\bm{x}, \bm{z})}{n h^k f(\bm{x}, \bm{z})} + o\left( \frac{1}{n h^k} \right),\\
\Cov[\hat{m}_{1, \bm{H}}(\bm{x}, \bm{z}),\hat{m}_{2, \bm{H}}(\bm{x}, \bm{z})]
=&\,\frac{R(K)c(\bm{x}, \bm{z})}{n h^k  f(\bm{x}, \bm{z})}
+ o\left( \frac{1}{n h^k} \right).
\end{align*}
\noindent
Here \( \mu_2(K) = \int_{\mathbb{R}} u^2 K(u)\, du \), \( R(K) =  \left( \int K^2(u)\, du \right)^k<\infty \), and
$d_H(\tilde{\bm{z}}, \bm{z})$ is the Hamming distance between the categorical vectors $\tilde{\bm{z}}$ and $\bm{z}$.
\end{proposition}

\begin{remark}
In Proposition~\ref{prop1}, the bias has the usual continuous-component term of order $O(h^2)$ (as in standard multivariate NW smoothing) and, in addition, a categorical-smoothing contribution of order $O(\lambda)$. The latter arises from borrowing information across categories and involves only those $\tilde{\bm{z}}$ that differ from $\bm{z}$ in exactly one categorical component (i.e., $d_H(\tilde{\bm{z}},\bm{z})=1$). 
Regarding variability, the presence of categorical smoothing does not change the leading order $1/(n h^k)$. It enters only through a multiplicative factor depending on the categorical kernel and the distribution of $\bm Z$, under standard mixed-kernel regularity conditions (uniformly bounded kernels and a design density $f(\bm x,\bm z)$ bounded away from zero on the region where the kernel weights are non-negligible). These conclusions follow from the mixed-kernel results in \citet[Theorem~2.1]{racine2004nonparametric}. The same asymptotic orders hold when using vector bandwidths.
\end{remark}

Using Proposition~\ref{prop1}, the following theorem provides the leading terms of the conditional bias and variance of $\hat m_{\bm H}(\bm x,\bm z)$ in~\eqref{eq:est}.
%Proofs of Proposition~\ref{prop1} and Theorem~\ref{teo1} are provided in Section~S1 of Supplement~A.

\begin{theorem}\label{teo1}
	 Let  \( \{(\bm{X}_i, \bm{Z}_i, \Theta_i)\}_{i=1}^n \) be a random sample from $(\bm X,\bm{Z}, \Theta)$.  Under assumptions  $({\rm A}1)$--$({\rm A}6)$, the asymptotic conditional  bias and variance of estimator  $\hat{m}_{\bm{H}}(\bm{x}, \bm{z})$, at  a fixed interior point $(\bm{x}, \bm{z})$ in the support of $f$, are given by:
	\begin{align*}
	\mathbb{E}[\hat{m}_{\bm{H}}(\bm{x}, \bm{z})-m(\bm{x}, \bm{z})]
=&\,\mu_2(K)h^2\left\{
\frac{1}{\ell(\bm{x}, \bm{z}) f(\bm{x}, \bm{z})}\,
\nabla_{\bm{x}} (\ell f)(\bm{x}, \bm{z})^\top \nabla_{\bm{x}} m(\bm{x}, \bm{z})
\right.\\
&\left.
+ \frac{1}{2}\operatorname{tr}\!\left[ \nabla_{\bm{x}}^2 m(\bm{x}, \bm{z}) \right]
\right\} \\
&+\lambda\frac{1}{\ell(\bm{x}, \bm{z})}
\sum_{\substack{\tilde{\bm{z}} \in \mathcal{Z} \\ d_H(\tilde{\bm{z}}, \bm{z}) = 1}}
\ell(\bm{x}, \tilde{\bm{z}})\sin\!\left[ m(\bm{x}, \tilde{\bm{z}}) - m(\bm{x}, \bm{z}) \right]
\frac{f(\bm{x}, \tilde{\bm{z}})}{f(\bm{x}, \bm{z})}\\
&+ o(h^2 + \lambda).\\
	\Var[\hat m_{\bm{H}}(\bm{x}, \bm{z})]=&\,\frac{R(K)\sigma^2_1(\bm{x}, \bm{z})}{n h^k  \ell^2(\bm{x},\bm{z}) f(\bm{x}, \bm{z})} + o\left( \frac{1}{n h^k} \right),
			\end{align*}
		where \( \ell(\bm{x}, \bm{z}) = \mathbb{E}[\cos(\varepsilon) \mid \bm{X} = \bm{x}, \bm{Z} = \bm{z}]>0 \) and
\( \sigma_1^2(\bm{x}, \bm{z}) =\Var[\sin(\varepsilon) \mid \bm{X} = \bm{x}, \bm{Z} = \bm{z}] \).
        %$$ \ell(\bm{x}, \bm{z}) = [m_1^2(\bm{x}, \bm{z}) + m_2^2(\bm{x}, \bm{z})]^{1/2}$$ and 
	%\[
%\sigma_1^2(\bm{x}, \bm{z}) = \frac{s_1^2(\bm{x}, \bm{z}) m_2^2(\bm{x}, \bm{z}) + s_2^2(\bm{x}, \bm{z}) m_1^2(\bm{x}, \bm{z}) - 2 c(\bm{x}, \bm{z}) m_1(\bm{x}, \bm{z}) m_2(\bm{x}, \bm{z})}{m_1^2(\bm{x}, \bm{z}) + m_2^2(\bm{x}, \bm{z})}.
%\]
%representing the variance of the angular estimator induced by the propagation of noise through the sine and cosine components. This term combines the variances of \( \sin(\theta) \) and \( \cos(\theta) \), along with their covariance, and is normalized by the square of the local concentration.
\end{theorem}

\begin{remark}
The asymptotic bias in Theorem~\ref{teo1} decomposes into three components: two terms of order $O(h^2)$ associated with smoothing over the continuous covariates and a term of order $O(\lambda)$ induced by smoothing across categorical levels. The $O(h^2)$ contributions coincide with the bias structure established for multivariate kernel regression with circular responses \citep{meilan2021nonparametricTEST}. Specifically, the term involving $\nabla_{\bm x}(\ell f)^\top \nabla_{\bm x} m$ arises from the propagation of the smoothing error of the Cartesian components through the nonlinear $\operatorname{atan2}$ map and reflects the interaction between the local design
quantity $(\ell f)$ and the local variation of the regression function in the continuous covariates (as measured by
$\nabla_{\bm x} m$). The term involving $\operatorname{tr}(\nabla_{\bm x}^2 m)$ corresponds to the standard curvature contribution. The $O(\lambda)$ component is specific to the smoothing of categorical covariates. It aggregates directional discrepancies among categories within Hamming distance one of $\bm z$ and vanishes when no categorical covariates are smoothed. The asymptotic variance is of order $1/(n h^k)$ and depends on the local design density $f(\bm x,\bm z)$ and the effective noise level $\sigma_1^2(\bm x,\bm z)$, scaled by the factor $1/\ell(\bm x,\bm z)^2$. Consequently, larger circular residual dispersion (larger $\sigma_1^2$) or weaker local alignment (smaller $\ell$) increases the variability of the estimator $\hat m_{\bm H}(\bm x,\bm z)$.
%The asymptotic bias in Theorem~\ref{teo1} has three components: two $O(h^2)$ contributions driven by the continuous covariates (the $\nabla_{\bm x}(\ell f)^\top\nabla_{\bm x} m$ and curvature terms) and an $O(\lambda)$ contribution induced by categorical smoothing. The two $O(h^2)$ terms match the bias structure known for multivariate kernel regression with circular responses \citep{meilan2021nonparametricTEST}. In particular, the term involving $\nabla_{\bm x}(\ell f)^\top \nabla_{\bm x} m$ arises from propagating the smoothing error of the Cartesian components through the nonlinear $\operatorname{atan2}$ map and reflects the interaction between the local design $(\ell f)$ and the slope of $m$, whereas the $\operatorname{tr}(\nabla_{\bm x}^2 m)$ term is the usual curvature contribution. The $O(\lambda)$ term is specific to smoothing over categorical levels: it aggregates directional differences between ``neighboring'' categories (in Hamming distance) and vanishes when no categorical covariates are smoothed.
%The variance decays at rate $1/(n h^k)$ and depends on the local design density $f(\bm x,\bm z)$ and on the effective noise level $\sigma_1^2(\bm x,\bm z)$, amplified by the factor $1/\ell(\bm x,\bm z)^2$. Thus, higher circular residual dispersion (large $\sigma_1^2$) or weaker local alignment (small $\ell$) lead to higher uncertainty in $\hat m_{\bm H}(\bm x,\bm z)$.
\end{remark}

The practical performance of $\hat m_{\bm H}(\bm x,\bm z)$ depends critically on the bandwidth vector
$\bm H=(\bm h,\boldsymbol\lambda)$, where $\bm h=(h_1,\ldots,h_k)$ governs smoothing over the continuous predictors
and $\boldsymbol\lambda=(\lambda_1,\ldots,\lambda_p)$ controls smoothing across the categorical factors.
If the bandwidths are too small, the estimator undersmooths and exhibits high variability; if they are too large, it oversmooths and can mask sensory-condition effects. This bias-variance trade-off motivates the use of data-driven bandwidth selectors designed to yield stable, practically reliable choices.
In principle, plug-in bandwidth selectors can be obtained by minimizing asymptotic approximations of the mean squared error based on the leading bias and variance terms in Theorem~\ref{teo1}. This minimization can be performed either locally or globally by integrating the approximation over the covariate space. In our setting, however, the resulting criteria involve several unknown and nonlinearly coupled quantities (e.g., $m$, $\ell$, $f$, and derivatives of $\ell f$ and $m$), so closed-form minimizers are unavailable, and the corresponding numerical optimization would require estimating multiple nuisance functions with sufficient accuracy.
For these reasons, and to keep the methodology practical for real-data analyses, we instead consider three practical selectors: a leave-one-out cross-validation criterion, a bootstrap-based procedure tailored
to the circular cosine loss, and a simple rule-of-thumb selector. Their performance is evaluated in Sections~\ref{sec:simulations} and~\ref{sec:results}.

\section{Bandwidth selection}
\label{sec:bw}

In this section, we describe the cross-validation, bootstrap, and rule-of-thumb selectors in detail.

\subsection{Cross-validation}
A common approach for selecting the bandwidth vector is to minimize the leave-one-out circular prediction loss
\[
\operatorname{CV}(\bm{H}) = \frac{1}{n} \sum_{i=1}^n \left\{1 - \cos\left[ \Theta_i - \hat{m}^{(-i)}_{\bm{H}}(\bm{X}_i, \bm{Z}_i) \right] \right\},
\]
where $\hat{m}^{(-i)}_{\bm{H}}(\bm{X}_i, \bm{Z}_i)$ denotes the estimator computed without the $i$th observation. This criterion is a circular analogue of the classical leave-one-out mean squared error and aims to approximate the expected cosine prediction loss for a new observation.
We denote by $\bm{H}_{\text{CV}}$ a minimizer of $\operatorname{CV}(\bm{H})$.

\subsection{Bootstrap-based bandwidth selection}
\label{sec:bw_boot}

We propose a residual bootstrap bandwidth selector tailored to circular regression with mixed covariates.
As in the rest of the paper, discrepancy is measured using the usual circular cosine loss \citep[see][]{kim2017multivariate}.
Our goal is to choose the bandwidth vector $\bm H$ to perform well in a \emph{global} sense, averaging this loss over the covariate distribution.
Since that distribution is unknown, we approximate the corresponding risk by the empirical average over the observed design points $(\bm X_i,\bm Z_i)$.
Algorithm~\ref{alg_3} summarizes the steps.

\begin{algorithm}[!htb]
\caption{Bootstrap-based bandwidth selection}
\label{alg_3}
\begin{algorithmic}
\State \textbf{Step 1.} Compute circular residuals using $\bm H_0$:
\[
\hat\varepsilon_i=\big[\Theta_i-\hat m_{\bm H_0}(\bm X_i,\bm Z_i)\big]\ (\operatorname{mod}\, 2\pi),
\]
and center them:
\[
\tilde\varepsilon_i=\big[\hat\varepsilon_i-\bar\varepsilon\big]\ (\operatorname{mod}\, 2\pi), \quad 
\bar\varepsilon=\operatorname{atan2}\!\left(\frac{1}{n}\sum_{i=1}^n \sin(\hat\varepsilon_i),\,
\frac{1}{n}\sum_{i=1}^n \cos(\hat\varepsilon_i)\right)
\]

\State \textbf{Step 2.} For $b=1,\ldots,B$:
\Statex \hspace{1.5em} (i) Sample $\tilde\varepsilon_1^{*(b)},\ldots,\tilde\varepsilon_n^{*(b)}$ with replacement from
$\{\tilde\varepsilon_1,\ldots,\tilde\varepsilon_n\}$ and generate pseudo-responses
\[
\Theta_i^{*(b)}=\big[\hat m_{\bm H_1}(\bm X_i,\bm Z_i)+\tilde\varepsilon_i^{*(b)}\big]\ (\operatorname{mod}\, 2\pi).
\]
\Statex \hspace{1.5em} (ii) Compute $\hat m_{\bm H}^{*(b)}$ from $\{(\bm X_i,\bm Z_i,\Theta_i^{*(b)})\}_{i=1}^n$
using the candidate bandwidth $\bm H$.

\State \textbf{Step 3.} Choose $\bm H_{\text{boot}}$ as a minimizer of
\begin{equation}
  \frac{1}{B} \sum_{b=1}^B\bigg(\frac{1}{n} \sum_{i=1}^n \left\{ 1 - \cos\left[  \hat{m}_{\bm{H}_1}(\bm{X}_i, \bm{Z}_i) -\hat{m}^{*(b)}_{\bm{H}}(\bm{X}_i, \bm{Z}_i)\right] \right\}\bigg).
\label{eq:boot}
\end{equation}
\end{algorithmic}
\end{algorithm}

The procedure uses two pilot bandwidth vectors, \(\bm{H}_{0} = (\bm{h}_{0} \boldsymbol{\lambda}_{0})\) for residual computation and \(\bm{H}_{1} =(\bm{h}_{1}, \boldsymbol{\lambda}_{1})\) for constructing a smooth reference fit that plays the role of the unknown regression function when generating pseudo-responses.

Centering the residuals via $\operatorname{atan2}$ yields resampled perturbations that are interpretable
signed angular deviations (effectively in $(-\pi,\pi]$), consistent with the signed error convention adopted throughout.
The criterion~\eqref{eq:boot} approximates the global cosine-loss prediction risk by replacing the unknown regression function
with the pilot fit $\hat m_{\bm H_1}$ and averaging the cosine loss over the observed design points.

In all numerical studies in this paper we set $\bm H_0=\bm H_1=\bm H_{\text{CV}}$, yielding a fully data-driven
bootstrap selector without introducing additional tuning parameters.

%This bootstrap criterion mimics the circular prediction risk and is particularly well-suited for circular regression with mixed-type predictors. 
%To the best of our knowledge, this constitutes the first bootstrap bandwidth selection method developed for this context. 
%The corresponding bootstrap-based bandwidth is denoted by \( \bm{H}_{\text{boot}} \).

\subsection{Rule-of-thumb}
\label{sec:bw_rot}

In addition to cross-validation and bootstrap-based selection, we consider a computationally efficient
rule-of-thumb (RoT) choice for the bandwidth vector, in the spirit of classical plug-in heuristics for
kernel regression \citep[e.g.,][]{FanGijbels1996,Hardle1990}. The aim is to provide a simple, fully
explicit default that depends only on the sample size and marginal covariate scales, and that is
straightforward to implement and to reproduce.

For each continuous covariate $X_j$, we set the marginal RoT bandwidth to
\[
h_{j,\text{RoT}} = c_h\,\hat{\sigma}_j\, n^{-1/5},
\]
where $\hat{\sigma}_j$ is the sample standard deviation of $X_j$ and $c_h>0$ is a scaling constant.
The $n^{-1/5}$ factor matches the classical one-dimensional normal-reference scaling and provides a
stable marginal smoothing level. Importantly, this construction is intentionally \emph{modular}, that is,  each
$h_{j,\text{RoT}}$ depends only on the marginal variability of $X_j$ and on $n$, so the same value is
used whether $X_j$ appears alone or together with other predictors. While such marginal calibration
does not explicitly adapt to interactions among covariates, it provides a transparent default that is
easy to interpret and implement. In multivariate settings, this RoT choice is therefore not intended
to be asymptotically optimal, but rather to offer a robust baseline with a clear coordinate-wise meaning.

For each categorical covariate $Z_l$, we use the analogous heuristic
\[
\lambda_{l,\text{RoT}} = c_\lambda\,L_l^{-\gamma}\, n^{-1/5},
\]
where $L_l$ is the number of levels of $Z_l$, $c_\lambda>0$ is a scaling constant, and $\gamma>0$
controls how the amount of categorical smoothing decreases as the number of categories increases.
This scaling is adopted as a pragmatic convention that yields a transparent default across experiments.
While the asymptotic expansion in Theorem~\ref{teo1} indicates that the effect of the categorical smoothing is captured
through an $O(\lambda)$ contribution, the finite-sample behavior of categorical smoothing depends strongly on
the chosen categorical kernel and on the distribution and balance of the levels of $Z_l$. For this
reason, we do not pursue a theoretically optimized rate for $\lambda_l$ within the RoT family.
Instead, we use the above simple scaling as a stable baseline, and rely on the data-driven selectors
(CV and bootstrap) when a more adaptive choice is desired. We denote the resulting bandwidth vector by
$\bm H_{\text{RoT}}$.

In all numerical studies, we fix $\gamma=1$ and use the default NW RoT constants $c_h=1.06$ and $c_\lambda=1.0$, so that
$\bm H_{\text{RoT}}$ is fully specified and directly reproducible. The value $c_h=1.06$ corresponds to the usual normal-reference
constant for one-dimensional Gaussian-kernel smoothing, while the categorical component uses the simple scaling
$\lambda_{l,\text{RoT}} \propto L_l^{-1} n^{-1/5}$. We keep these constants fixed throughout
Sections~\ref{sec:simulations}--\ref{sec:results} for comparability; performance-sensitive
choices are handled by the CV and bootstrap selectors. 

\section{Bootstrap algorithm for simultaneous confidence bands}
\label{sec:bands}

To assess uncertainty in the estimated regression function at a global scale, we propose a bootstrap procedure
to construct simultaneous confidence bands for the circular regression function with target coverage
$1-\alpha$. We focus on the setting with one continuous predictor $X$ and one categorical factor $Z$, for which the circular regression function is denoted by $m(x,z)$. The bands
are reported on a grid of predictor values $x_1,\ldots,x_G$ and separately for each level $z$ of $Z$, while
respecting the circular scale.

Unlike classical approaches based on pointwise intervals or conservative Bonferroni corrections, our
method uses a data-driven calibration to balance coverage and tightness in the circular setting. The
algorithm has three stages: (i) construction of pointwise bootstrap bands; (ii) computation of a
Bonferroni baseline; and (iii) iterative calibration of the simultaneous level.
 A concise summary is
given in Algorithm~\ref{alg_3_concise}.

\begin{algorithm}[!htb]
\caption{Bootstrap for simultaneous confidence bands}
\label{alg_3_concise}
\begin{algorithmic}
\State \textbf{Step 1.}
Fix a grid \(x_1,\ldots,x_G\) and compute \(\hat m_{\bm H}(x_j,z)\) for each level \(z\) using \eqref{eq:est}.

\State \textbf{Step 2.}
Compute centered residuals \(\tilde{\varepsilon}_i\) (as in Algorithm~\ref{alg_3}). For \(b=1,\ldots,B\):
\Statex \hspace{1.5em} (i) Generate pseudo-responses
\[
\Theta_i^{*(b)}=\big[\hat m_{\bm H}(X_i,Z_i)+\tilde\varepsilon_i^{*(b)}\big]\ (\operatorname{mod}\, 2\pi),
\]
where \(\tilde\varepsilon_i^{*(b)}\) are sampled with replacement from
\(\{\tilde\varepsilon_1,\ldots,\tilde\varepsilon_n\}\).
\Statex \hspace{1.5em} (ii) Fit \(\hat m_{\bm H}^{*(b)}(x_j,z)\) on the pseudo-sample using the bandwidth \(\bm H\).

\State \textbf{Step 3.}
For \(j=1,\ldots,G\) and \(b=1,\ldots,B\), compute
\[
\Delta_j^{(b)}=
\big[\hat m_{\bm H}^{*(b)}(x_j,z)-\hat m_{\bm H}(x_j,z)+\pi\big]\ (\operatorname{mod}\, 2\pi)-\pi,
\]
so that \(\Delta_j^{(b)}\in(-\pi,\pi]\).

\State \textbf{Step 4.}
Set \(\alpha_{\text{low}}=\alpha/G\), \(\alpha_{\text{high}}=\alpha\), and choose a tolerance \(\delta>0\).

\State \textbf{Step 5.}
For any \(\alpha'\in[\alpha_{\text{low}},\alpha_{\text{high}}]\), let
\(q_{j,\alpha'/2}\) and \(q_{j,1-\alpha'/2}\) be the empirical quantiles of
\(\{\Delta_j^{(b)}\}_{b=1}^B\), and define
\[
\widehat P_{\text{in}}(\alpha')
=\frac{1}{B}\sum_{b=1}^B
\mathbb{I}\!\left\{
\Delta_j^{(b)}\in\big[q_{j,\alpha'/2},\,q_{j,1-\alpha'/2}\big]
\ \ \forall\, j=1,\ldots,G
\right\}.
\]
\Statex \hspace{1.5em} (i) Compute \(\widehat P_{\text{in}}(\alpha_{\text{low}})\) and \(\widehat P_{\text{in}}(\alpha_{\text{high}})\).
\Statex \hspace{1.5em} (ii) \emph{If} \(\widehat P_{\text{in}}(\alpha_{\text{low}}) < 1-\alpha\), \emph{stop} and set \(\alpha'=\alpha_{\text{low}}\).
\Statex \hspace{1.5em} (iii) \emph{If} \(\widehat P_{\text{in}}(\alpha_{\text{high}}) \ge 1-\alpha\), \emph{stop} and set \(\alpha'=\alpha_{\text{high}}\).
\Statex \hspace{1.5em} (iv) \emph{Otherwise}, iterate until
$|\widehat P_{\text{in}}(\alpha')-(1-\alpha)|<\delta$
(or a maximum number of iterations is reached):
\Statex \hspace{3.0em} 1. Set $\alpha'=(\alpha_{\text{low}}+\alpha_{\text{high}})/2$.
\Statex \hspace{3.0em} 2. Compute $\widehat P_{\text{in}}(\alpha')$.
\Statex \hspace{3.0em} 3. If $\widehat P_{\text{in}}(\alpha') \ge 1-\alpha$, set $\alpha_{\text{low}}=\alpha'$;
\Statex \hspace{3.6em} otherwise set $\alpha_{\text{high}}=\alpha'$.

\State \textbf{Step 6.}
With the final \(\alpha'\), compute \(q_{j,\alpha'/2}\) and \(q_{j,1-\alpha'/2}\) for all \(j\) and define
\[
\mathsf{Band}_j(z;\alpha')=
\Big[\hat m_{\bm H}(x_j,z)+q_{j,\alpha'/2},\ \hat m_{\bm H}(x_j,z)+q_{j,1-\alpha'/2}\Big]\ (\operatorname{mod}\, 2\pi).
\]
\end{algorithmic}
\end{algorithm}

Step~5 in Algorithm~\ref{alg_3_concise} already covers two boundary cases that avoid unnecessary calibration:
(i) if the Bonferroni choice \(\alpha'=\alpha/G\) attains the target simultaneous coverage, the algorithm stops and returns that band;
(ii) if even the narrowest candidate in the search range (\(\alpha'=\alpha\)) attains the target, the algorithm also stops and returns \(\alpha'=\alpha\).
In practice, we recommend reporting the final band either as a ribbon on an unwrapped scale or as paired lower/upper angular bounds at each grid point, to avoid misleading visualizations associated with the periodicity of the circular scale.

Algorithm~\ref{alg_3_concise} is presented for the case of one continuous predictor $X$ and one categorical factor
$Z$. Simultaneity is enforced over the evaluation grid $x_1,\ldots,x_G$, and calibration is performed separately
for each level $z$ of $Z$. Extensions to multiple continuous predictors and/or to joint simultaneity across the
levels of $Z$ can be developed, but would require a different construction (e.g., calibration over
multivariate grids together with explicit multiplicity control across categories) and are beyond the scope of
this work.

\section{Simulation study}
\label{sec:simulations}

To assess the finite-sample performance of the bandwidth selectors introduced in
Section~\ref{sec:bw}, we conducted a simulation study with synthetic circular data for which the
regression function is known. We compared the cross-validation, bootstrap, and rule-of-thumb selectors across a range of scenarios using a two-phase design that separates the
construction of a Monte Carlo risk benchmark from the evaluation of the competing selectors.

In Phase~I, we generated $N_1=100$ independent samples of size $n$. Following the circular average
squared error (CASE) criterion commonly used to assess circular regression fits
\citep{meilan2021trendSURF,kim2017multivariate}, defined as
\begin{equation}
\operatorname{CASE}(\bm H)
=\frac{1}{n}\sum_{i=1}^{n}\left\{1-\cos\!\Big[m(\bm X_i,\bm Z_i)
-\hat m_{\bm H}(\bm X_i,\bm Z_i)\Big]\right\},
\label{eq:CASE}
\end{equation}
for a candidate bandwidth $\bm H$ and the $j$th simulated sample
$\{(\bm X_i^{(j)},\bm Z_i^{(j)},\Theta_i^{(j)})\}_{i=1}^n$, we computed
$\operatorname{CASE}^{(j)}(\bm H)$ by evaluating \eqref{eq:CASE} on that sample, $j=1,\ldots,N_1$.
We then averaged over the $N_1$ replications to obtain the Monte Carlo risk surface
\begin{equation}
\overline{\operatorname{CASE}}(\bm H)
=\frac{1}{N_1}\sum_{j=1}^{N_1}\operatorname{CASE}^{(j)}(\bm H),
\label{eq:mCASE}
\end{equation}
which provides a benchmark approximation to the global circular prediction risk of the estimates
under the data-generating design. Finally, we define $\bm H_{\overline{\operatorname{CASE}}}$ as the minimizer of
\eqref{eq:mCASE} over the considered bandwidth grid, and use it as a Monte Carlo benchmark
in the comparisons reported below.
    
In Phase~II, we generated \(N_2=500\) additional independent samples (not used in Phase~I). For each evaluation
sample \(j=1,\ldots,N_2\), we computed the data-driven selectors \(\bm H_{\text{CV}}^{(j)}, \bm H_{\text{boot}}^{(j)}\), and \(\bm H_{\text{RoT}}^{(j)}\). To assess their performance, we scored each selected bandwidth using the Phase~I risk surface, that is, we recorded
\(\overline{\operatorname{CASE}}\!\big(\bm H_{\text{method}}^{(j)}\big)\) for each method.
This two-phase separation avoids optimistic evaluation due to data reuse.

Each sample consisted of \(n\) i.i.d.\ observations \((X_i,Z_i,\Theta_i)\), where
\(X\sim\mathcal{U}(0,1)\), \(Z\in\{A,B,C\}\) was sampled uniformly, and
\(\Theta=[m(X,Z)+\varepsilon]\;(\operatorname{mod}\, 2\pi)\), with \(\varepsilon\sim\text{von Mises}(0,\kappa)\).
We considered \(n\in\{100,200,500\}\) and two concentration levels: \(\kappa=3\) (moderate noise) and
\(\kappa=10\) (low noise). Two regression functions were used (see also Figure~\ref{fig:regs}):

\noindent
\begin{minipage}[t]{0.48\textwidth}
\[
\text{R1.\ } m(x, z) = 
\begin{cases} 
2\pi x, & z = A, \\ 
2\pi(1 - x), & z = B, \\ 
\pi, & z = C
\end{cases}
\]
\end{minipage}\hfill
\begin{minipage}[t]{0.48\textwidth}
\[
\text{R2.\ } m(x, z) = 
\begin{cases} 
\pi x^2, & z = A, \\ 
\pi(1 - x^2), & z = B, \\ 
\frac{3\pi}{2}|\sin(\pi x)|, & z = C
\end{cases}
\]
\end{minipage}

\begin{figure}[!htb]
\centering
\begin{minipage}{0.48\textwidth}
\centering
\includegraphics[width=\linewidth]{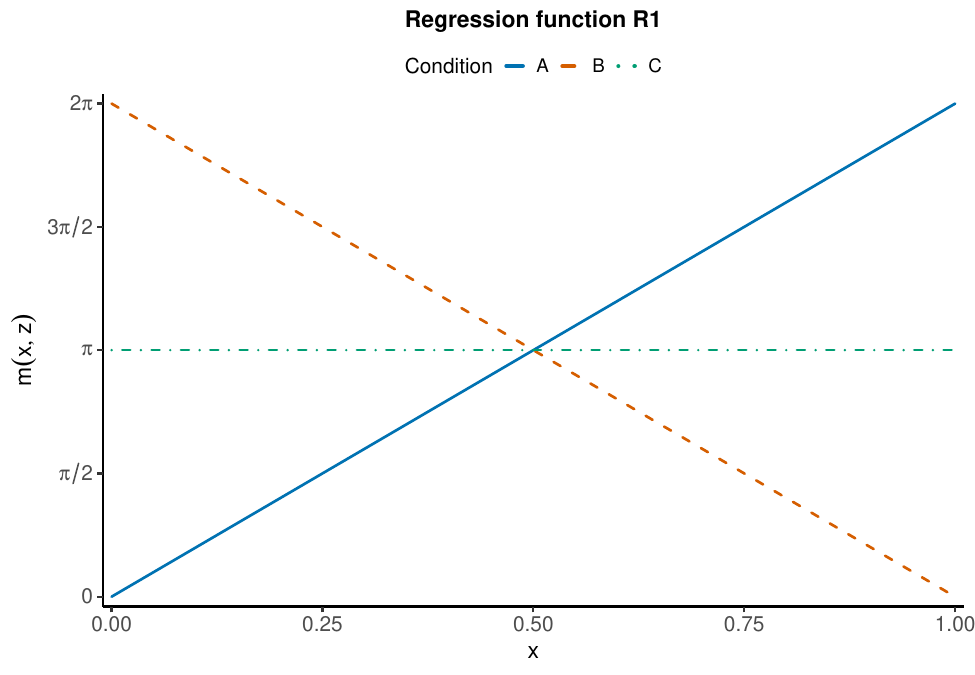}
\end{minipage}
\hfill
\begin{minipage}{0.48\textwidth}
\centering
\includegraphics[width=\linewidth]{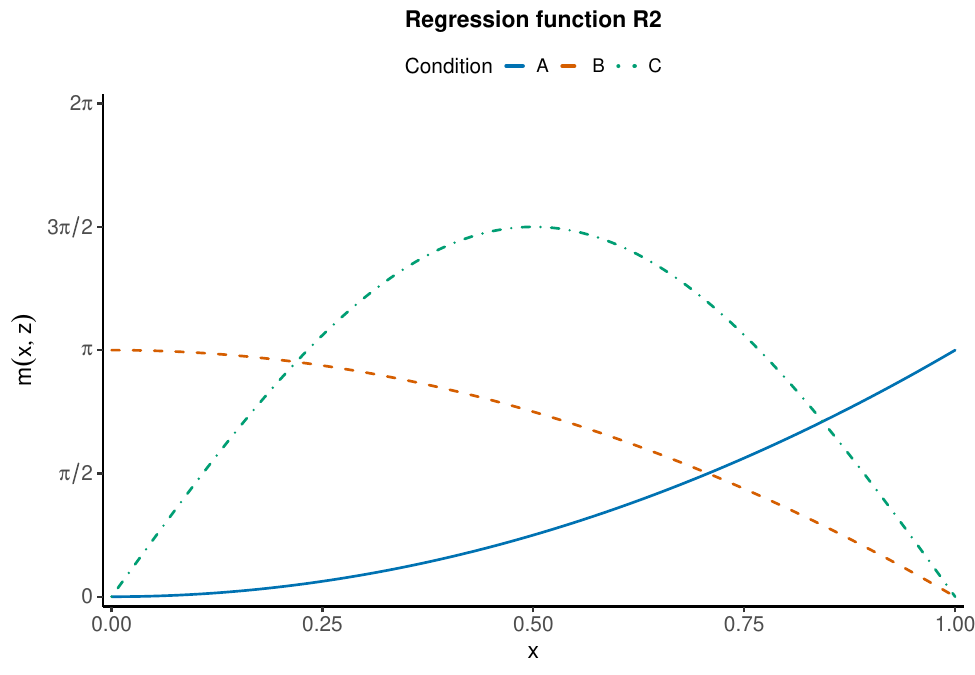}
\end{minipage}
\caption{Circular regression functions used in the simulation study: the left panel corresponds to R1 and the right panel to R2. Each line corresponds to a different value of the categorical variable \( Z \in \{A, B, C\} \).}
\label{fig:regs}
\end{figure}

The circular regression estimator given in equation \eqref{eq:est} was computed using product weights
\[
w^i_{\bm{H}}({x}, z) = \frac{K[(x - X_{i})/h]  L(z, Z_i;\lambda)}{\sum_{j=1}^n K[(x - X_{j})/h]   L(z, Z_j;\lambda)},
\]
where $\bm{H}=(h,\lambda)$. We used the Gaussian kernel \(K(u) = (2\pi)^{-1/2}\exp(-u^2/2)\) for the continuous predictor,
and the Aitchison--Aitken kernel for the categorical predictor,
\begin{equation}
L(z, z_i; \lambda) =
\begin{cases}
1 - \lambda, & \text{if } z = z_i, \\
\displaystyle\frac{\lambda}{c - 1}, & \text{if } z \neq z_i,
\end{cases}
\label{eq:AA}
\end{equation}
where \( c \) is the number of categories and \( \lambda \in [0,1] \) controls the weight assigned to mismatches.

We report two complementary summaries. First, estimation accuracy is measured by the Phase~I benchmark risk
\(\overline{\operatorname{CASE}}\!\big(\bm H_{\text{method}}^{(j)}\big)\), summarized over the \(N_2\) evaluation replications.
Second, bandwidth efficiency is assessed via the distribution of the ratios
\(\|\bm H_{\text{method}}^{(j)}\|/\|\bm H_{\overline{\operatorname{CASE}}}\|\), where \(\|\cdot\|\) denotes the Euclidean norm.

Tables~\ref{tab:sim_results_r1} and~\ref{tab:sim_results_r2} report the mean and variance (in parentheses) of
\(\overline{\operatorname{CASE}}\!\big(\bm H_{\text{method}}^{(j)}\big)\) over \(j=1,\ldots,N_2\), for R1 and R2, respectively.
Overall, the bootstrap selector yielded the smallest average error across scenarios, with the most noticeable gains
for smaller \(n\) and for the more complex regression function R2. For example, in the most demanding configuration
(\(n=100\), \(\kappa=3\)) for R2, the bootstrap selector reduced the mean benchmark risk from 0.135 (CV) and 0.140 (RoT)
to 0.128, corresponding to reductions of about 5\% and 9\%, respectively.

% Table placeholders
\begin{table}[!htb]
\centering
\caption{
Mean and variance (in parentheses) of $\overline{\operatorname{CASE}}$ given in \eqref{eq:mCASE} for the regression function R1 and the bandwidth selectors. Each entry is computed over \( N_2 = 500 \) evaluation samples of size \( n \) and different \( \kappa \).
}
\label{tab:sim_results_r1}
\begin{tabular}{lccccc}
\toprule
$\kappa$ & $n$ & \( \bm{H}_{\text{CV}}\) & $\bm{H}_{\text{boot}}$& $\bm{H}_{\text{RoT}}$  & \( \bm{H}_{\overline{\operatorname{CASE}}} \) \\
\midrule
3  & 100 & 0.091 (0.0082) & 0.087 (0.0075) & 0.094 (0.0090) & 0.081 (0) \\
   & 200 & 0.071 (0.0044) & 0.068 (0.0039) & 0.074 (0.0051) & 0.063 (0) \\
   & 500 & 0.046 (0.0015) & 0.044 (0.0012) & 0.049 (0.0018) & 0.041 (0) \\
10 & 100 & 0.059 (0.0059) & 0.055 (0.0050) & 0.061 (0.0068) & 0.050 (0) \\
   & 200 & 0.038 (0.0022) & 0.035 (0.0018) & 0.041 (0.0027) & 0.031 (0) \\
   & 500 & 0.022 (0.0006) & 0.020 (0.0005) & 0.024 (0.0008) & 0.018 (0) \\
\bottomrule
\end{tabular}
%\caption*{\footnotesize\textit{Note.} Variance values in parentheses represent the variability of $\overline{\operatorname{CMISE}}$ across the \( N_2 = 500 \) replicated samples. For the oracle selector \( \bm{H}_{\overline{\operatorname{CASE}}} \), a fixed bandwidth is used in each scenario, leading to zero variance.}
\end{table}

\begin{table}[!htb]
\centering
\caption{
Mean and variance (in parentheses) of $\overline{\operatorname{CASE}}$ given in \eqref{eq:mCASE} for the regression function R2 and the bandwidth selectors. Each entry is computed over \( N_2 = 500 \) evaluation samples of size \( n \) and different \( \kappa \).
}
\label{tab:sim_results_r2}
\begin{tabular}{lccccc}
\toprule
$\kappa$ & $n$ & \( \bm{H}_{\text{CV}}\) & $\bm{H}_{\text{boot}}$& $\bm{H}_{\text{RoT}}$  & \( \bm{H}_{\overline{\operatorname{CASE}}} \) \\
\midrule
3  & 100 & 0.135 (0.0111) & 0.128 (0.0104) & 0.140 (0.0126) & 0.116 (0) \\
   & 200 & 0.106 (0.0063) & 0.099 (0.0056) & 0.112 (0.0074) & 0.089 (0) \\
   & 500 & 0.071 (0.0027) & 0.066 (0.0022) & 0.075 (0.0031) & 0.059 (0) \\
10 & 100 & 0.094 (0.0075) & 0.088 (0.0066) & 0.096 (0.0089) & 0.076 (0) \\
   & 200 & 0.064 (0.0031) & 0.058 (0.0025) & 0.066 (0.0037) & 0.049 (0) \\
   & 500 & 0.036 (0.0009) & 0.033 (0.0007) & 0.038 (0.0011) & 0.028 (0) \\
\bottomrule
\end{tabular}
%\caption*{\footnotesize\textit{Note.} Variance values in parentheses represent the variability of $\overline{\operatorname{CMISE}}$ across the \( N_2 = 500 \) replicated samples. For the oracle selector \( \bm{H}_{\overline{\operatorname{CASE}}} \), a fixed bandwidth is used in each scenario, leading to zero variance.}
\end{table}

Figures~\ref{fig:panel_mtrue1} and~\ref{fig:panel_mtrue2} complement the results by showing the distribution of
\(\|\bm{H}_{\text{method}}^{(j)}\| / \|\bm{H}_{\overline{\operatorname{CASE}}}\|\) across scenarios for R1 and R2, respectively.
The CV selector exhibited more dispersion, particularly for small \(n\), reflecting its sensitivity to sample variability.
In contrast, the bootstrap method was more tightly concentrated around $1$, indicating a favorable accuracy--stability trade-off.
The RoT selector, while computationally cheap, tended to oversmooth, as evidenced by median ratios above $1$ in all scenarios.

\begin{figure}[!htb]
  \centering
  \begin{subfigure}{0.32\textwidth}
    \includegraphics[width=\linewidth]{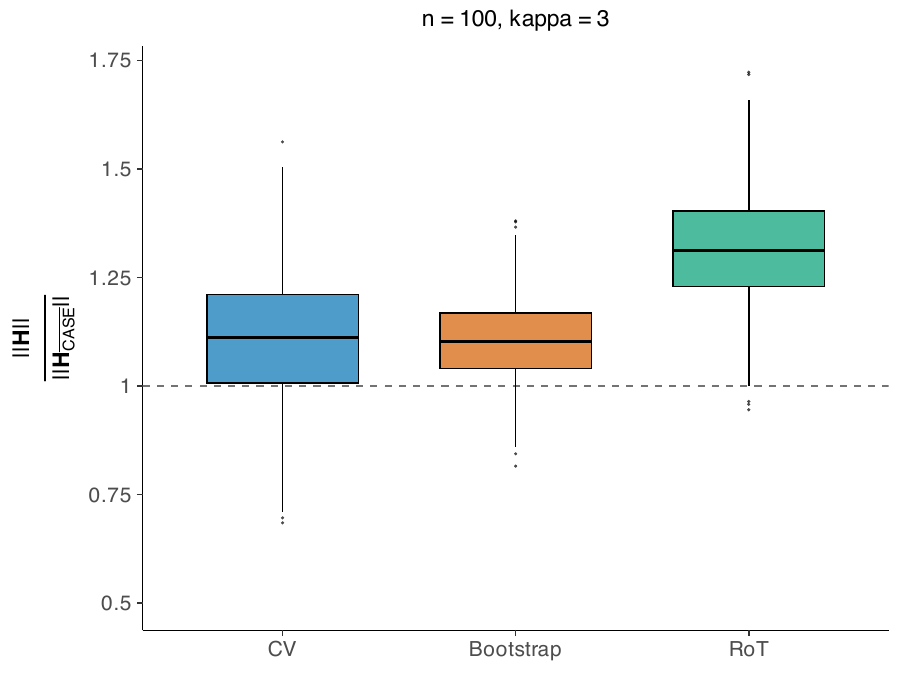}
  \end{subfigure}
  \hfill
  \begin{subfigure}{0.32\textwidth}
    \includegraphics[width=\linewidth]{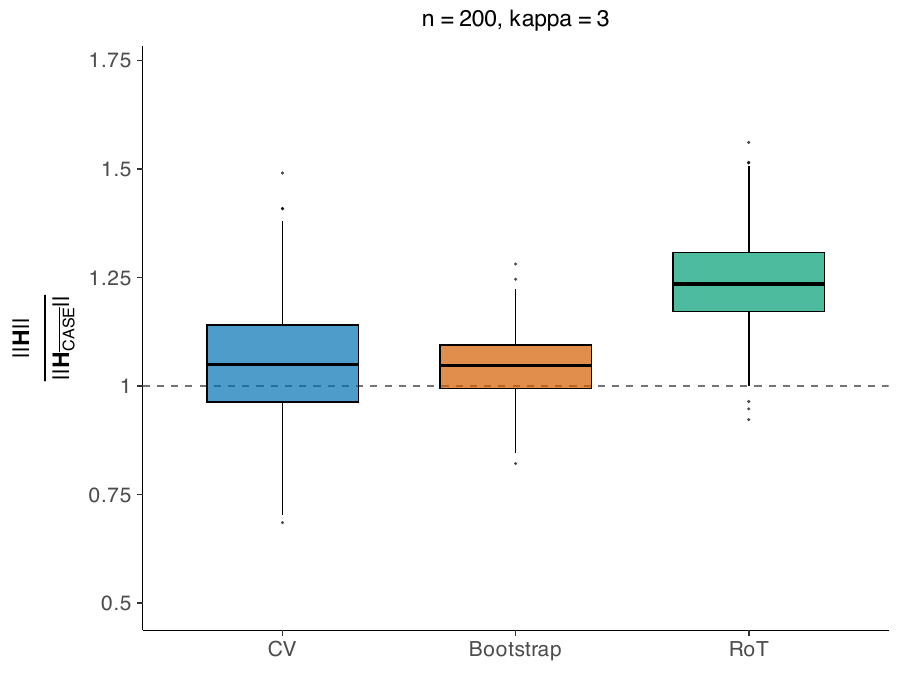}
     \end{subfigure}
  \hfill
  \begin{subfigure}{0.32\textwidth}
    \includegraphics[width=\linewidth]{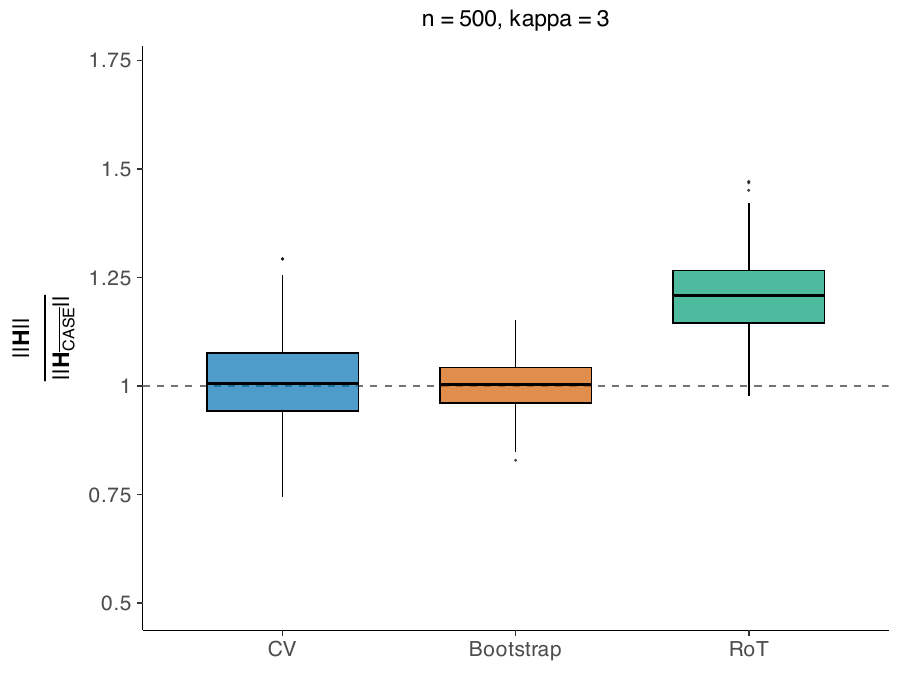}
  \end{subfigure}

  \vspace{1em}

  \begin{subfigure}{0.32\textwidth}
    \includegraphics[width=\linewidth]{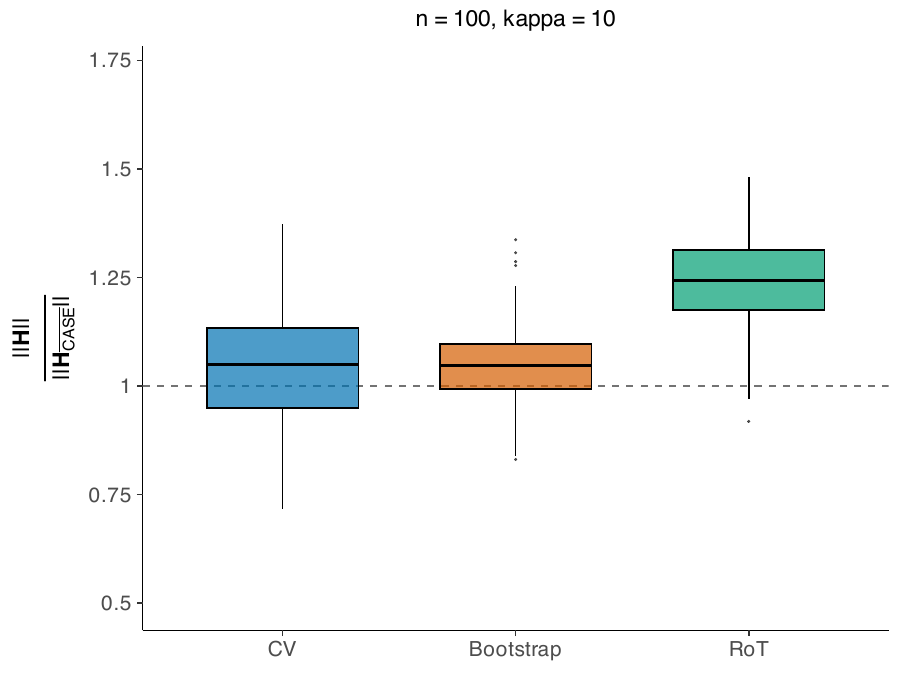}
  \end{subfigure}
  \hfill
  \begin{subfigure}{0.32\textwidth}
    \includegraphics[width=\linewidth]{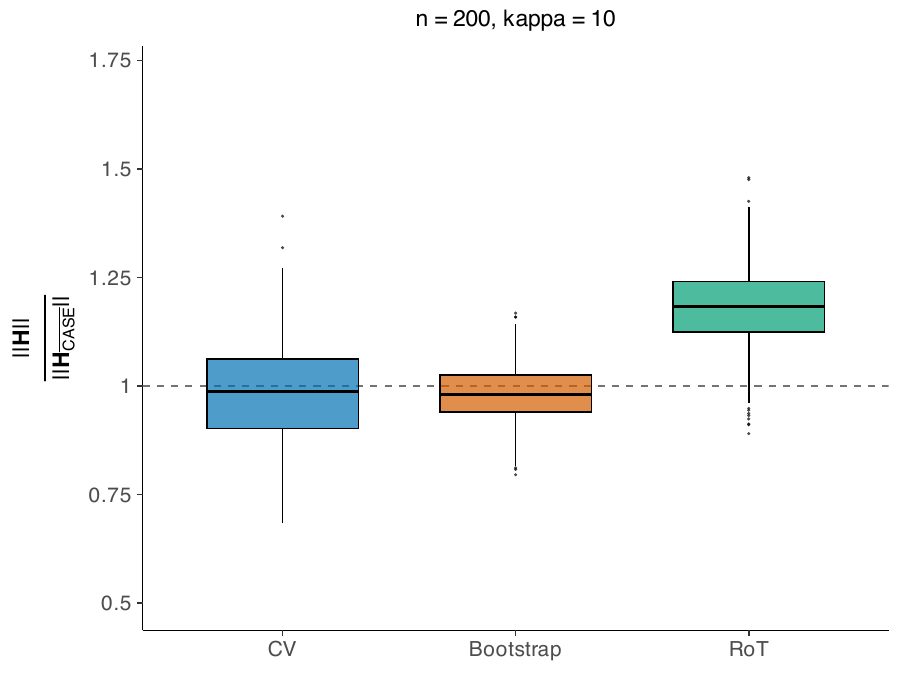}
  \end{subfigure}
  \hfill
  \begin{subfigure}{0.32\textwidth}
    \includegraphics[width=\linewidth]{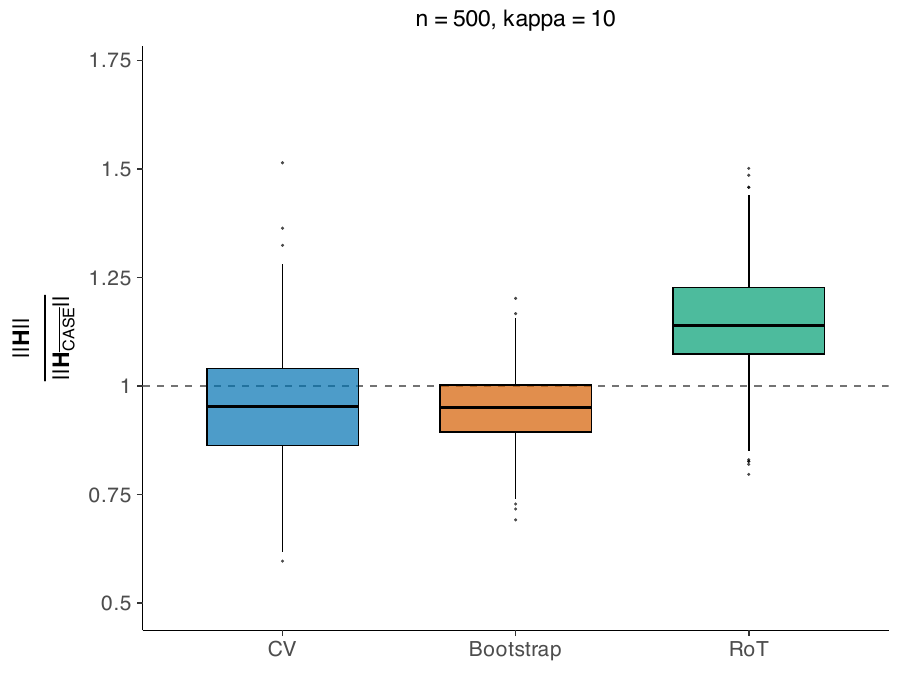}
  \end{subfigure}

  \caption{Distribution of $\|\bm{H}_{\text{method}}\| / \|\bm{H}_{\overline{\operatorname{CASE}}}\|$ for regression function R1.}
  \label{fig:panel_mtrue1}
\end{figure}

\begin{figure}[!htb]
  \centering
  \begin{subfigure}{0.32\textwidth}
    \includegraphics[width=\linewidth]{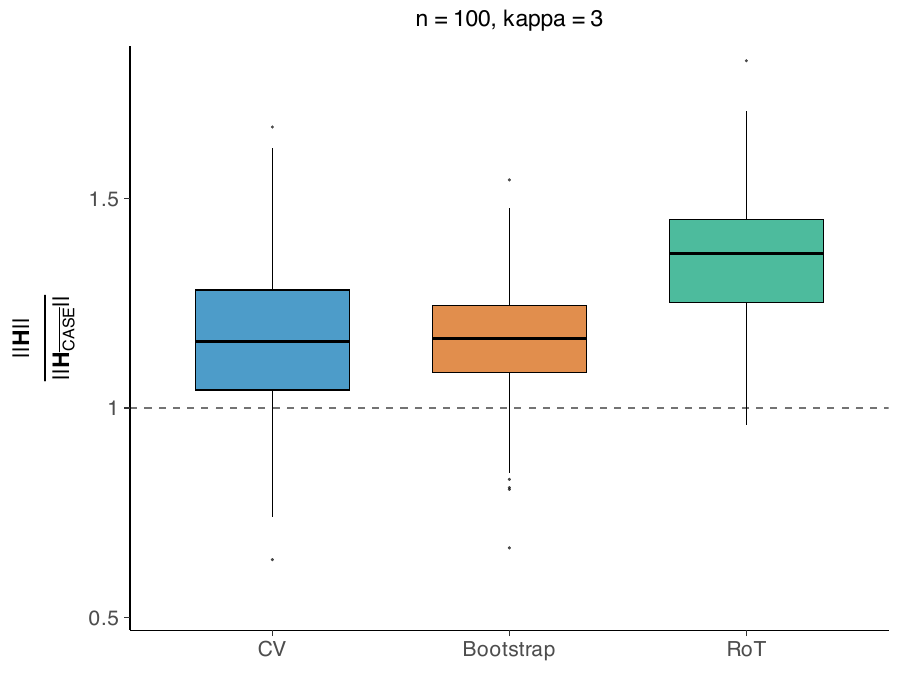}
  \end{subfigure}
  \hfill
  \begin{subfigure}{0.32\textwidth}
    \includegraphics[width=\linewidth]{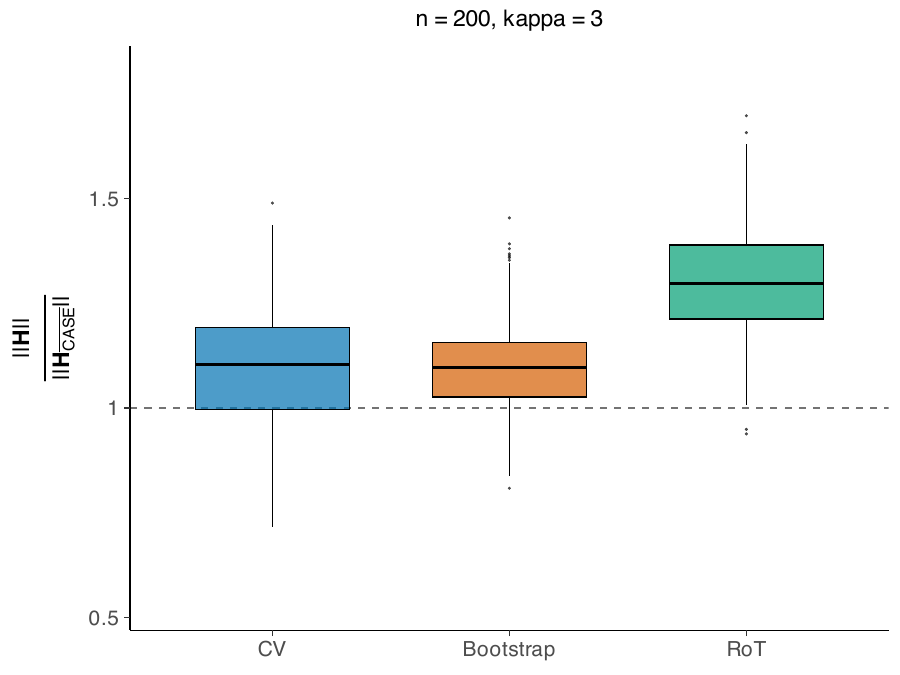}
  \end{subfigure}
  \hfill
  \begin{subfigure}{0.32\textwidth}
    \includegraphics[width=\linewidth]{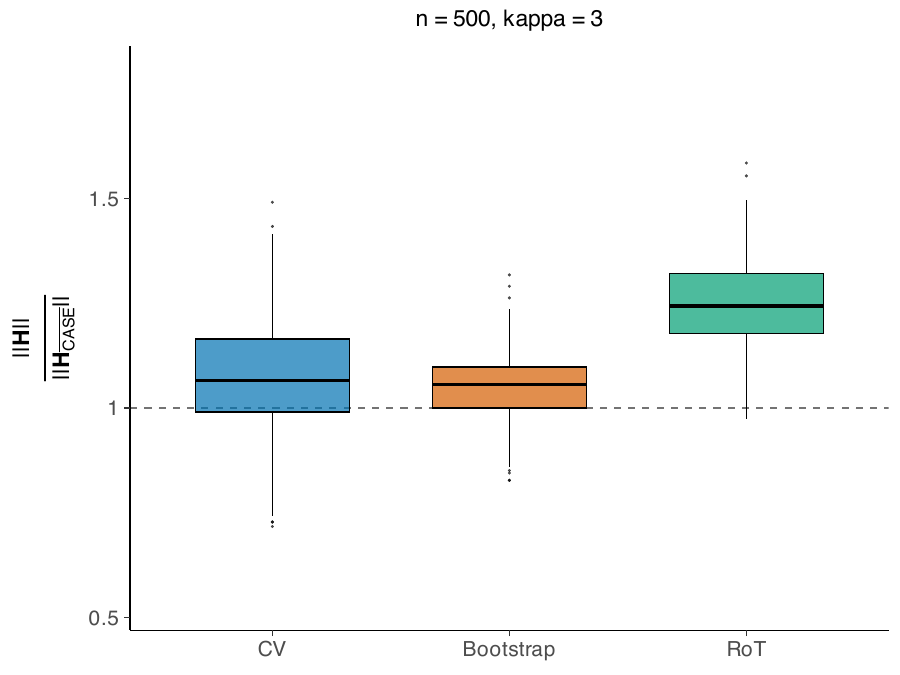}
  \end{subfigure}

  \vspace{1em}

  \begin{subfigure}{0.32\textwidth}
    \includegraphics[width=\linewidth]{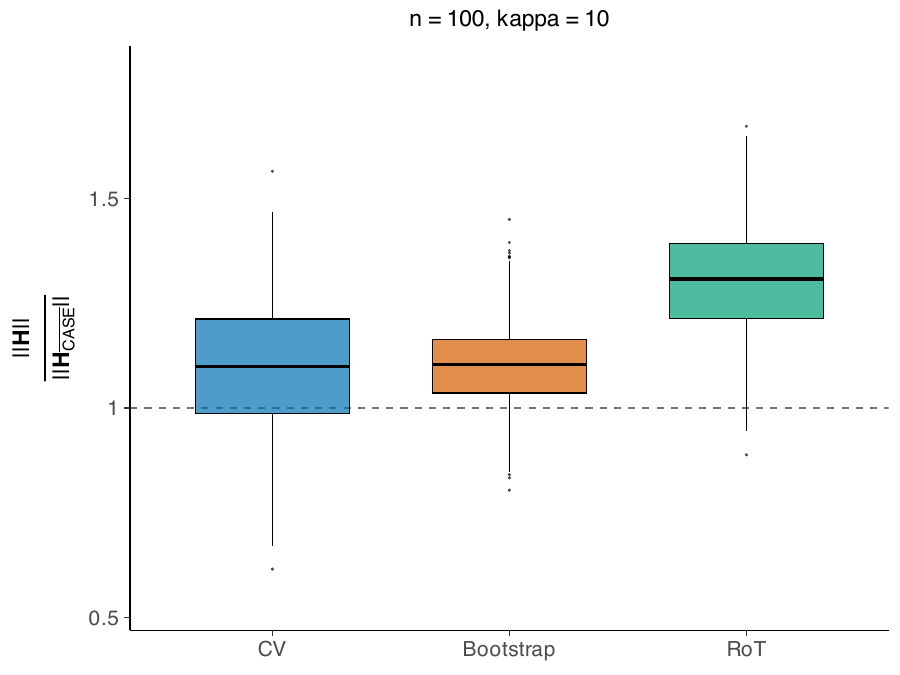}
  \end{subfigure}
  \hfill
  \begin{subfigure}{0.32\textwidth}
    \includegraphics[width=\linewidth]{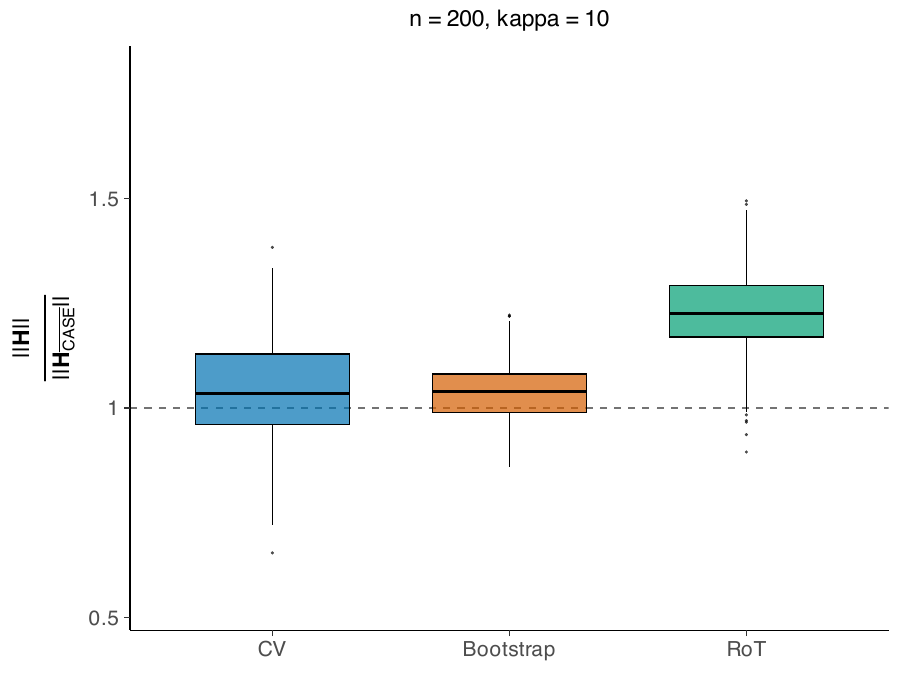}
  \end{subfigure}
  \hfill
  \begin{subfigure}{0.32\textwidth}
    \includegraphics[width=\linewidth]{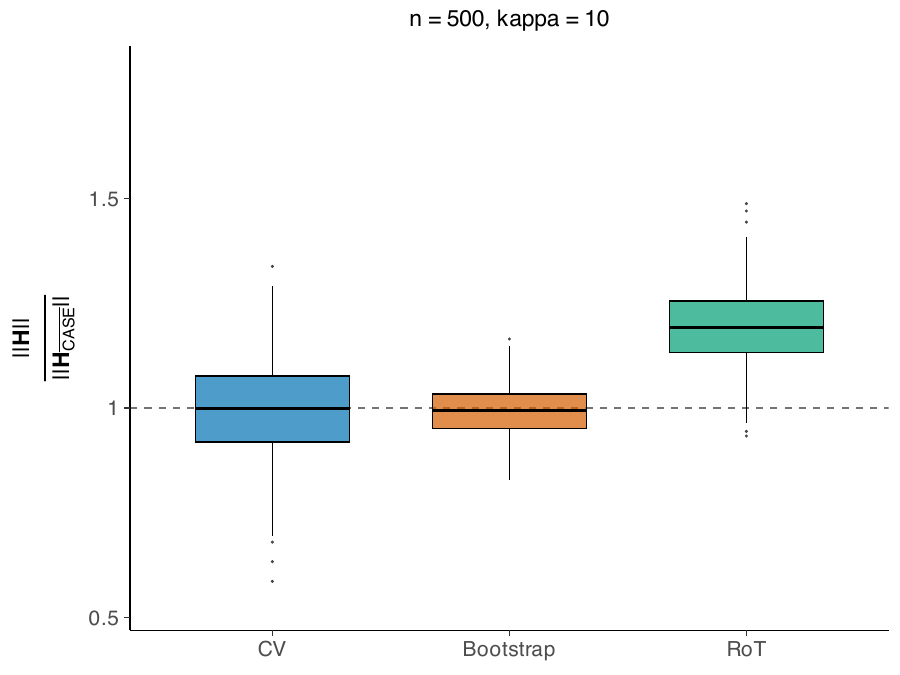}
  \end{subfigure}

  \caption{Distribution of $\|\bm{H}_{\text{method}}\| / \|\bm{H}_{\overline{\operatorname{CASE}}}\|$ for regression function R2.}
  \label{fig:panel_mtrue2}
\end{figure}

Taken together, the CASE summaries in Tables~\ref{tab:sim_results_r1} and~\ref{tab:sim_results_r2} and the norm-ratio summaries in Figures~\ref{fig:panel_mtrue1} and~\ref{fig:panel_mtrue2} indicate that the bootstrap-based selector provides the most
stable performance across sample sizes and regression complexities in the settings considered here, albeit at a higher computational cost.
In large-scale applications where runtime is a primary constraint, the RoT selector can
serve as a fast baseline (or initializer), often providing competitive fits when \(n\) is large.

\section{Analysis of the spatial orientation under sensory conditions}
\label{sec:results}

We now apply the proposed mixed-covariate circular regression methodology to the spatial-updating data of
\citet{legge2016indoor}, briefly described in the Introduction. The data are publicly available through the
University of Minnesota Data Repository (DRUM) at \url{http://hdl.handle.net/11299/182367}; for convenience,
we also distribute the analysis-ready dataset used in this paper with the \textsf{R} package \texttt{circMixedReg}.
As outlined in the Introduction, our response is the signed
angular error, computed as the circular difference between the reported and true target directions. This outcome
is directly relevant to spatial behavior, as it quantifies how accurately participants update the egocentric
location of a target after movement.
For visualization and interpretation, we represent this circular error on $(-\pi,\pi]$ so that $0$ denotes ``no
error'' and positive/negative deviations are displayed symmetrically. Section~\ref{sec:statistical_model} develops the model and estimator with the response on $[0,2\pi)$ via the
$\operatorname{atan2}$ mapping, but the two parameterizations are equivalent modulo $2\pi$, and all estimation
and inferential summaries are invariant to this wrapping convention. We therefore adopt $(-\pi,\pi]$ here solely to
aid interpretability.
Alongside the response, we consider one categorical and one continuous predictor. The categorical covariate is
the sensory condition of the trial, which modulates the available perceptual input and is therefore expected to
affect spatial-updating accuracy. As continuous input, we use the target distance in the main analysis
below, as it captures the physical demand of the task. The participant's distance-estimation error (reported minus true distance) can be
viewed as a complementary continuous measure reflecting subjective difficulty. Its inclusion as an additional
covariate is considered separately in Section~\ref{sec:two_predictors}, and an alternative one-predictor model
based solely on distance-estimation error is reported in Section~\ref{app:alt_predictor} of the Supplementary Material.
%Finally, although each participant contributes multiple trials, we analyze the data at the trial level to
%focus on the relationship between directional error, task demand, and sensory restriction. We do not include the
%vision group in the main specification, since it is constant within participants and would mix subject-level
%traits with trial-level predictors; we return to this point in the Discussion.
Finally, although the data have a nested structure (multiple trials per participant), we analyze the observations at the trial level to characterize how directional error varies with task demand and sensory restriction at the population level. In this setting, the i.i.d. framework provides a convenient benchmark for presenting the estimator, while the Discussion outlines natural extensions that explicitly account for within-subject dependence. We do not include vision group in the main specification, since it is constant within participants and would mix subject-level traits with trial-level predictors. We return to this point in the Discussion.

\begin{figure}[!htb]
\centering
\begin{subfigure}[t]{0.19\textwidth}
  \centering
  \includegraphics[width=\linewidth]{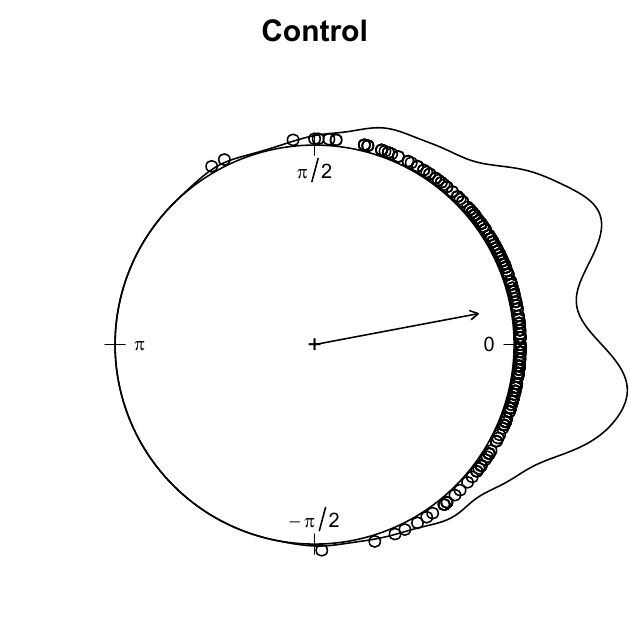}
\end{subfigure}\hfill
\begin{subfigure}[t]{0.19\textwidth}
  \centering
  \includegraphics[width=\linewidth]{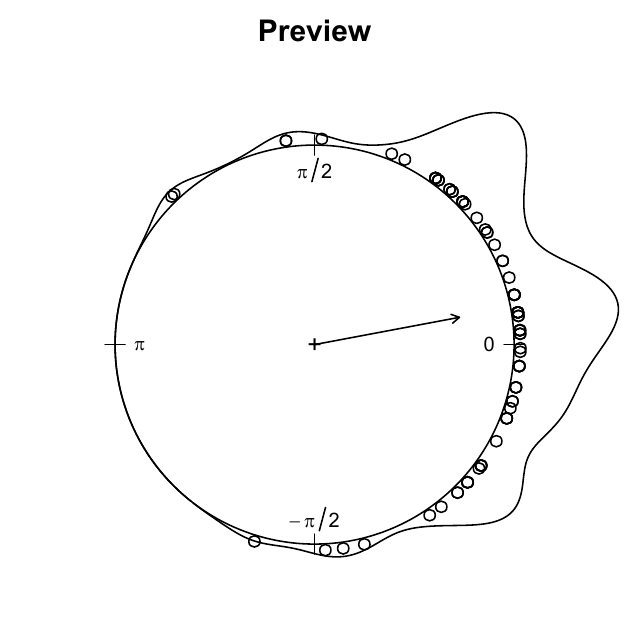}
\end{subfigure}\hfill
\begin{subfigure}[t]{0.19\textwidth}
  \centering
  \includegraphics[width=\linewidth]{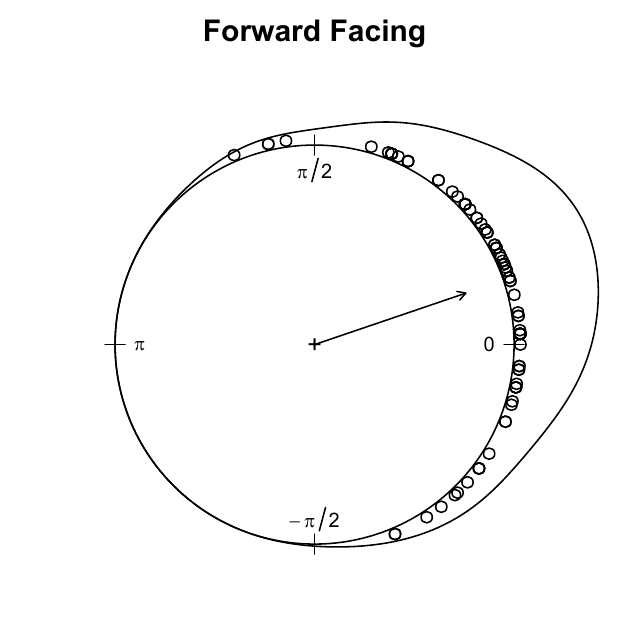}
\end{subfigure}\hfill
\begin{subfigure}[t]{0.19\textwidth}
  \centering
  \includegraphics[width=\linewidth]{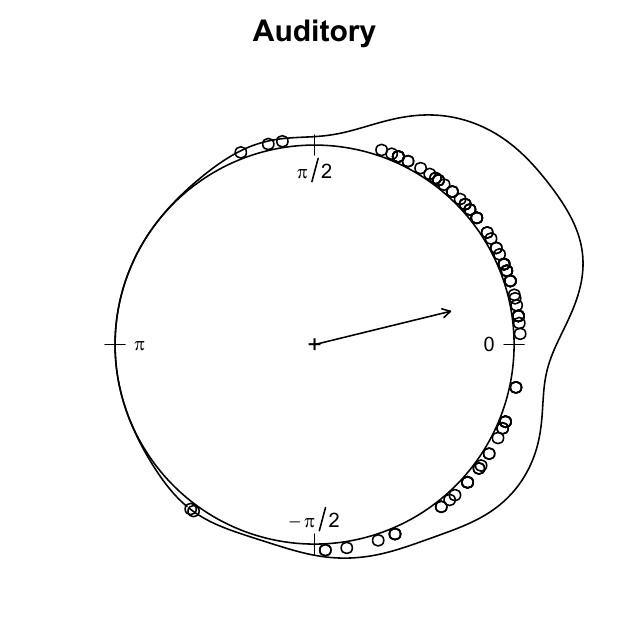}
  \end{subfigure}\hfill
\begin{subfigure}[t]{0.19\textwidth}
  \centering
  \includegraphics[width=\linewidth]{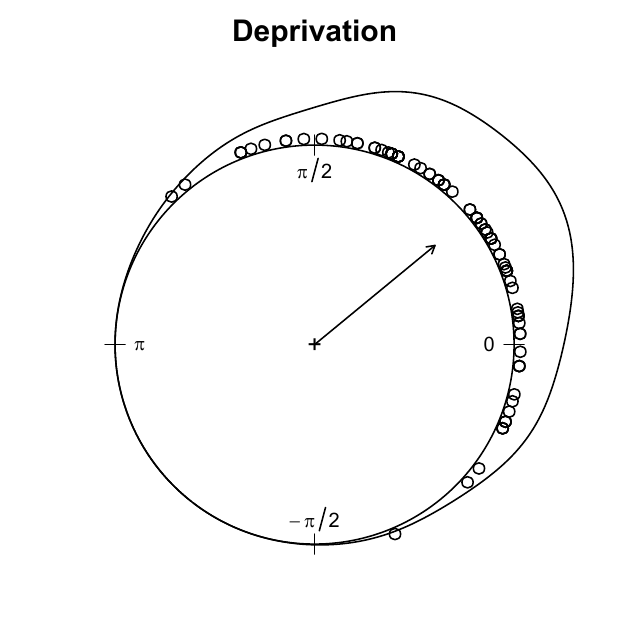}
  \end{subfigure}

\caption{Summaries of directional errors by sensory condition. Points on the outer ring are individual observations. The curve drawn outside the reference circle is a von~Mises kernel density estimates. The arrow indicates the sample mean direction, with length proportional to the mean resultant length. 
%The angular reference is $0$ to the right, $\pi/2$ at the top, $\pi$ to the left, and $-\pi/2$ at the bottom.
}
\label{polar_kde}
\end{figure}

Before presenting the condition-specific regression fits, we provide a brief exploratory summary of the signed directional error by sensory condition. Figure~\ref{polar_kde} displays circular density plots for each condition. In each panel, the points on the outer ring represent the individual observations, the curve drawn outside the reference circle is a von~Mises kernel density estimate of the error distribution, and the arrow indicates the sample mean direction, with length proportional to the mean resultant length (longer arrows indicate higher concentration). The angular reference is $0$ to the right, $\pi/2$ at the top, $\pi$ to the left, and $-\pi/2$ at the bottom.

Several qualitative patterns emerge. Under \textit{Control}, the error distribution is tightly concentrated around $0$, with a mean direction close to $0$ and the largest mean resultant length, indicating small and relatively concentrated directional errors. \textit{Preview} also remains centered near $0$, though with weaker concentration. In \textit{Forward Facing}, the mean direction shifts modestly away from $0$ while the distribution remains fairly concentrated. The \textit{Auditory} condition stays centered relatively close to $0$ but shows the lowest concentration, consistent with increased dispersion. Finally, \textit{Deprivation} exhibits the clearest departure in mean direction, with a noticeable positive shift together with substantial spread. Overall, these summaries suggest that sensory restrictions are associated with both changes in concentration and systematic directional bias, motivating flexible mixed-covariate circular regression models that can adapt across conditions.

%Before presenting the condition-specific regression fits, we carried out a descriptive exploratory analysis of
%the angular error by sensory condition. Figure~\ref{polar_kde} summarizes these directional errors for each condition using circular plots. Each panel shows: (i) the individual observations displayed on the outer ring (points), (ii) a von Mises kernel density estimate of the error distribution (shaded polygon), represented in polar coordinates so that, for each direction, the radial distance from the center is proportional to the estimated density, and drawn using a \emph{common radial scale across conditions} to enable direct visual comparison (so that differences in height reflect relative concentration rather than panel-specific rescaling), and (iii) the sample mean direction (arrow), whose length is proportional to the mean resultant length and, therefore, reflects concentration (longer arrows indicate less dispersion). The angular reference is $0$ to the right, $\pi/2$ at the top, $\pi$ to the left, and $-\pi/2$ at the bottom. 

\begin{figure}[!htb]
\centering
\includegraphics[width=0.85\textwidth]{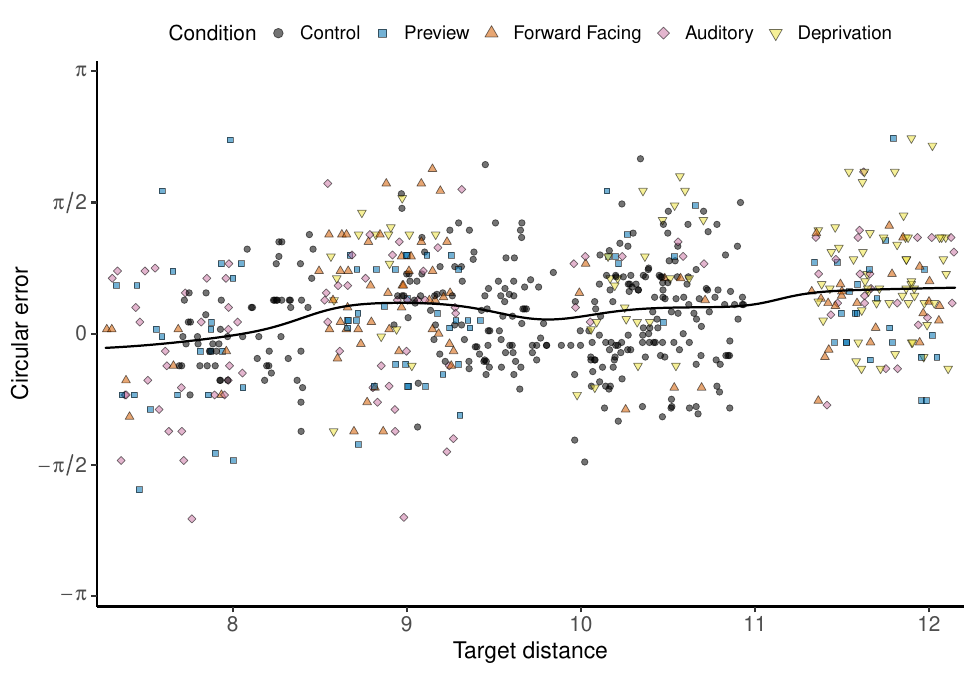}
\caption{Circular error versus target distance, colored and shaped by sensory condition (semi-transparent points to mitigate overlap). The overlaid curve is a global NW circular fit (Gaussian kernel) using a rule-of-thumb bandwidth \(h_{\text{RoT}}=0.354\). The fit ignores the condition factor and is computed on the original (non-jittered) data.}
\label{fig:scatter_global}
\end{figure}

As a complementary view of the raw data, Figure~\ref{fig:scatter_global} shows a scatterplot of circular error
versus target distance, with points colored and shaped by sensory condition. To improve readability in the
presence of repeated values, a small horizontal jitter is added to the distance values. We also overlay a global
nonparametric trend obtained from a NW circular estimator with a Gaussian kernel
\citep{dimarzio2013nonparametric}, using a robust rule-of-thumb bandwidth
\(h_{\text{RoT}} = 1.06\,\hat{\sigma}\,n^{-1/5} = 0.354\), where
\(\hat{\sigma}=\min\{s,\text{IQR}/1.349\}\) is a robust scale estimate of the predictor based on the sample standard deviation $s$ and the interquartile range \(\text{IQR}\). This choice provides a
stable descriptive smooth for visualization. The global fit indicates an overall upward tendency of the angular
error with distance, with the steepest increase occurring at the largest distances. At shorter distances, the error
remains comparatively stable, while dispersion increases as targets become more distant. Importantly, this pooled trend is only meant as an exploratory summary: it averages across sensory conditions and
therefore cannot reveal condition-specific patterns. This motivates fitting our mixed-covariate circular regression model, which estimates a separate smooth function
of distance for each sensory condition; uncertainty is then quantified by constructing simultaneous confidence
bands via the bootstrap procedure in Section~\ref{sec:bands}.

\subsection{Regression estimates and simultaneous confidence bands}
\label{sec:regs}

Building on the exploratory pooled smooth above, we now fit the mixed-kernel circular regression estimator in
\eqref{eq:est} to obtain condition-specific regression curves. The response is the signed angular error, the
continuous predictor is the target distance, and the sensory condition enters as a categorical factor. 
Following the convention adopted above, we display angular errors on $(-\pi,\pi]$ (with $0$ denoting ``no error'').
All conclusions are invariant to this representation modulo $2\pi$. The continuous component
uses a Gaussian kernel, and the categorical component uses the Aitchison--Aitken kernel. We computed the
estimator under three bandwidth selectors, cross-validation, rule-of-thumb, and the proposed bootstrap
criterion, yielding $\bm H_{\text{CV}}=(0.31,\,0.12)$, $\bm H_{\text{RoT}}=(0.37,\,0.05)$, and
$\bm H_{\text{boot}}=(0.28,\,0.08)$, respectively. For brevity, we report only the bootstrap-based curves and bands in this section. The bootstrap selector was marginally more stable in the simulation study, and in this application, all three selectors yield essentially the same fit and lead to the same substantive conclusions. The corresponding fits using CV and RoT selectors are reported in Section~\ref{app:selectors} of the Supplementary Material.

\begin{figure}[!htb]
\centering
\begin{subfigure}[b]{0.32\textwidth}
    \includegraphics[width=\textwidth]{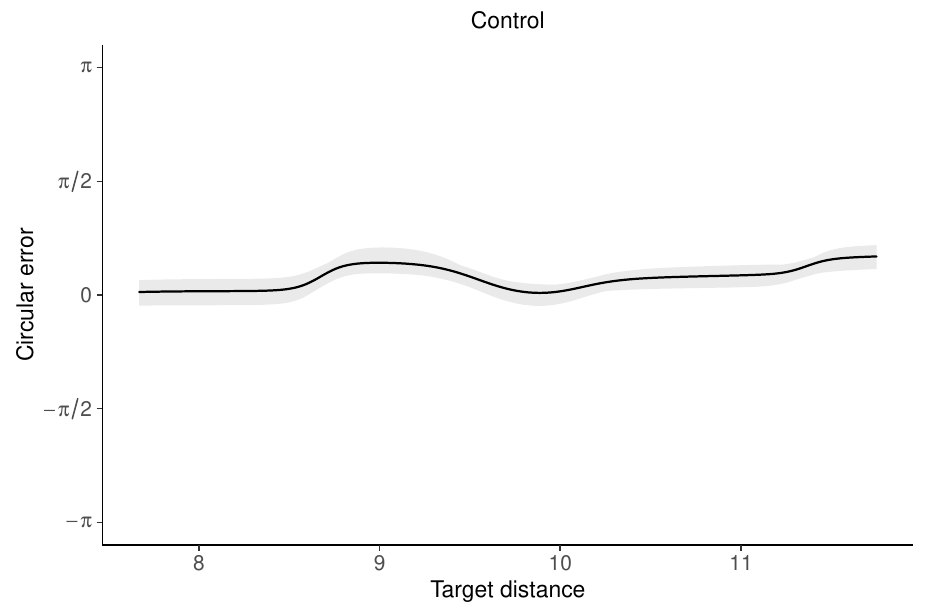}
\end{subfigure}
\begin{subfigure}[b]{0.32\textwidth}
    \includegraphics[width=\textwidth]{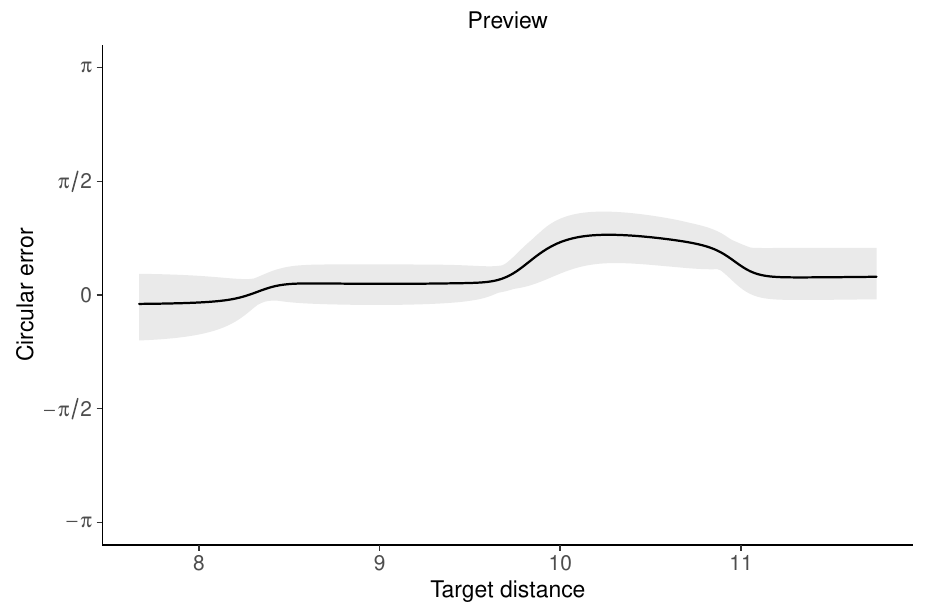}
\end{subfigure}
\begin{subfigure}[b]{0.32\textwidth}
    \includegraphics[width=\textwidth]{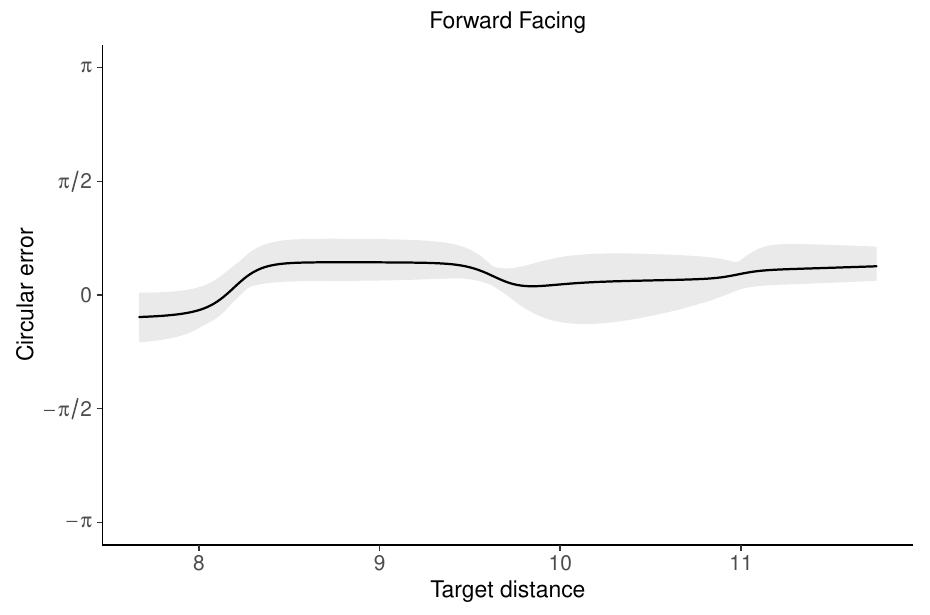}
\end{subfigure}

\vspace{0.6em}

\begin{minipage}{0.66\textwidth}
\centering
\begin{subfigure}[b]{0.48\textwidth}
  \includegraphics[width=\textwidth]{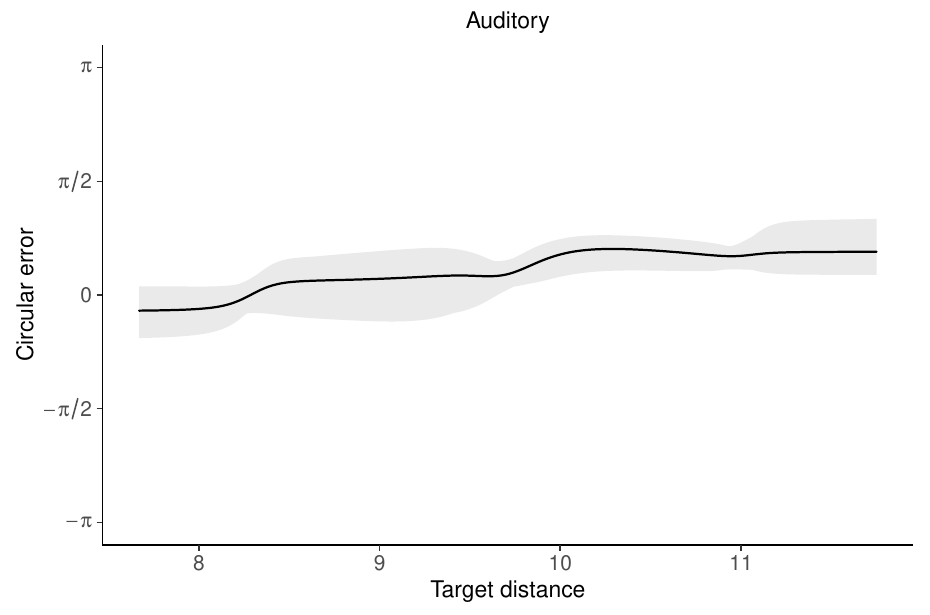}
\end{subfigure}
\hfill
\begin{subfigure}[b]{0.48\textwidth}
  \includegraphics[width=\textwidth]{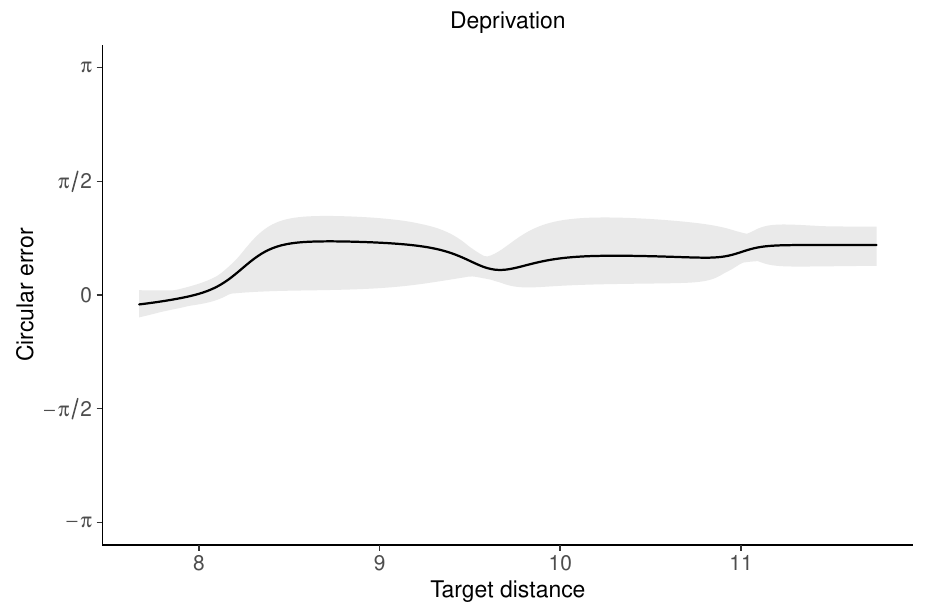}
\end{subfigure}
\end{minipage}

\caption{
Estimated circular regression curves and simultaneous confidence bands for the estimated circular regression functions under each sensory condition. Bands were computed via bootstrap with \( B = 200 \) resamples and significance level \( \alpha = 0.05 \). 
%The bottom-right panel summarizes the mean angular width of each band.
}
\label{fig:confidence_bands}
\end{figure}

Figure~\ref{fig:confidence_bands} shows, for each condition, the nonparametric regression estimate computed with the bootstrap-selected bandwidth $\bm H_{\text{boot}}$, together with $95\%$ simultaneous bootstrap confidence bands ($\alpha=0.05$, $B=200$), calibrated to provide global coverage over the full distance range within each condition.
Importantly, we do \emph{not}  fit separate models by condition. Instead, the estimator remains unified across conditions through the mixed-kernel weights $w_{\bm H}^i(\bm x,\bm z)$ in \eqref{eq:pesos}. In the present setting (one continuous covariate and one factor), these weights are proportional to $K_h(x-X_i)\,L_\lambda(z,Z_i)$ and are normalized to sum to one over $i$. Thus, all observations contribute to $\hat m_{\bm H}(x,z)$, with influence modulated by distance proximity and categorical concordance.
This induces a principled borrowing of strength across conditions, while still allowing condition-specific regression structure.

The five panels in Figure~\ref{fig:confidence_bands} reveal clear, condition-dependent relationships between
target distance and signed angular error. In \textit{Control}, the estimated curve stays close to the no-error
reference across most of the distance range, and the simultaneous band is comparatively tight, indicating stable
directional updating when both vision and hearing are available. By contrast, under reduced sensory input, the
mean error departs from zero in structured, distance-dependent ways. In \textit{Auditory}, the fitted curve
increases with distance, transitioning from slightly negative values at the shortest distances to persistently
positive errors at moderate-to-long distances, consistent with a gradual drift when orientation relies primarily
on sound. The \textit{Forward Facing} and \textit{Deprivation} conditions show an early rise to positive errors at
short-to-intermediate distances, followed by a mild flattening (and a local dip around the mid-range) before
returning to positive values toward the largest distances. Finally, \textit{Preview} exhibits a nonmonotone
profile, remaining near zero at short distances and displaying a mid-to-long-distance bump that attenuates
toward the upper end of the observed range.
Because the bands are simultaneous over the distance grid within each sensory condition, they support global
inference in $x$ for a fixed $z$. In particular, when the horizontal reference at 0 (no-error) lies entirely
outside the band over a distance interval, this provides evidence at the stated confidence level that the mean
signed directional error differs from zero throughout that interval (on the evaluated grid). Throughout, we
interpret the fitted errors on the $(-\pi,\pi]$ scale and avoid over-interpreting any behavior that could be
affected by the $\pm\pi$ wrap-around.

Overall, the fitted curves and simultaneous bands provide a compact, condition-specific summary of how the \emph{signed directional error} varies with task demand (target distance). The visually rich \textit{Control} condition stays close to the no-error reference across most of the distance range, whereas conditions with reduced or absent vision tend to show more pronounced positive deviations over moderate-to-long distances and wider uncertainty bands, particularly near the extremes of the observed distance range. These findings are consistent with a setting in which sensory restriction is associated not only with increased variability but also with structured, distance-dependent directional distortions. Practically, the results identify distance ranges and sensory regimes where systematic error is most evident, and they illustrate how the proposed mixed-covariate circular regression framework can extract interpretable, condition-specific patterns without stratifying the data or imposing restrictive parametric assumptions, an approach that quantifies both systematic angular bias and uncertainty as smooth functions of task demand under different sensory conditions, and is directly relevant to spatial navigation, rehabilitation, and human-machine interaction studies involving directional responses.

\subsection{Residual analysis}
\label{sec:resi}

As a next step, we examine residual diagnostics to check whether the fitted regression function captures the
main systematic structure in the angular errors and whether the resulting centered residuals are reasonably
compatible with the residual-resampling scheme used for our bootstrap-based bandwidth selection and band
construction. Using the regression estimates obtained with the bootstrap-selected bandwidth, we consider pooled
and condition-specific diagnostics, circular-uniformity tests, and sensitivity checks. For brevity, we report
only the main conclusions here; full results (additional figures, alternative bandwidth selectors, and detailed
test outputs) are provided in Section~\ref{app:residuals} of the Supplementary Material.

Using the bootstrap-selected bandwidth $\bm H_{\text{boot}}$, circular residuals are defined as
$
\hat{\varepsilon}_i
=\big[\Theta_i-\hat m_{\bm H_{\text{boot}}}(\bm X_i,\bm Z_i)\big]\ (\operatorname{mod}\, 2\pi),
$
and are mapped to $(-\pi,\pi]$ to retain the interpretation as signed angular deviations.
Figure~\ref{fig:residuals_boxplot_bootstrap} shows circular residual boxplots by sensory condition.
Across conditions, residuals are centered close to $0$, indicating no strong directional bias, although
dispersion varies.
\textit{Preview} exhibits the widest spread, whereas \textit{Control} and \textit{Forward Facing} are more
concentrated. \textit{Auditory} and \textit{Deprivation} lie in between, with a slight clockwise shift.
A few outliers are present, particularly in \textit{Control}. Numerical summaries in the Supplementary Material support these
patterns. A pooled residuals-versus-fitted plot (Fig.~\ref{fig:residuals_vs_fitted_sm} in the Supplementary Material) reveals no discernible trend or
heteroscedasticity.

\begin{figure}[!htb]
    \centering
    \begin{subfigure}[b]{0.19\textwidth}
        \centering
        \includegraphics[width=\textwidth]{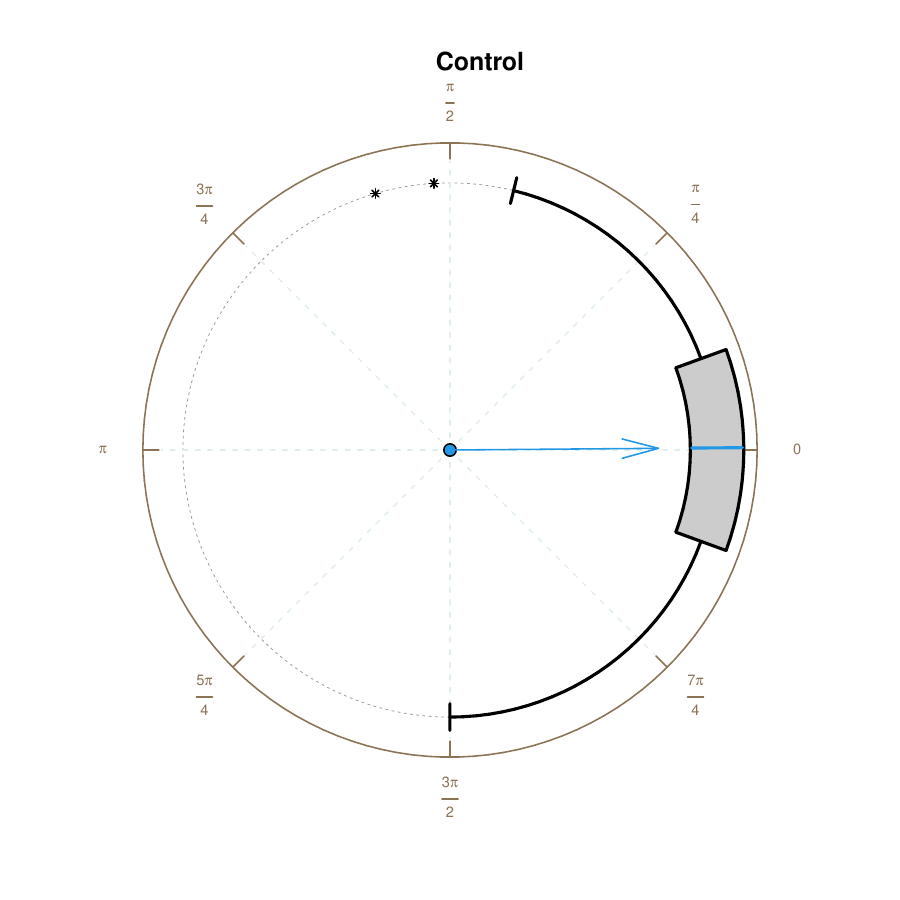}
    \end{subfigure}
    \hfill
    \begin{subfigure}[b]{0.19\textwidth}
        \centering
        \includegraphics[width=\textwidth]{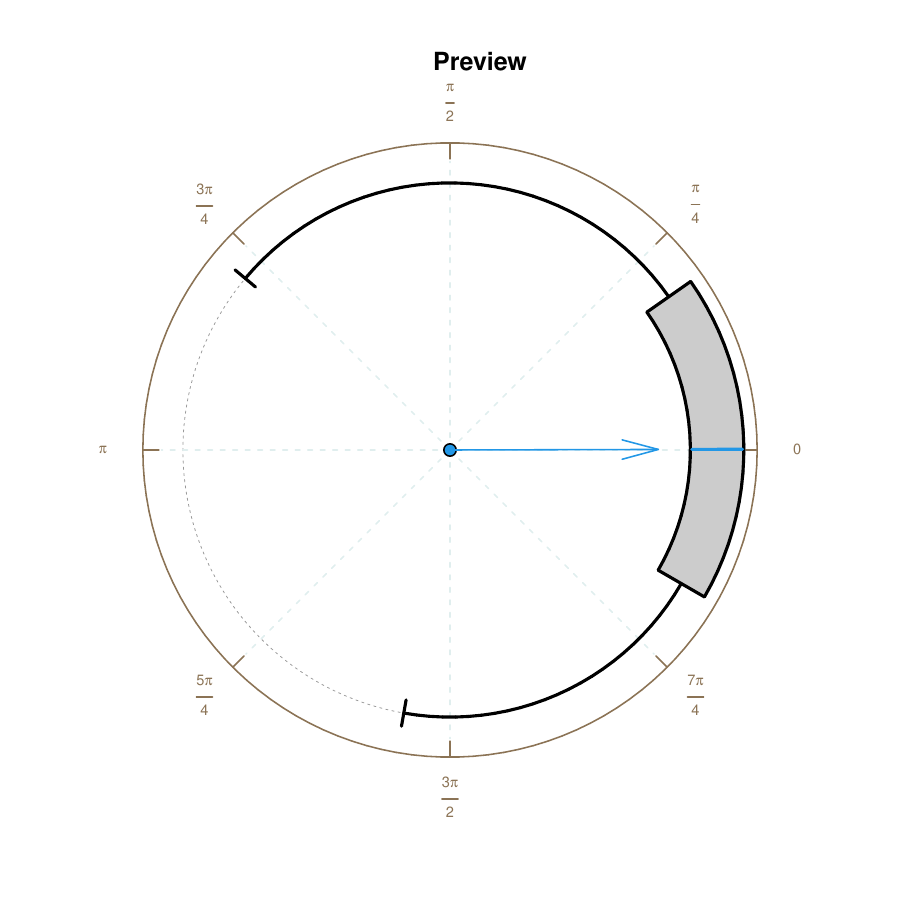}
    \end{subfigure}
    \hfill
    \begin{subfigure}[b]{0.19\textwidth}
        \centering
        \includegraphics[width=\textwidth]{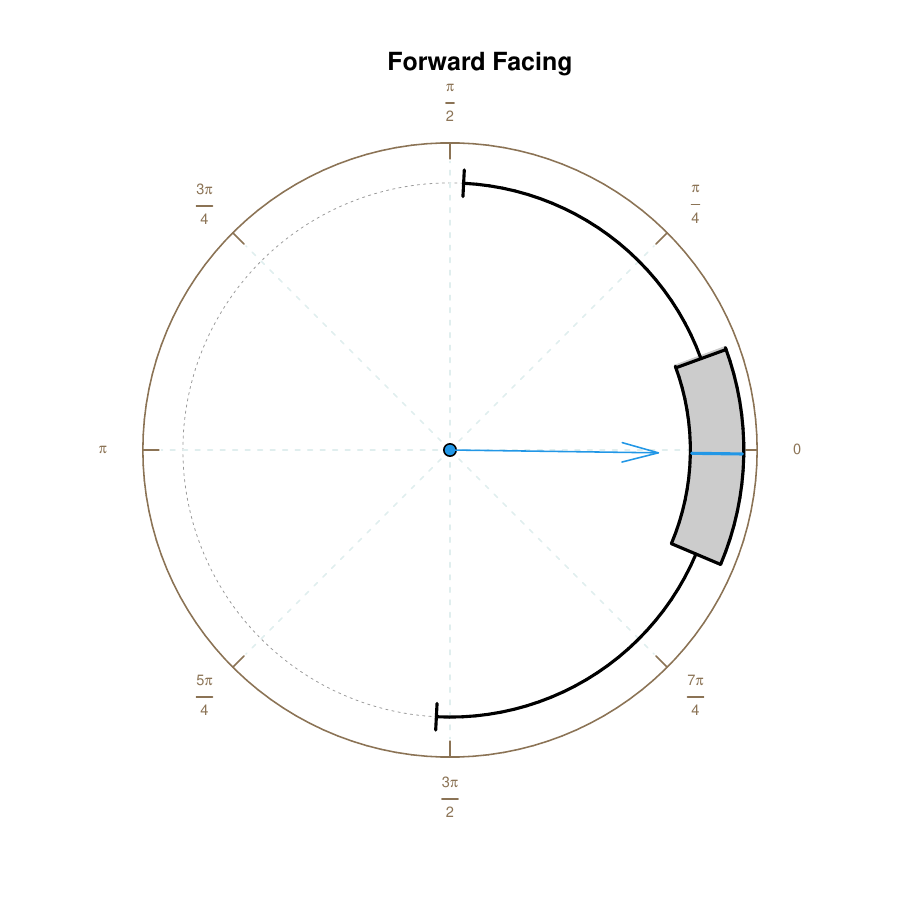}
    \end{subfigure}
    \hfill
    \begin{subfigure}[b]{0.19\textwidth}
        \centering
        \includegraphics[width=\textwidth]{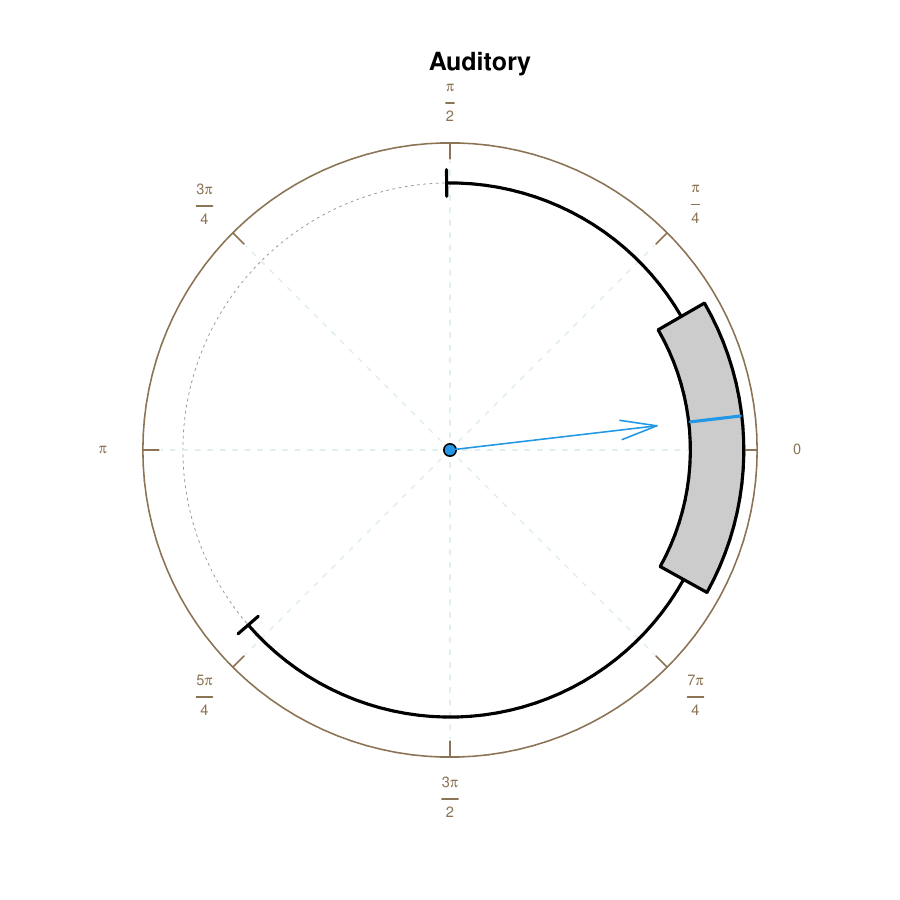}
    \end{subfigure}
    \hfill
    \begin{subfigure}[b]{0.19\textwidth}
        \centering
        \includegraphics[width=\textwidth]{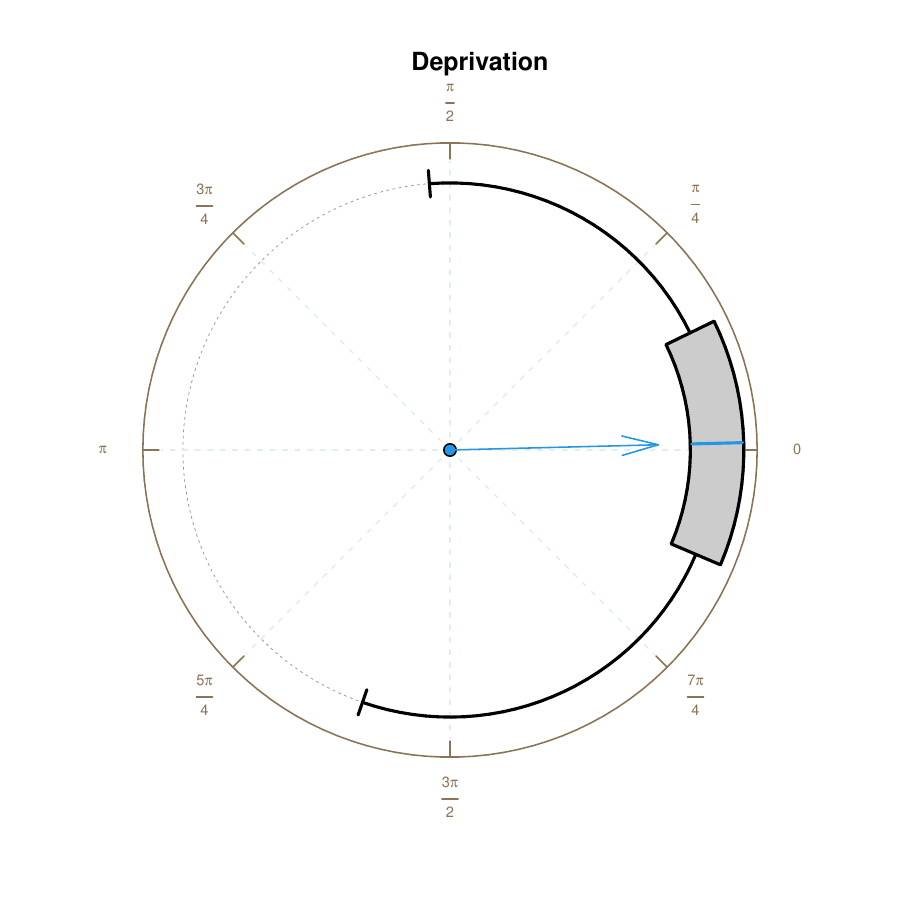}
    \end{subfigure}
    \caption{Circular residual boxplots by sensory condition for the bootstrap estimator $\hat m_{\bm H_{\text{boot}}}$.}
    \label{fig:residuals_boxplot_bootstrap}
\end{figure}

To assess whether residuals exhibit remaining circular structure, we also provide in the Supplementary Material standard tests of circular
uniformity on the pooled residuals (Rayleigh, Watson $U^2$, and Kuiper; see, e.g.,
\citealp{mardia2009directional,watson1961goodness,kuiper1960tests}). These tests reject strict uniformity,
suggesting mild residual structure under the most challenging conditions, but they do not point to a substantial
model misspecification.

To summarize global goodness-of-fit, we report two metrics tailored to the cosine loss used throughout the paper.
First, we compute an observational version of the circular average squared error, $\operatorname{CASE}_{\text{obs}}(\bm H)$, 
which corresponds to \eqref{eq:CASE} with the unknown regression function replaced by the observed response.
Second, we report a cosine-loss pseudo-coefficient of determination, $R^2_{\text{circ}}
= 1-\text{SSE}_{\text{circ}}/\text{SST}_{\text{circ}}$, 
with
$
\text{SSE}_{\text{circ}}
=\sum_{i=1}^n \{1-\cos[\Theta_i-\hat m_{\bm H}(\bm X_i,\bm Z_i)]\},
$ and $
\text{SST}_{\text{circ}}
=\sum_{i=1}^n \{1-\cos(\Theta_i-\bar\Theta)\},
$
where $\bar\Theta$ is the sample mean direction. This coefficient equals $1$ for a perfect fit and equals $0$ for
the circular-mean predictor (it may be negative if the fitted model underperforms the circular mean). See
Section~\ref{app:residuals} of the Supplementary Material for further discussion.
For the bootstrap-selected bandwidth, we obtain
$\operatorname{CASE}_{\text{obs}}(\bm H_{\text{boot}})=0.181$ and $R^2_{\text{circ}}=0.144$.
A full comparison with $\bm H_{\text{CV}}$ and $\bm H_{\text{RoT}}$ is reported in the Supplementary Material: the bootstrap selector
attains the smallest $\operatorname{CASE}_{\text{obs}}$ and the largest $R^2_{\text{circ}}$, although the
differences are small. Additional diagnostics, including condition-wise $\operatorname{CASE}_{\text{obs}}$ and
pooled residuals-versus-fitted plots, are also provided in the Supplementary Material.

Overall, these diagnostics indicate that the fitted model captures the main structure of the circular response.
Residuals are centered near zero across conditions, with dispersion patterns consistent with task difficulty, and
the pooled residuals-versus-fitted plot shows no systematic trend. The uniformity tests suggest some remaining
structure, most pronounced under \textit{Deprivation}, but not of a magnitude that undermines inference or band
construction in the present analysis. Any remaining structure could be addressed in future work by allowing
condition-specific concentration (heteroscedasticity), incorporating subject/room effects, or adopting a
hierarchical circular regression formulation.

\subsection{Extension to two continuous predictors}
\label{sec:two_predictors}

The main analysis above uses a single continuous predictor, the \emph{target distance} (true distance to the
target, in feet). The dataset also records each participant's reported distance, which captures a perceptual
component of the task. To assess whether this additional information improves predictive performance, we
consider a two-predictor specification with $X_1$ the target distance and $X_2$ the \emph{distance-estimation
error}, defined as reported minus true distance.

Let $\bm x=(x_1,x_2)$ denote covariate values and let $z$ denote the sensory condition. We estimate
$m(\bm x,z)$ using the same product-kernel circular estimator \eqref{eq:est}, now with weights
\[
w^i_{\bm H}(\bm x,z)
=\frac{K_1\!\left[(x_1-X_{i1})/h_1\right]\,
       K_2\!\left[(x_2-X_{i2})/h_2\right]\,
       L(z,Z_i;\lambda)}
      {\sum_{j=1}^n
       K_1\!\left[(x_1-X_{j1})/h_1\right]\,
       K_2\!\left[(x_2-X_{j2})/h_2\right]\,
       L(z,Z_j;\lambda)},
\]
where $\bm H=(h_1,h_2,\lambda)$, $K_1$ and $K_2$ are Gaussian kernels, and $L$ is the Aitchison--Aitken kernel for
the categorical factor. As in the one-predictor analysis, we wrap angles to $(-\pi,\pi]$ for visualization; this
display choice is immaterial to estimation.

\begin{figure}[!htb]
\centering
% --- top row ---
\begin{subfigure}[b]{0.32\textwidth}
  \includegraphics[width=\textwidth]{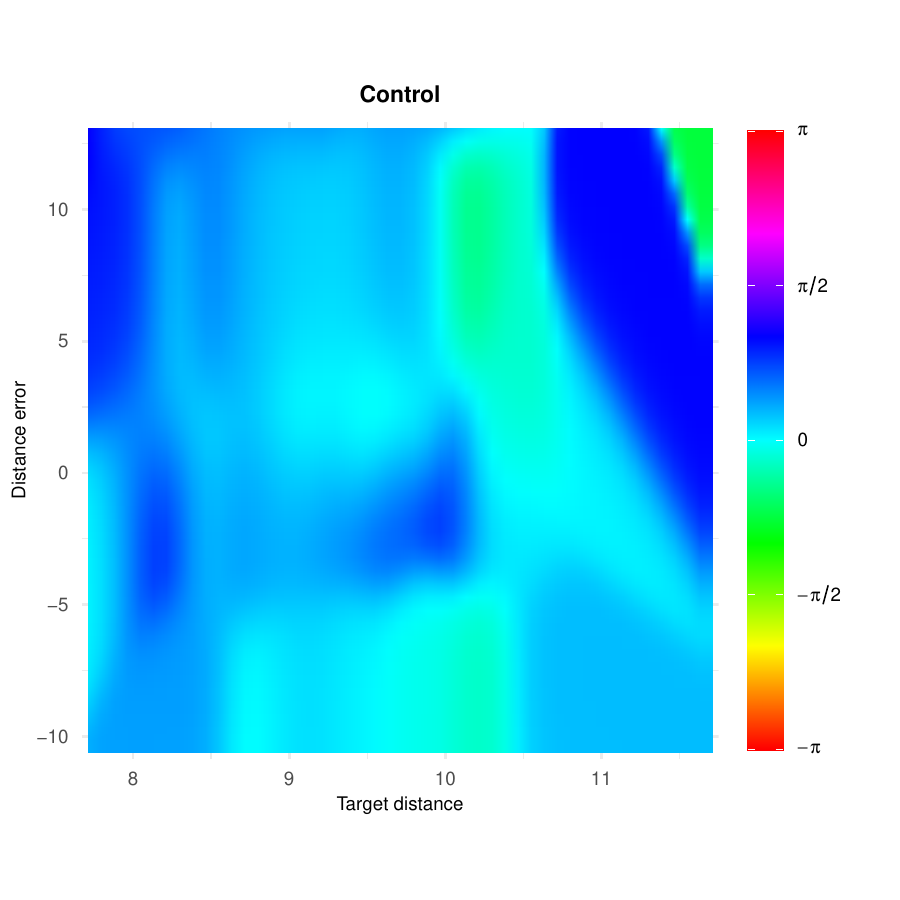}
\end{subfigure}\hfill
\begin{subfigure}[b]{0.32\textwidth}
  \includegraphics[width=\textwidth]{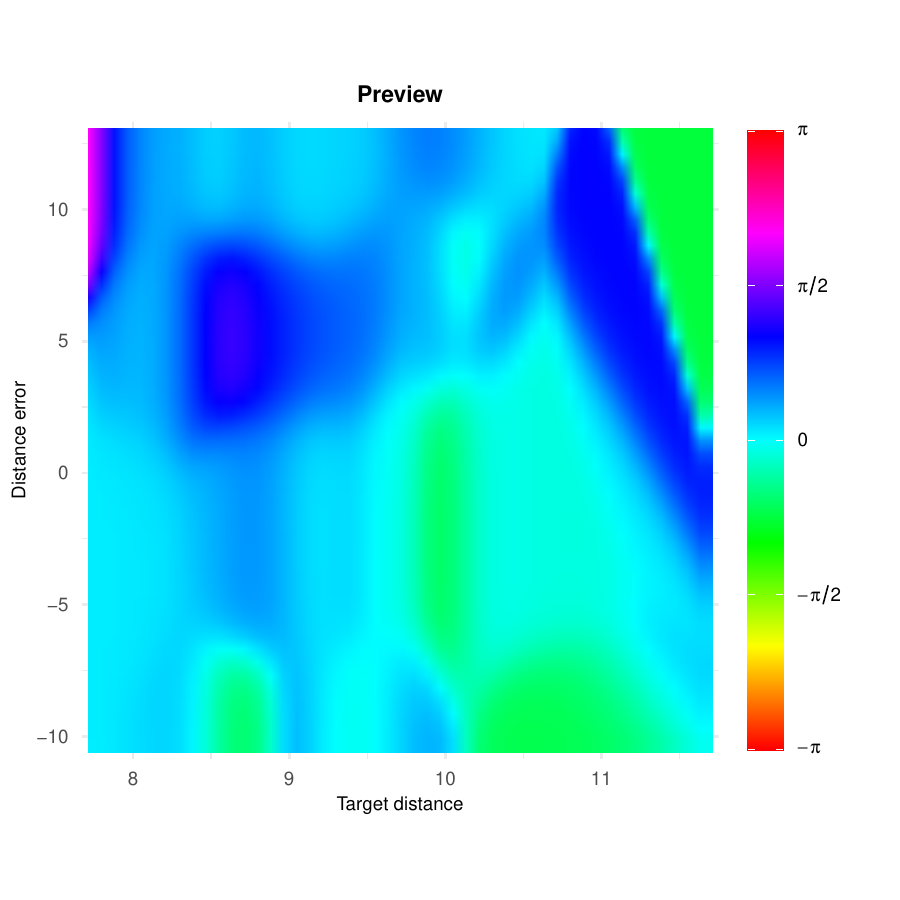}
\end{subfigure}\hfill
\begin{subfigure}[b]{0.32\textwidth}
  \includegraphics[width=\textwidth]{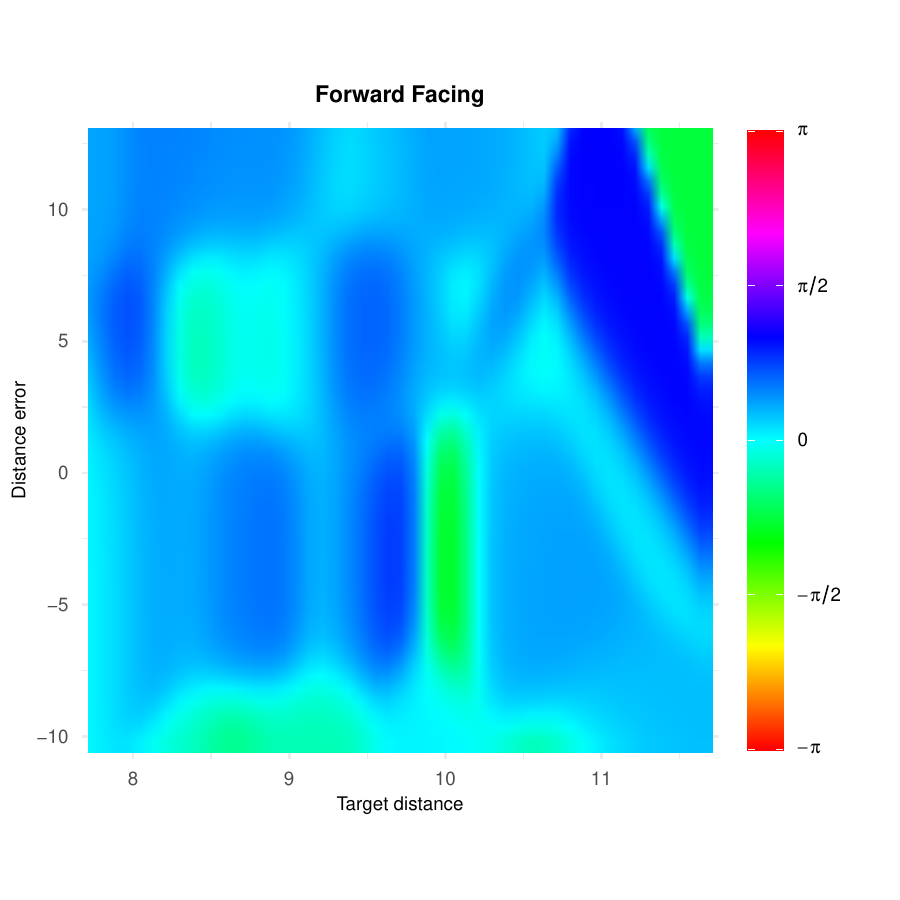}
\end{subfigure}

\vspace{0.6em}

% --- bottom row (two panels centered) ---
\begin{minipage}{0.66\textwidth}
\centering
\begin{subfigure}[b]{0.48\textwidth}
  \includegraphics[width=\textwidth]{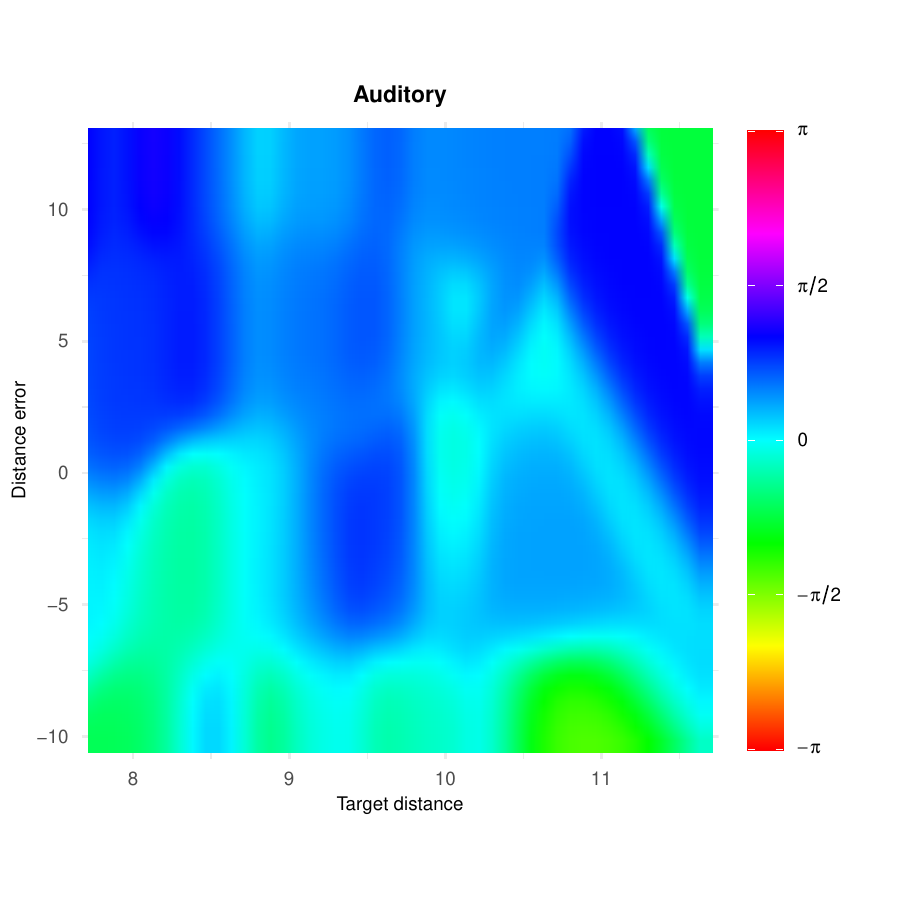}
\end{subfigure}
\hfill
\begin{subfigure}[b]{0.48\textwidth}
  \includegraphics[width=\textwidth]{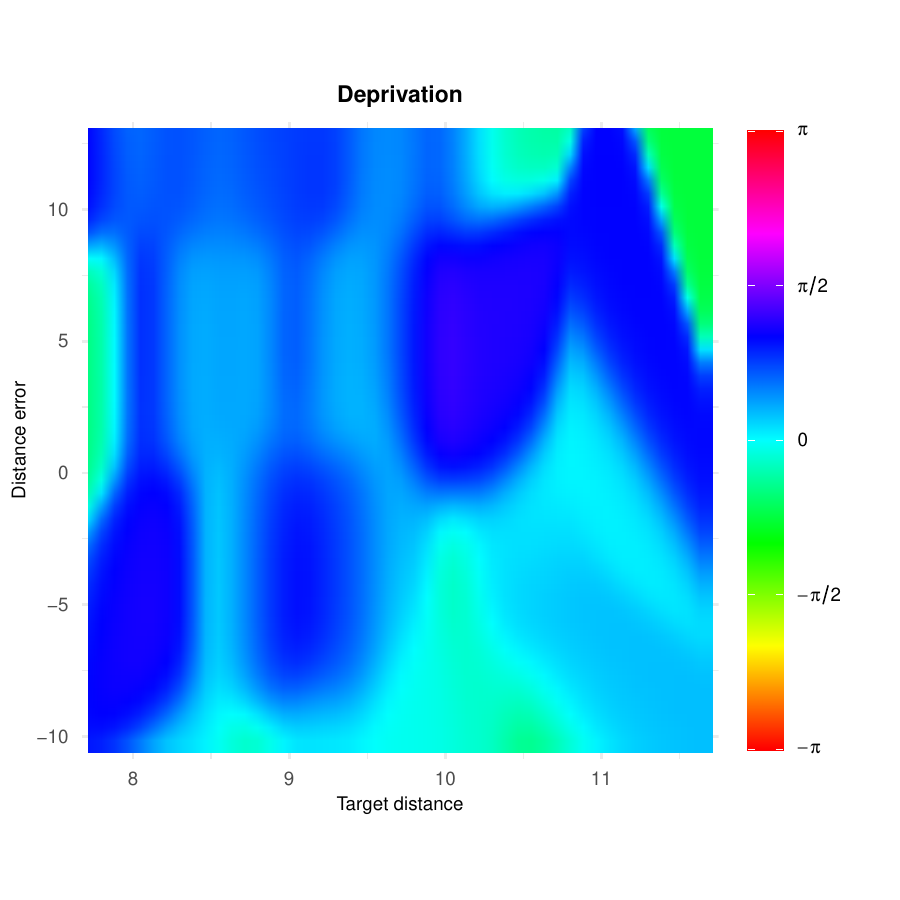}
\end{subfigure}
\end{minipage}

\caption{Estimated circular regression surfaces using a
single global bandwidth $\bm H_{\text{RoT}}=(0.37,0.89,0.05)$ shared across conditions. Color encodes the
predicted angular response; the circular colormap ensures that identical hues correspond to equivalent
directions modulo $2\pi$.}
\label{fig:surface_condition}
\end{figure}

For this two-dimensional illustration, we use the rule-of-thumb selector, which is computationally cheap and
stable in this setting. Figure~\ref{fig:surface_condition} displays the estimated regression surfaces
$\hat m_{\bm H}(x_1,x_2,z)$ for the five sensory conditions using the common bandwidth
$\bm H_{\text{RoT}}=(0.37,0.89,0.05)$.
The fits reveal systematic differences across conditions and suggest that the distance effect may vary with
misestimation error. Under \textit{Control} the surface is smooth and regular, with gradual angular variation along
both covariates; \textit{Preview} exhibits a similar overall pattern with modest local features. In
\textit{Forward Facing} a narrow vertical feature around intermediate distances ($\approx 10$~ft) is consistent
with a localized phase shift. In \textit{Auditory}, the structure becomes more heterogeneous, particularly for
negative distance errors at larger target distances, whereas \textit{Deprivation} shows the steepest gradients and
more localized irregularities for positive distance errors. Overall, variation along $x_2$ is mild in visually
rich or anticipatory-vision settings (\textit{Control}, \textit{Preview}), while under reduced or non-visual cues
(\textit{Forward Facing}, \textit{Auditory}, \textit{Deprivation}) the response changes more sharply with $x_1$ and
exhibits visible asymmetries between positive and negative $x_2$ values.

From an applied standpoint, these surfaces provide a compact description of directional error as a \emph{joint}
function of physical demand ($x_1$) and perceptual bias ($x_2$). They help delineate regions of the $(x_1,x_2)$
plane where performance is relatively stable versus regions where errors increase rapidly, and they suggest
testable hypotheses about cue use and reweighting under sensory restriction (e.g., regimes in which perceptual
bias appears to dominate versus regimes driven primarily by distance). This perspective can inform follow-up
designs by motivating targeted sampling of informative combinations of distance and misestimation, rather than
varying each factor in isolation.

Two remarks are in order. First, sharp color transitions near $\pm\pi$ reflect circular wrap-around rather than
true discontinuities of the underlying surface. Second, the combination of sparse data near the boundary of the
$(x_1,x_2)$ domain and a global rule-of-thumb bandwidth can produce mild ``banding'' artifacts. These plots are
intended to demonstrate how the proposed mixed-covariate framework extends to multiple continuous predictors. A 
more exhaustive multivariate analysis (e.g., adaptive bandwidths and formal uncertainty quantification for
surfaces) is beyond the scope of the present application.

To quantify the potential gain from including $X_2$, we compare the observational cosine-loss error
$\operatorname{CASE}_{\text{obs}}$ for the one- and two-predictor models under the same rule-of-thumb selection.
The one-predictor model (distance only) yields $\operatorname{CASE}_{\text{obs}}=0.182$, whereas the two-predictor
model (distance and distance-estimation error) yields $\operatorname{CASE}_{\text{obs}}=0.152$, indicating an
improved in-sample fit. Overall, this extension illustrates that additional continuous covariates can be
incorporated into the proposed framework with minimal changes to the estimator and can capture extra structure in
the angular response beyond that explained by distance alone.

\section{Discussion}
\label{sec:discussion}

Our analyses show that sensory conditions are strongly associated with spatial orientation performance.
Under full sensory access (\textit{Control}), participants exhibit smaller directional errors, whereas partial or
total sensory restriction is accompanied by broader error distributions and increased variability. These
patterns are consistent with the empirical findings reported by \citet{legge2016indoor}, while our statistical
framework goes beyond descriptive summaries by providing flexible regression estimates that capture nonlinear
and condition-specific trends in a data-driven manner. In contrast to the original analysis, which emphasized
group-level contrasts and mean errors, we model circular errors as a smooth function of distance within each
sensory condition, enabling the detection of subtler nonlinearities and the construction of simultaneous
confidence bands over distance (within each condition).

Beyond the statistical contribution, the fitted curves and simultaneous bands provide quantitative summaries
that may help inform subsequent work in spatial cognition and related applied settings. From a behavioral and
spatial-cognition perspective, the results are consistent with the view that non-visual modalities can support
egocentric updating, but typically with reduced precision relative to visually guided settings
\citep[e.g.,][]{schinazi2016navigation,gori2014impairment}. In clinical and rehabilitative contexts, the marked
differences across sensory conditions may help motivate more targeted assessments of orientation performance under
specific cue restrictions. More broadly, characterizing how directional error and uncertainty evolve with task
demand and sensory availability may be useful for informing the design and evaluation of assistive cues or
wayfinding strategies in built environments where sensory information is limited, degraded, or conflicting.
In our analysis, this is captured by condition-specific fitted curves together with simultaneous confidence bands
over distance (within each sensory condition).

From a methodological perspective, our circular regression framework offers several advantages. First, it
respects the intrinsic geometry of circular data. Second, the product-kernel formulation accommodates continuous
and categorical covariates in a unified nonparametric estimator, allowing interaction effects to emerge without
pre-specification. Third, the asymptotic bias and variance results extend existing circular regression theory to
mixed-type covariates for the NW estimator studied here. While not used directly in the data
analysis, they provide a foundation for future developments, including plug-in bandwidth selectors and
asymptotic inferential procedures. Fourth, the proposed bootstrap bandwidth selector performs favorably in
simulations and yields slightly improved fit summaries in the real-data analysis, offering a principled and
practically stable data-driven choice. All methods are implemented in the \textsf{R} package
\texttt{circMixedReg}, facilitating estimation, visualization, and diagnostics for circular regression with
mixed covariates.

While the proposed framework is flexible, several extensions would be worthwhile in settings such as ours.
Our real-data analysis is conducted at the trial level and does not explicitly model within-subject dependence
or other sources of clustering. Extensions to hierarchical formulations or mixed-effects circular regression,
incorporating, for instance, subject- and room-level effects, are therefore of clear interest. At the same time,
because sensory condition is the primary within-subject manipulation, the trial-level analysis provides a useful first-order characterization of population-averaged directional bias and dispersion patterns across sensory restrictions and task demands. Developing cluster-aware inference, including resampling schemes tailored to circular mixed-covariate regression, is a natural direction for future work, particularly to quantify the impact of within-subject dependence on inference.
%a trial-level analysis still provides a
%useful first-order characterization of how signed directional error varies with sensory restriction and task demand across trials. 
Another natural extension is to enrich the covariate set to separate stable subject-level
traits from trial-level sensory constraints. For example, the product-kernel formulation naturally extends to
additional categorical predictors, such as participants' vision group (sighted, low-vision, or blind). This
would help disentangle permanent visual ability from transient trial-level sensory constraints, without
requiring sample stratification or relying on classical ANOVA/ANCOVA-style designs in small subgroups, where
precision can deteriorate quickly \citep{alonso2021nonparametric}. Although adding categorical predictors
increases complexity and places additional demands on bandwidth selection, the extension is conceptually
straightforward within our framework and appears promising for future applications. Further developments could
also consider richer predictor sets (e.g., dynamic motion cues), variable-selection strategies, or adaptations
to circular functional settings.

Finally, a natural methodological extension is the LL version of the proposed mixed-covariate
circular regression estimator. LL smoothing has been investigated for circular regression with purely
continuous predictors in both univariate and multivariate settings
\citep{dimarzio2013nonparametric,meilan2021nonparametricTEST}. As noted in Remark~\ref{rem:LL} in
Section~\ref{sec:statistical_model}, an LL formulation can be constructed within our product-kernel framework by
replacing the NW local-constant fits by their LL counterparts in the continuous covariates (while keeping the categorical kernel component unchanged).
Our main theoretical development focuses on the NW estimator because it aligns most directly with the
existing kernel-regression theory for mixed-type predictors under Euclidean responses, where product-kernel
constructions and their large-sample properties are typically presented for NW-type estimators
\citep{racine2004nonparametric,li2007nonparametric}. Developing a fully parallel asymptotic theory for LL estimators in
the mixed-covariate setting (and then transferring it to circular outcomes) requires additional technical work and is
beyond the scope of this paper. Nonetheless, the \texttt{circMixedReg} package implements both NW and LL versions,
enabling empirical comparisons. Extending the mixed-covariate LL theory and its associated inference, bias correction,
and variable-selection strategies for circular responses remains a promising direction for future work.

%-----------------------------------------------%
\section*{Acknowledgments}
This work is part of the grants PID2023\allowbreak-147127OB\allowbreak-I00 and PID2024\allowbreak-158399NB\allowbreak-I00 ``ERDF/\allowbreak EU", funded by MCIN/AEI/\allowbreak10.13039/\allowbreak501100011033. It has also been supported by the Xunta de Galicia (Grupos de Referencia Competitiva, ED431C-2024/14) and by CITIC, a center accredited for excellence within the Galician University System and a member of the CIGUS Network. CITIC receives subsidies from the Department of Education, Science, Universities, and Vocational Training of the Xunta de Galicia and is co-financed by the EU through the FEDER Galicia 2021-27 operational program (Ref. ED431G 2023/01).
%-----------------------------------------------%

\bibliographystyle{apalike}
\bibliography{biblio_circ_mixed}

\newpage

\clearpage
\setcounter{figure}{0}
\setcounter{table}{0}
\setcounter{equation}{0}

\renewcommand{\thefigure}{S\arabic{figure}}
\renewcommand{\thetable}{S\arabic{table}}
\renewcommand{\theequation}{S\arabic{equation}}

\title{Supplementary material for ``Analyzing directional errors in spatial orientation using
nonparametric circular regression with mixed covariates''}
\setlength{\droptitle}{-1cm}
\predate{}%
\postdate{}%
\date{}
\author{Mario Francisco-Fernández$^{1,3}$ and Andrea Meil\'an-Vila$^{2}$}
\footnotetext[1]{Department of Mathematics, Universidade da Coruña (Spain).}
\footnotetext[2]{Department of Statistics, Universidad Carlos III de Madrid (Spain).}
\footnotetext[3]{Corresponding author. e-mail: \href{mailto:mario.francisco@udc.es}{mario.francisco@udc.es}.}
\maketitle

\begin{abstract}
 This Supplementary Material is organized as follows. Section~\ref{app:theory} states the
regularity assumptions used to derive the asymptotic bias and variance of the proposed kernel regression
estimator and provides the corresponding proofs. Section~\ref{app:data} reports additional results for the
real-data application.
\end{abstract}

\appendix

\section{Asymptotic properties of the estimator: assumptions and proofs}
\label{app:theory}

This section states the assumptions used to derive the asymptotic bias and variance of the proposed kernel
regression estimator and provides the corresponding proofs. The derivations proceed by rewriting the circular
problem in terms of two auxiliary Euclidean regression models for the sine and cosine components of the angular
response, which allows us to appeal to standard kernel-regression arguments. 

Consider independent and identically distributed (i.i.d.) observations \( \{(\bm{X}_i, \bm{Z}_i, \Theta_i)\}_{i=1}^n \) following the circular regression model
\begin{equation}
\Theta_i = [m(\bm{X}_i, \bm{Z}_i) + \varepsilon_i] \quad (\operatorname{mod}\, 2\pi), \quad i = 1, \ldots, n,
\label{eq:model_SM}
\end{equation}
where \( \varepsilon_i \), $i=1, \ldots, n$ are independent realizations of a circular variable $\varepsilon$ such that $\mathbb{E}[\sin(\varepsilon) \mid \bm{X} = \bm{x}, \bm{Z} = \bm{z}] = 0$ and having finite concentration. Throughout, define 
\( \ell(\bm{x}, \bm{z}) = \mathbb{E}[\cos(\varepsilon) \mid \bm{X} = \bm{x}, \bm{Z} = \bm{z}] \), \( \sigma_1^2(\bm{x}, \bm{z}) = \Var[\sin(\varepsilon) \mid \bm{X} = \bm{x}, \bm{Z} = \bm{z}] \),  \( \sigma_2^2(\bm{x}, \bm{z}) = \Var[\cos(\varepsilon) \mid \bm{X} = \bm{x}, \bm{Z} = \bm{z}] \),
and \( \sigma_{12}(\bm{x}, \bm{z}) = \mathbb{E}[\sin(\varepsilon) \cos(\varepsilon) \mid \bm{X} = \bm{x}, \bm{Z} = \bm{z}] \) 
(equivalently, the conditional covariance of $\sin(\varepsilon)$ and $\cos(\varepsilon)$).

As noted in Section~\ref{sec:statistical_model}, the circular regression function can be expressed as:
\[
m(\bm{x}, \bm{z}) = \operatorname{atan2}[m_1(\bm{x}, \bm{z}), m_2(\bm{x}, \bm{z})],
\]
where \( m_1(\bm{x}, \bm{z}) = \mathbb{E}[\sin(\Theta) \mid \bm{X} = \bm{x}, \bm{Z} = \bm{z}] \) and \( m_2(\bm{x}, \bm{z}) = \mathbb{E}[\cos(\Theta) \mid \bm{X} = \bm{x}, \bm{Z} = \bm{z}] \). Equivalently, $m_1$ and $m_2$ are the regression functions in the auxiliary Euclidean models with responses \( \sin(\Theta) \) and \( \cos(\Theta) \), respectively,
\begin{align}
\sin(\Theta_i) &= m_1(\bm{X}_i, \bm{Z}_i) + \xi_i, \quad i=1,\dots,n, \label{model1_SM} \\
\cos(\Theta_i) &= m_2(\bm{X}_i, \bm{Z}_i) + \zeta_i, \quad i=1,\dots,n, \label{model2_SM}
\end{align}
where \( \xi_i \) and \( \zeta_i \) are bounded error terms with conditional mean zero. Define their variances and covariance as: \( s_1^2(\bm{x}, \bm{z}) = \Var(\xi \mid \bm{X} = \bm{x}, \bm{Z} = \bm{z}) \), \( s_2^2(\bm{x}, \bm{z}) = \Var(\zeta \mid \bm{X} = \bm{x}, \bm{Z} = \bm{z}) \), \( c(\bm{x}, \bm{z}) = \mathbb{E}(\xi \zeta \mid \bm{X} = \bm{x}, \bm{Z} = \bm{z}) \).

Using trigonometric identities in \eqref{eq:model_SM} gives
\begin{align}
\sin(\Theta_i) &= \sin[m(\bm{X}_i, \bm{Z}_i)]\cos(\varepsilon_i) + \cos[m(\bm{X}_i, \bm{Z}_i)]\sin(\varepsilon_i), \label{eq_sin_SM} \\
\cos(\Theta_i) &= \cos[m(\bm{X}_i, \bm{Z}_i)]\cos(\varepsilon_i) - \sin[m(\bm{X}_i, \bm{Z}_i)]\sin(\varepsilon_i). \label{eq_cos_SM}
\end{align}
Let
\[ f_1(\bm{x}, \bm{z}) = \sin[m(\bm{x}, \bm{z})], \quad f_2(\bm{x}, \bm{z}) = \cos[m(\bm{x}, \bm{z})]. \]
Taking conditional expectations in \eqref{eq_sin_SM}--\eqref{eq_cos_SM} yields
\[ m_1(\bm{x}, \bm{z}) = f_1(\bm{x}, \bm{z}) \ell(\bm{x}, \bm{z}), \quad m_2(\bm{x}, \bm{z}) = f_2(\bm{x}, \bm{z}) \ell(\bm{x}, \bm{z}). \]
This shows that \( m_1 \) and \( m_2 \) are scaled versions of \( f_1 \) and \( f_2 \), with \( \ell(\bm{x}, \bm{z}) = [m_1^2(\bm{x}, \bm{z}) + m_2^2(\bm{x}, \bm{z})]^{1/2} \) being the mean resultant length of \( \Theta \) (or equivalently, of \( \varepsilon \)).

Combining expressions \eqref{model1_SM}--\eqref{eq_cos_SM}, the auxiliary errors can be written as
\begin{align*}
\xi_i &= f_1(\bm{X}_i, \bm{Z}_i)[\cos(\varepsilon_i) - \ell(\bm{X}_i, \bm{Z}_i)] + f_2(\bm{X}_i, \bm{Z}_i)\sin(\varepsilon_i), \\
\zeta_i &= f_2(\bm{X}_i, \bm{Z}_i)[\cos(\varepsilon_i) - \ell(\bm{X}_i, \bm{Z}_i)] - f_1(\bm{X}_i, \bm{Z}_i)\sin(\varepsilon_i). 
\end{align*}
Therefore, the variances and covariance in \eqref{model1_SM}--\eqref{model2_SM} satisfy
\begin{align}
s_1^2(\bm{x}, \bm{z}) &= f_1^2 \sigma_2^2 + 2 f_1 f_2 \sigma_{12} + f_2^2 \sigma_1^2, \label{cvar1_SM} \\
s_2^2(\bm{x}, \bm{z}) &= f_2^2 \sigma_2^2 - 2 f_1 f_2 \sigma_{12} + f_1^2 \sigma_1^2, \label{cvar2_SM} \\
c(\bm{x}, \bm{z}) &= f_1 f_2 \sigma_2^2 + f_2^2 \sigma_{12} - f_1^2 \sigma_{12} - f_1 f_2 \sigma_1^2. \label{ccovar_SM}
\end{align}

A Nadaraya--Watson-type kernel estimator of \( m(\bm{x}, \bm{z}) \) is
\begin{equation}
\hat{m}_{\bm{H}}(\bm{x}, \bm{z}) =\textrm{atan2}\left[ \hat{m}_{1,\bm{H}}(\bm{x}, \bm{z}), \hat{m}_{2,\bm{H}}(\bm{x}, \bm{z})  \right],
\label{eq:est_SM}
\end{equation}
where
\begin{equation*}
\hat{m}_{1,\bm{H}}(\bm{x}, \bm{z})=\sum_{i=1}^n w^{i}_{\bm{H}}(\bm{x}, \bm{z}) \sin (\Theta_i)
\end{equation*} 
and
\begin{equation*}
\hat{m}_{2,\bm{H}}(\bm{x}, \bm{z})=\sum_{i=1}^n w^{i}_{\bm{H}}(\bm{x}, \bm{z}) \cos (\Theta_i),
\end{equation*}
with
\begin{equation*}
w^{i}_{\bm{H}}(\bm{x}, \bm{z}) = \frac{K_{\bm{h}}(\bm{x}, \bm{X}_i)\,L_{\boldsymbol{\lambda}}(\bm{z}, \bm{Z}_i)}{\sum_{j=1}^n K_{\bm{h}}(\bm{x}, \bm{X}_j)\,L_{\boldsymbol{\lambda}}(\bm{z}, \bm{Z}_j)}.
\end{equation*}

To derive the asymptotic bias and variance of the circular estimator in \eqref{eq:est_SM}, we first study the large-sample behavior of the component Nadaraya--Watson estimators $\hat m_{j,\bm H}(\bm x,\bm z)$ of $m_j(\bm x,\bm z)$, $j=1,2$ (Proposition~\ref{prop1}). This follows standard mixed-covariate kernel-regression arguments and can be obtained by adapting the results of \citet{racine2004nonparametric} to our setting. The derivations require a set of regularity conditions, which we state next.

\begin{enumerate}[label=(A\arabic*)]
\item 
Let $\bm h_n=(h_{1,n},\ldots,h_{k,n})$ and $\boldsymbol\lambda_n=(\lambda_{1,n},\ldots,\lambda_{p,n})$ denote the
bandwidths for the $k$ continuous and $p$ categorical covariates. Assume that
$h_{j,n}=c_1\,n^{-1/(4+k+p)}$, for $j=1,\ldots,k$, and
$\lambda_{l,n}=c_2\,n^{-2/(4+k+p)}$, for $l=1,\ldots,p$,
for fixed constants $c_1,c_2>0$.
For notational simplicity, we further assume $h_{j,n}\equiv h_n$ for $j=1,\ldots,k$ and
$\lambda_{l,n}\equiv \lambda_n$ for $l=1,\ldots,p$.
We also suppress the subscript $n$ and write $h$ and $\lambda$ for
$h_n$ and $\lambda_n$.

\item
The continuous kernel is of product form, $K_{\bm{h}}(\bm{x}, \bm{X}_i) = \prod_{j=1}^k h^{-1} K_{j}[(x_j - X_{ij})/h]$,
%\[
%K_{\bm h_n}(\bm x,\bm X_i)=\prod_{j=1}^k h_n^{-1}\,
%K_j\!\left(\frac{x_j-X_{ij}}{h_n}\right),
%\]
where each $K_j$ is a nonnegative, symmetric, bounded kernel with $\int K_j(u)\,du=1$, $\int u^4K_j(u)\,du<\infty$, and $m$-times continuously differentiable with
$
\int |K_j^{(s)}(u)|\,|u|^s\,du<\infty$, $s=1,\ldots,m,
$
for some $m>\max\{2+\tfrac{4}{k+p},\,1+\tfrac{k+p}{2}\}$.
This condition is satisfied, for example, by the Gaussian kernel.
For notational simplicity, we assume $K_j\equiv K$, for $j=1,\ldots,k$.

%\item
The categorical kernel factorizes as
$
L_{\boldsymbol\lambda}(\bm z,\bm Z_i)=\prod_{l=1}^p L_l(z_l,Z_{il};\lambda),
$
where each univariate kernel $L_l(\cdot,\cdot;\lambda)$ is nonnegative and uniformly bounded in $\lambda$.
Moreover, for each $l$ and each fixed category $z_l$, the kernel concentrates its mass at $z_l$ in the sense that, as $\lambda\to 0,
$,
$
L_l(z_l,z_l;\lambda)=1+O(\lambda)$ and $\sum_{\tilde z_l\neq z_l} L_l(z_l,\tilde z_l;\lambda)=O(\lambda)$, uniformly over $z_l$.
This condition is satisfied by the Aitchison--Aitken kernel and by standard order-sensitive categorical kernels
considered in \citet{racine2004nonparametric}.
For notational simplicity, we take $L_l\equiv L$, for $l=1,\ldots,p$.

 \item The joint distribution of $(\bm X,\bm Z)$ admits a joint density-mass function $f(\bm x,\bm z)$, where
$\bm x\in\mathbb R^k$ is continuous and $\bm z\in\mathcal Z$ is discrete. Moreover, $f(\bm x,\bm z)$ is bounded
and continuously differentiable in $\bm x$ uniformly in $\bm z$, with first- and second-order partial derivatives
(w.r.t.\ $\bm x$) bounded on compact subsets of $\mathbb R^k\times\mathcal Z$, and $f(\bm x,\bm z)$ is bounded
away from zero on the region where the kernel weights are non-negligible.

\item For the auxiliary models \eqref{model1_SM}--\eqref{model2_SM}, the error terms have finite fourth moments:
$\mathbb{E}(\xi_i^4)<\infty$ and $\mathbb{E}(\zeta_i^4)<\infty$.

\item 
For $j=1,2$, the functions $m_j(\bm x,\bm z)$ and $s_j^2(\bm x,\bm z)$ belong to $C_2^4$ in $\bm x$,
uniformly in $\bm z\in\mathcal{Z}$.
Let $\ell(\bm x,\bm z)=[m_1^2(\bm x,\bm z)+m_2^2(\bm x,\bm z)]^{1/2}$, and assume that
$\ell(\bm x,\bm z)>0$ at the interior point $(\bm x,\bm z)$ under consideration, so that the
$\operatorname{atan2}$ map is locally well defined.

\item For $j=1,2$, define the continuous-smoothing  component
\[
B_{1,1}^{(j)}(\bm{x}, \bm{z})
=
\left\{
\nabla_{\bm{x}} f(\bm{x}, \bm{z})^\top \nabla_{\bm{x}} m_j(\bm{x}, \bm{z})
+ \frac{1}{2} f(\bm{x}, \bm{z})\, \operatorname{tr}\!\left[  \nabla_{\bm{x}}^2 m_j(\bm{x}, \bm{z}) \right]
\right\}\mu_2(K),
\]
where $\mu_2(K)=\int_{\mathbb R} u^2K(u)\,du$.
Define also the categorical-smoothing component
\[
B_{1,2}^{(j)}(\bm{x}, \bm{z})
=
\mathbb{E}\!\left[ m_j(\bm{x}, \bm{Z})-m_j(\bm{x}, \bm{z})
\mid \bm{X}=\bm{x},\, d_H(\bm{Z}, \bm{z})=1 \right]\,
\mathbb{P}\!\left(d_H(\bm{Z}, \bm{z})=1\mid \bm{X}=\bm{x}\right),
\]
where $d_H(\bm{Z}, \bm{z})$ denotes the Hamming distance between $\bm{Z}$ and $\bm{z}$.
Let
\[
B_1^{(j)}=\mathbb{E}\!\left\{\left[\frac{B_{1,1}^{(j)}(\bm X,\bm Z)}{f(\bm X,\bm Z)}\right]^2\right\},\qquad
B_2^{(j)}=2\,\mathbb{E}\!\left[\frac{B_{1,1}^{(j)}(\bm X,\bm Z)\,B_{1,2}^{(j)}(\bm X,\bm Z)}{f(\bm X,\bm Z)^2}\right],
\]
\[
B_3^{(j)}=\mathbb{E}\!\left\{\left[\frac{B_{1,2}^{(j)}(\bm X,\bm Z)}{f(\bm X,\bm Z)}\right]^2\right\}.
\]
Assume these expectations exist and are finite, and that
\[
4B_1^{(j)}B_3^{(j)}-\big(B_2^{(j)}\big)^2>0.
\]

\end{enumerate}

\begin{proof}[Proof of Proposition~1]
Fix an interior point $(\bm x,\bm z)$ in the support of $f$.
For $j=1,2$, $\hat m_{j,\bm H}(\bm x,\bm z)$ is a mixed-data Nadaraya--Watson estimator
applied to responses $\sin(\Theta_i)$ and $\cos(\Theta_i)$, respectively.
Under assumptions (A1)--(A6), the general mixed-kernel expansions in
\citet[Theorem~2.1]{racine2004nonparametric} apply to each Cartesian component.
Specializing that result to our product-kernel weights and rewriting the categorical contribution
in terms of Hamming neighbors yields, for $j=1,2$,
\begin{align*}
\mathbb{E}[\hat m_{j, \bm{H}}(\bm{x},\bm{z})]- m_j(\bm{x}, \bm{z})
=&\,\mu_2(K)h^2 \left\{
\frac{\nabla_{\bm{x}} f(\bm{x}, \bm{z})^\top \nabla_{\bm{x}} m_j(\bm{x}, \bm{z})}{f(\bm{x}, \bm{z})}
+ \frac{1}{2} \operatorname{tr} \left[ \nabla_{\bm{x}}^2 m_j(\bm{x}, \bm{z}) \right]
\right\} \\
&\quad+\lambda \sum_{\substack{\tilde{\bm{z}} \in \mathcal{Z} \\ d_H(\tilde{\bm{z}}, \bm{z}) = 1}}
\left[ m_j(\bm{x}, \tilde{\bm{z}}) - m_j(\bm{x}, \bm{z}) \right]
 \frac{f(\bm{x}, \tilde{\bm{z}})}{f(\bm{x}, \bm{z})}
+ o(h^2 + \lambda),
\end{align*}
which matches the bias expression stated in the proposition.

For the variance, again by \citet[Theorem~2.1]{racine2004nonparametric}, at an interior point
\[
\Var\!\big[\hat m_{j,\bm H}(\bm x,\bm z)\big]
=
\frac{R(K)\,s_j^2(\bm x,\bm z)}{n\,h^k\,f(\bm x,\bm z)}
+o\!\left(\frac{1}{n h^k}\right),
\]
where $R(K)=\big(\int K^2(u)\,du\big)^k$ for the product kernel and $s_j^2(\bm x,\bm z)$ is the conditional
variance of the auxiliary error in \eqref{model1_SM}--\eqref{model2_SM}. The categorical kernel enters only
through bounded multiplicative constants arising from the mixed-kernel normalization. Under the concentration
property in (A3) these constants remain bounded as $\lambda\to 0$ and, therefore, do not affect the leading
$1/(n h^k)$ rate.

Finally, since $\hat m_{1,\bm H}$ and $\hat m_{2,\bm H}$ use the same weights, the same mixed-kernel variance
expansion yields
\[
\Cov\!\big[\hat m_{1,\bm H}(\bm x,\bm z),\hat m_{2,\bm H}(\bm x,\bm z)\big]
=
\frac{R(K)\,c(\bm x,\bm z)}{n\,h^k\,f(\bm x,\bm z)}
+o\!\left(\frac{1}{n h^k}\right),
\]
where $c(\bm x,\bm z)=\Cov(\xi,\zeta\mid \bm X=\bm x,\bm Z=\bm z)$ is the conditional covariance in the auxiliary
models. This completes the proof.
\end{proof}

\begin{proof}[Proof of Theorem~1]
Fix an interior point $(\bm x,\bm z)$ in the support of $f$ and let
\[
\bm m(\bm x,\bm z)=\big[m_1(\bm x,\bm z),m_2(\bm x,\bm z)\big]^\top
\quad \text{and}  \quad
\hat{\bm m}_{\bm H}(\bm x,\bm z)=\big[\hat m_{1,\bm H}(\bm x,\bm z),\hat m_{2,\bm H}(\bm x,\bm z)\big]^\top.
\]
Define $g:\mathbb R^2\to(-\pi,\pi]$ by $g(u_1,u_2)=\operatorname{atan2}(u_1,u_2)$, so that
\[
m(\bm x,\bm z)=g[\bm m(\bm x,\bm z)] \quad \text{and}  \quad
\hat m_{\bm H}(\bm x,\bm z)=g[\hat{\bm m}_{\bm H}(\bm x,\bm z)].
\]
Since $\ell(\bm x,\bm z)=\|\bm m(\bm x,\bm z)\|>0$ by (A5), a second-order Taylor expansion of $g$ around
$\bm m(\bm x,\bm z)$ yields
\begin{equation}\label{eq:taylor_g_SM}
\begin{aligned}
\hat m_{\bm H}(\bm x,\bm z)-m(\bm x,\bm z)
&=
\nabla g\big[\bm m(\bm x,\bm z)\big]^\top
\Big[\hat{\bm m}_{\bm H}(\bm x,\bm z)-\bm m(\bm x,\bm z)\Big] \\
&\quad+
\frac12\Big[\hat{\bm m}_{\bm H}(\bm x,\bm z)-\bm m(\bm x,\bm z)\Big]^\top
\nabla^2 g \big[\bm m(\bm x,\bm z)\big]
\Big[\hat{\bm m}_{\bm H}(\bm x,\bm z)-\bm m(\bm x,\bm z)\Big] \\
&\quad + r_n, 
\end{aligned}
\end{equation}
where $r_n=o_p\big[h^2+\lambda+(nh^k)^{-1/2}\big]$.
For $g(u_1,u_2)=\operatorname{atan2}(u_1,u_2)$ we have
\[
\nabla g(u_1,u_2)=\frac{1}{u_1^2+u_2^2}\,(u_2,\,-u_1)^\top,
\]
and the entries of $\nabla^2 g(u_1,u_2)$ are $O(\|(u_1,u_2)\|^{-2})$.

Taking expectations in \eqref{eq:taylor_g_SM} and using Proposition~1 gives
\begin{align}\label{eq:Etaylor_g_SM}
\mathbb E\big[\hat m_{\bm H}(\bm x,\bm z)-m(\bm x,\bm z)\big]
&=
\nabla g \big[\bm m(\bm x,\bm z)\big]^\top
\mathbb E \big[\hat{\bm m}_{\bm H}(\bm x,\bm z)-\bm m(\bm x,\bm z)\big]
\nonumber\\
&\quad+
\frac12\,\operatorname{tr}\!\left\{
\nabla^2 g \big[\bm m(\bm x,\bm z)\big]\,
\Var\!\big[\hat{\bm m}_{\bm H}(\bm x,\bm z)\big]
\right\}
+o(h^2+\lambda).
\end{align}
The first term in \eqref{eq:Etaylor_g_SM} yields the contributions of order $O(h^2)$ and $O(\lambda)$. The second term is $O(1/(nh^k))$, which is $o(h^2)$ under assumption (A1). Consequently, it is absorbed into the remainder term $o(h^2+\lambda)$.

Writting $\bm m(\bm x,\bm z)=[m_1(\bm x,\bm z),m_2(\bm x,\bm z)]^\top$ and
$\ell(\bm x,\bm z)^2=m_1(\bm x,\bm z)^2+m_2(\bm x,\bm z)^2$, the linear term can be written as
\begin{equation}\label{eq:linear_term_bias}
\begin{aligned}
\nabla g\big[\bm m(\bm x,\bm z)\big]^\top\,
\mathbb E\big[\hat{\bm m}_{\bm H}(\bm x,\bm z)-\bm m(\bm x,\bm z)\big]
&=
\frac{1}{\ell(\bm x,\bm z)^2}\left\{
m_2(\bm x,\bm z)\,\mathrm{Bias} \big[\hat m_{1,\bm H}(\bm x,\bm z)\big] \right. \\
&\quad
\left. - m_1(\bm x,\bm z)\,\mathrm{Bias} \big[\hat m_{2,\bm H}(\bm x,\bm z)\big]
\right\} \\
&\quad+o(h^2+\lambda).
\end{aligned}
\end{equation}
Substituting the $O(h^2)$ biases from Proposition~1 and using
$m_1(\bm x,\bm z)=\ell(\bm x,\bm z)\sin[m(\bm x,\bm z)]$ and
$m_2(\bm x,\bm z)=\ell(\bm x,\bm z)\cos[m(\bm x,\bm z)]$
yields the continuous-smoothing contribution in the statement, namely
\[
\mu_2(K)h^2\left\{
\frac{1}{\ell(\bm x,\bm z)\, f(\bm x,\bm z)}\,
\nabla_{\bm x}\big(\ell f\big)(\bm x,\bm z)^\top \nabla_{\bm x}m(\bm x,\bm z)
+\frac12\,\operatorname{tr}\!\big[\nabla^2_{\bm x}m(\bm x,\bm z)\big]
\right\}.
\]

For the categorical part, Proposition~1 gives, for $j=1,2$,
\[
\mathrm{Bias}_\lambda \big[\hat m_{j,\bm H}(\bm x,\bm z)\big]
=
\lambda \sum_{d_H(\tilde{\bm z},\bm z)=1}
\Big[m_j(\bm x,\tilde{\bm z})-m_j(\bm x,\bm z)\Big]\,
\frac{f(\bm x,\tilde{\bm z})}{f(\bm x,\bm z)}.
\]
Define the increment
\[
\Delta m_j(\bm x,\bm z;\tilde{\bm z})
= m_j(\bm x,\tilde{\bm z})-m_j(\bm x,\bm z),
\qquad j=1,2.
\]
Then
\[
\frac{
m_2(\bm x,\bm z)\,\Delta m_1(\bm x,\bm z;\tilde{\bm z})
-
m_1(\bm x,\bm z)\,\Delta m_2(\bm x,\bm z;\tilde{\bm z})
}{\ell(\bm x,\bm z)^2}
=
\frac{
m_2(\bm x,\bm z)m_1(\bm x,\tilde{\bm z})
-
m_1(\bm x,\bm z)m_2(\bm x,\tilde{\bm z})
}{\ell(\bm x,\bm z)^2}.
\]
Using,
\[
m_2(\bm x,\bm z)m_1(\bm x,\tilde{\bm z})-m_1(\bm x,\bm z)m_2(\bm x,\tilde{\bm z})
=
\ell(\bm x,\bm z)\,\ell(\bm x,\tilde{\bm z})\,
\sin\!\Big[m(\bm x,\tilde{\bm z})-m(\bm x,\bm z)\Big],
\]
we obtain
\[
\frac{
m_2(\bm x,\bm z)\,\Delta m_1(\bm x,\bm z;\tilde{\bm z})
-
m_1(\bm x,\bm z)\,\Delta m_2(\bm x,\bm z;\tilde{\bm z})
}{\ell(\bm x,\bm z)^2}
=
\frac{1}{\ell(\bm x,\bm z)}\,
\ell(\bm x,\tilde{\bm z})\,
\sin\!\Big[m(\bm x,\tilde{\bm z})-m(\bm x,\bm z)\Big].
\]
Summing over all Hamming neighbors and multiplying by
$\lambda f(\bm x,\tilde{\bm z})/f(\bm x,\bm z)$ yields the $O(\lambda)$ term in the stated bias.

For the variance, \eqref{eq:taylor_g_SM} and Proposition~1 imply the first-order delta-method expansion
\[
\Var\!\big[\hat m_{\bm H}(\bm x,\bm z)\big]
=
\nabla g \big[\bm m(\bm x,\bm z)\big]^\top
\Var\!\big[\hat{\bm m}_{\bm H}(\bm x,\bm z)\big]\,
\nabla g \big[\bm m(\bm x,\bm z)\big]
+o\!\left(\frac{1}{nh^k}\right).
\]
Since $\nabla g \big[\bm m(\bm x,\bm z)\big]=\big[m_2(\bm x,\bm z),-m_1(\bm x,\bm z)\big]^\top/\ell(\bm x,\bm z)^2$ and
\[
\Var\!\big[\hat{\bm m}_{\bm H}(\bm x,\bm z)\big]
=
\frac{R(K)}{nh^k f(\bm x,\bm z)}
\begin{pmatrix}
s_1^2(\bm x,\bm z) & c(\bm x,\bm z)\\
c(\bm x,\bm z) & s_2^2(\bm x,\bm z)
\end{pmatrix}
+o\!\left(\frac{1}{nh^k}\right),
\]
we obtain
\[
\Var\!\big[\hat m_{\bm H}(\bm x,\bm z)\big]
=
\frac{R(K)}{nh^k\,\ell(\bm x,\bm z)^2\,f(\bm x,\bm z)}\,
\sigma_1^2(\bm x,\bm z)
+o\!\left(\frac{1}{nh^k}\right),
\]
where $\sigma_1^2(\bm x,\bm z)=\Var(\sin\varepsilon\mid \bm X=\bm x,\bm Z=\bm z)$.

Finally, the effective noise term $\sigma_1^2(\bm x,\bm z)$ can be expressed in terms of the auxiliary-model
quantities as
\[
\sigma_1^2(\bm{x}, \bm{z}) =
\frac{s_1^2(\bm{x}, \bm{z}) m_2^2(\bm{x}, \bm{z}) + s_2^2(\bm{x}, \bm{z}) m_1^2(\bm{x}, \bm{z})
- 2 c(\bm{x}, \bm{z}) m_1(\bm{x}, \bm{z}) m_2(\bm{x}, \bm{z})}{m_1^2(\bm{x}, \bm{z}) + m_2^2(\bm{x}, \bm{z})}.
\]
Using equations \eqref{cvar1_SM}--\eqref{ccovar_SM} and considering  the identities
$m_1(\bm x,\bm z)=\ell(\bm x,\bm z)\sin\!\big[m(\bm x,\bm z)\big]$ and
$m_2(\bm x,\bm z)=\ell(\bm x,\bm z)\cos\!\big[m(\bm x,\bm z)\big]$,
the above expression simplifies to
$\sigma_1^2(\bm x,\bm z)=\Var(\sin\varepsilon\mid \bm X=\bm x,\bm Z=\bm z)$,
which completes the proof.
\end{proof}

\section{Additional results from the real-data application}
\label{app:data}

This section collects complementary analyses for the real-data study in the Section~\ref{sec:results}.
%In Section~\ref{app:desc_errors}, we provide a brief descriptive summary of the circular response (directional error)
%across experimental conditions, using condition-wise polar density representations and violin plots.
In Section~\ref{app:selectors}, we compare the condition-specific regression fits obtained under the three
bandwidth selectors considered in the paper (cross-validation, rule-of-thumb, and the proposed bootstrap criterion),
thereby complementing the bootstrap-based curves and simultaneous confidence bands shown in that section.
In Section~\ref{app:residuals}, we report additional residual diagnostics and circular goodness-of-fit to those shown in Section~\ref{sec:resi}.
Finally, in Section~\ref{app:alt_predictor}, we study an alternative specification in which a different continuous
predictor is included, and we provide both global and condition-specific estimates together with simultaneous
bootstrap confidence bands.

\subsection{Condition-specific regression fits under alternative bandwidth selectors}
\label{app:selectors}

The response variable in our real-data analysis is the \emph{directional error}, defined as the signed angular
discrepancy between the participant's reported target direction and the true target direction, wrapped to the
principal interval $(-\pi,\pi]$. Thus, values close to $0$ correspond to nearly perfect updating, whereas values
near $\pm\pi$ indicate very large errors close to a reversal.

In Section~\ref{sec:regs}, we report the bootstrap-based condition-specific fits and simultaneous confidence bands.
To assess sensitivity to bandwidth choice, this section compares the condition-specific regression curves obtained
with the three selectors considered: cross-validation, rule-of-thumb, and the proposed bootstrap criterion.
In this application, the resulting bandwidth vectors were
$\bm H_{\text{CV}}=(0.31,\,0.12)$, $\bm H_{\text{RoT}}=(0.37,\,0.05)$, and
$\bm H_{\text{boot}}=(0.28,\,0.08)$.

Figure~\ref{fig:selectors_by_condition} displays the three fitted curves for each sensory condition.
Across all conditions, the differences among selectors are minor: the three curves are nearly indistinguishable over
most of the covariate range, and they lead to the same substantive conclusions as those discussed in Section~\ref{sec:regs}
based on the bootstrap selector. In particular, any condition-specific bias patterns and the increase in variability
under sensory restriction remain essentially unchanged. Therefore, the present comparison is best viewed as a
robustness check of the bandwidth choice. Since all three selectors yield essentially the same fit, we refer the
reader to the discussion in Section~\ref{sec:regs}
for the substantive interpretation.

\begin{figure}[!htb]
\centering
\begin{subfigure}[b]{0.32\textwidth}
    \includegraphics[width=\textwidth]{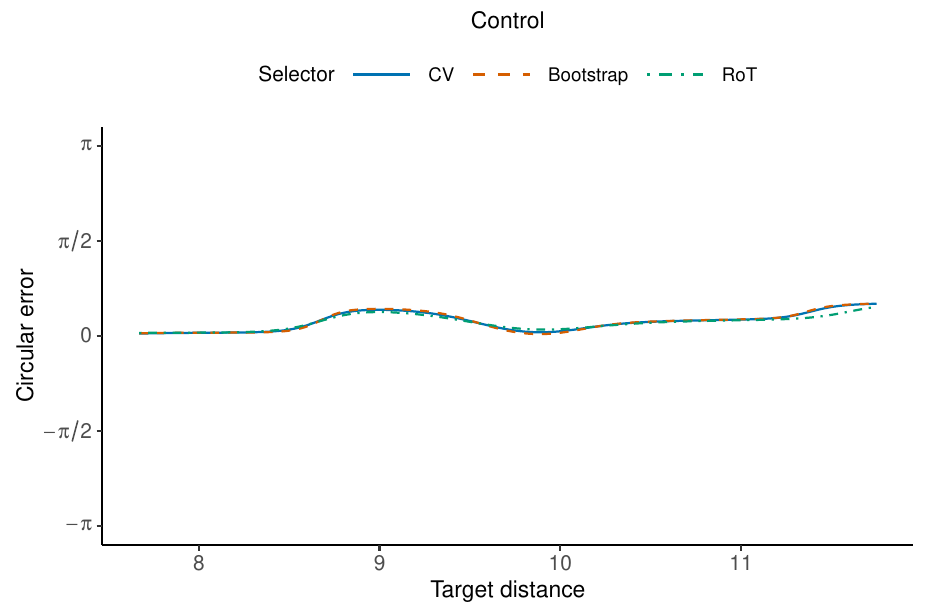}
\end{subfigure}
\begin{subfigure}[b]{0.32\textwidth}
    \includegraphics[width=\textwidth]{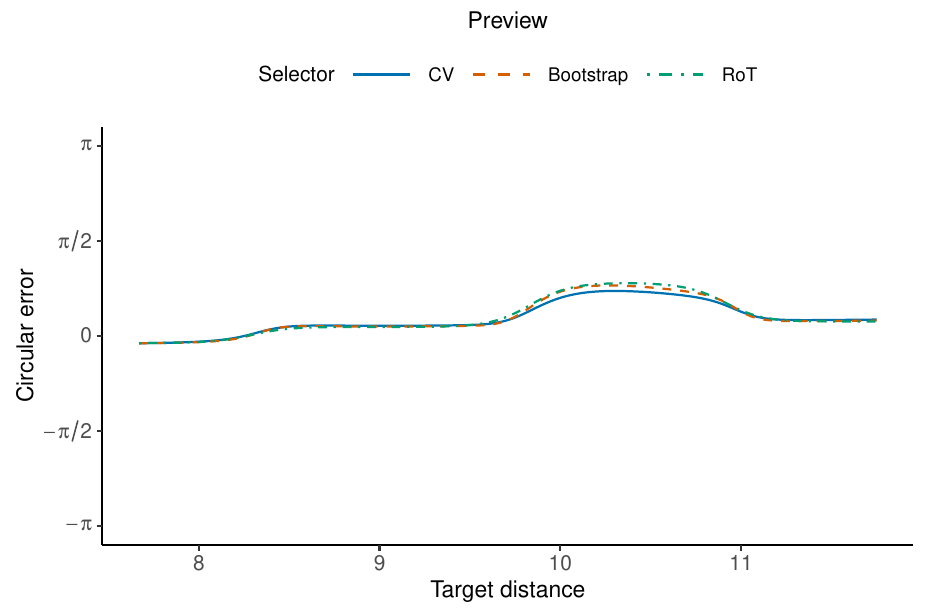}
\end{subfigure}
\begin{subfigure}[b]{0.32\textwidth}
    \includegraphics[width=\textwidth]{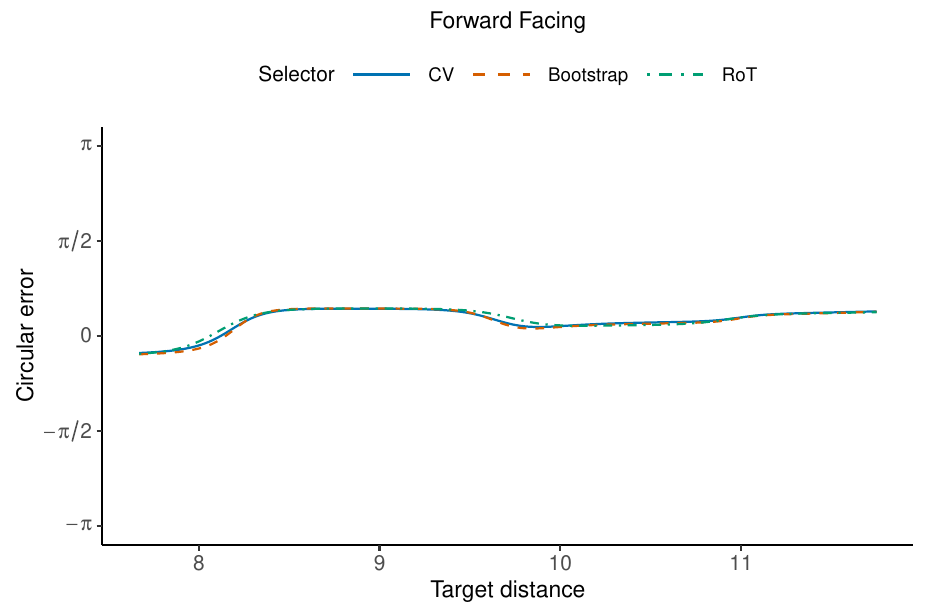}
\end{subfigure}

\vspace{0.6em}

\begin{minipage}{0.66\textwidth}
\centering
\begin{subfigure}[b]{0.48\textwidth}
  \includegraphics[width=\textwidth]{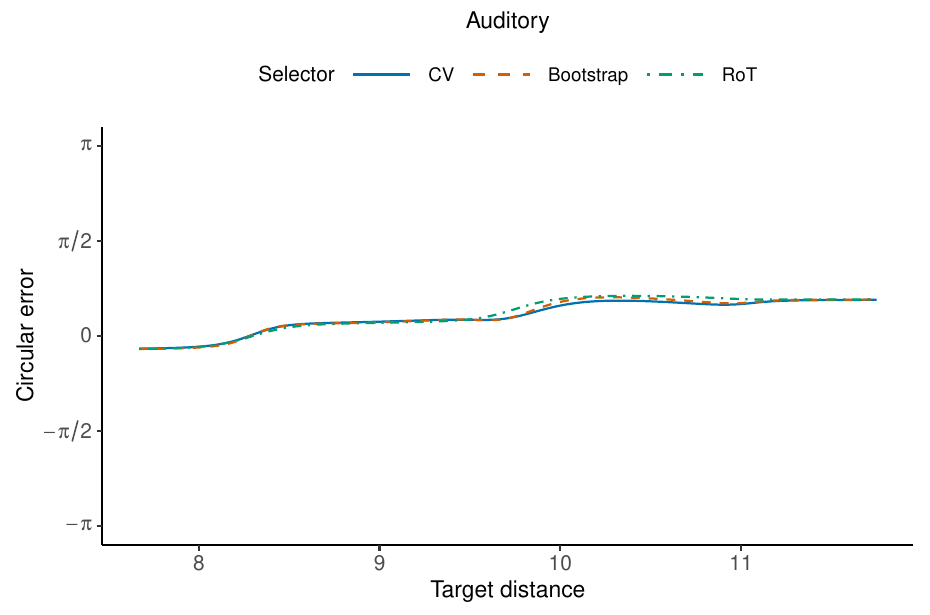}
\end{subfigure}
\hfill
\begin{subfigure}[b]{0.48\textwidth}
  \includegraphics[width=\textwidth]{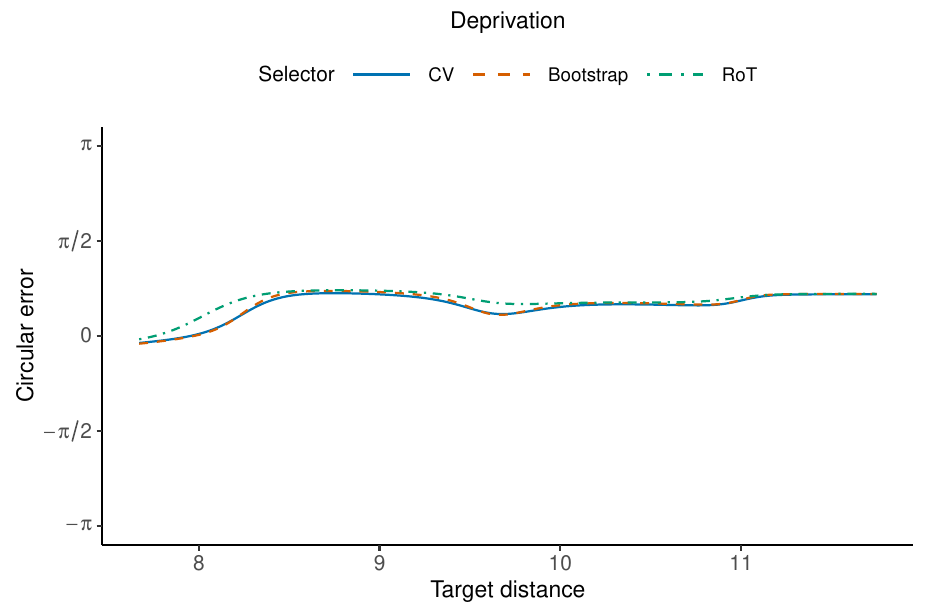}
\end{subfigure}
\end{minipage}

\caption{
Estimated regression curves for angular error as a function of target distance, stratified by sensory condition. Each panel shows one experimental condition with circular regression fits obtained using cross-validation (solid blue line), bootstrap (dashed orange line), and rule-of-thumb (dot-dashed green line) bandwidth selectors.}
\label{fig:selectors_by_condition}
\end{figure}

\subsection{Residual analysis and circular goodness-of-fit metrics}
\label{app:residuals}

% -------------------------------------------------------------------------
% Residual analysis and circular goodness-of-fit metrics (SM)
% -------------------------------------------------------------------------
This section complements the brief residual analysis reported in Section~\ref{sec:resi}. We provide the pooled residuals-versus-fitted diagnostic omitted for space, pooled tests of circular uniformity, and condition-wise summaries of the observational cosine loss $
\operatorname{CASE}_{\text{obs}}$. Unless otherwise noted, all results refer to the bootstrap-based fit $\hat m_{\bm H_{\text{boot}}}$. A short comparison across bandwidth selectors, based on $\operatorname{CASE}_{\text{obs}}$ and $R^2_{\text{circ}}$ (pooled and by condition), is reported at the end of this section.

Let $\Theta_i\in(-\pi,\pi]$ denote the observed directional error (in radians) and
$\hat m_{\bm H}(\bm X_i,\bm Z_i)$ the fitted value. Define the circular residuals as
\[
\hat{\varepsilon}_i
=
\bigl[(\Theta_i-\hat m_{\bm H}(\bm X_i,\bm Z_i)+\pi)\ (\operatorname{mod}\,2\pi)\bigr]-\pi,
\]
so that $\hat{\varepsilon}_i\in(-\pi,\pi]$.
This preserves a signed interpretation and is invariant under global rotations.
Because $\hat m_{\bm H}$ is estimated from the data, the residuals $\hat\varepsilon_i$ are not independent in a
strict sense, and the repeated-trial structure (multiple observations per participant) may further induce
within-subject dependence. Therefore, the circular uniformity tests reported below should be interpreted as
diagnostic checks rather than exact level-$\alpha$ procedures.

%Let $\Theta_i\in(-\pi,\pi]$ denote the observed directional error (in radians) and $\hat m_{\bm H}(\bm X_i,\bm Z_i)$ the fitted value. Circular residuals are defined as
%\[
%\hat{\varepsilon}_i
%= \big[\Theta_i - \hat m_{\bm H}(\bm X_i,\bm Z_i)\big]\quad (\operatorname{mod}\, 2\pi),
%\]
%using the wrap operator to $(-\pi,\pi]$,
%which preserves a signed interpretation and is invariant to global rotations. Because residuals arise from a smoothed nonparametric fit, they are not strictly i.i.d. Therefore, the formal tests below should be interpreted as heuristics rather than exact level-$\alpha$ procedures.

Figure~\ref{fig:residuals_vs_fitted_sm} displays the pooled scatterplot of residuals versus fitted values (both wrapped to $(-\pi,\pi]$). We assess trends, curvature, and changes in dispersion across fitted angles. No discernible trend or heteroscedasticity is observed, consistent with the condition-wise circular boxplots reported in Figure~\ref{fig:residuals_boxplot_bootstrap}, suggesting that the fitted regression captures the dominant structure of the angular response at the pooled level.

\begin{figure}[!htb]
\centering
\includegraphics[width=0.72\textwidth]{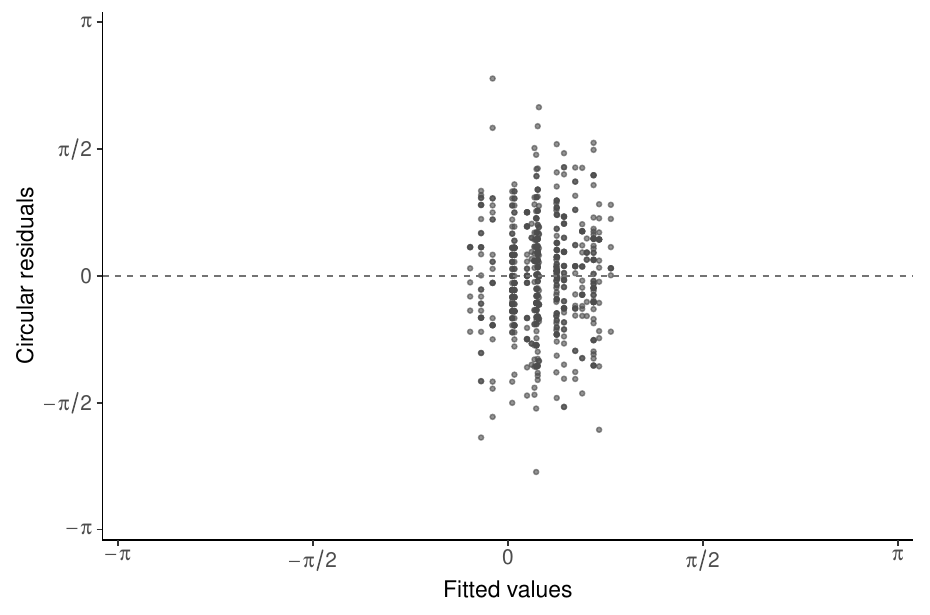}
\caption{Residuals versus fitted values for the bootstrap-based estimator. Residuals are wrapped to $(-\pi,\pi]$. No visible trend or heteroscedasticity is observed.}
\label{fig:residuals_vs_fitted_sm}
\end{figure}

Complementing Figure~\ref{fig:residuals_vs_fitted_sm}, we also tested whether the pooled residuals are compatible with circular uniformity, $H_0:\hat\varepsilon_i\sim \mathrm{Unif}(-\pi,\pi)$. Specifically, we report three standard tests: Rayleigh’s test \citep{mardia2009directional}, which is most powerful against unimodal departures driven by a nonzero mean direction; Watson’s $U^2$ test \citep{watson1961goodness}, an omnibus CDF-based test sensitive to multimodality and asymmetry; and Kuiper’s test \citep{kuiper1960tests}, the rotation-invariant analogue of the Kolmogorov--Smirnov test. All three tests reject uniformity (Rayleigh statistic $0.8122$, $p<0.001$; Watson $U^2=24.6620$, $p<0.01$; Kuiper statistic $14.973$, $p<0.01$). In light of the point made above, the residuals are not strictly i.i.d.\ because they are computed from a smoothed fit, these $p$-values should be viewed as heuristic rather than as exact level-$\alpha$ evidence. Nonetheless, the rejections suggest some remaining circular structure not fully captured by the fitted mean model. When testing separately by condition, multiplicity adjustment is recommended. Qualitatively, the strongest departures from uniformity occur under \textit{Deprivation}, with weaker effects in other settings. Taken together with the diagnostics above and the condition-wise residual boxplots in Figure~\ref{fig:residuals_boxplot_bootstrap}, these findings suggest an adequate fit of the primary mean structure, while motivating potential extensions allowing condition-specific dispersion (heteroscedasticity) or hierarchical components (see the discussion in Section~\ref{sec:discussion}).

To quantify in-sample fit, we use two complementary global summaries aligned with the cosine loss used throughout the paper. First, the observational cosine loss (denoted by $\operatorname{CASE}_{\text{obs}}$) is:
\[
\operatorname{CASE}_{\text{obs}}(\bm H)
=\frac{1}{n}\sum_{i=1}^n\left\{1-\cos\!\left[\Theta_i-\hat m_{\bm H}(\bm X_i,\bm Z_i)\right]\right\}.
\]
Unlike the loss criterion used in the simulation study (which compares $\hat m_{\bm H}$ to the true regression function), $\operatorname{CASE}_{\text{obs}}$ depends only on observed responses and fitted values and is therefore directly applicable to real data. It is rotation-invariant and does not depend on the particular wrapping convention (any representation equivalent to $(-\pi,\pi]$ yields the same value). A condition-wise decomposition,
\[
\operatorname{CASE}_{\text{obs}}^{(z)}(\bm H)
=\frac{1}{n_z}\sum_{i:\,Z_i=z}\left\{1-\cos\!\left[\Theta_i-\hat m_{\bm H}(\bm X_i,\bm Z_i)\right]\right\},
\]
highlights how performance varies across sensory settings.

Second, we report a circular analogue of the coefficient of determination,
\[
R^2_{\text{circ}}
=1-\frac{\text{SSE}_{\text{circ}}}{\text{SST}_{\text{circ}}},
\]
with 
\[
\text{SSE}_{\text{circ}}=\sum_{i=1}^n\left\{1-\cos\!\left[\Theta_i-\hat m_{\bm H}(\bm X_i,\bm Z_i)\right]\right\}
\]
and
\[
\text{SST}_{\text{circ}}=\sum_{i=1}^n\left\{1-\cos\!\left[\Theta_i-\bar\Theta\right]\right\},
\]
where $\bar\Theta$ is the sample mean direction. This coefficient is invariant to global rotations and wrapping and provides a transparent, cosine-loss--aligned measure of explained circular variability:
$R^2_{\mathrm{circ}}=1$ for a perfect fit, $R^2_{\mathrm{circ}}=0$ for the (sample) circular mean predictor,
and $R^2_{\mathrm{circ}}<0$ if the model underperforms the circular mean. We emphasize that this is not a unique
canonical circular $R^2$, but rather a convenient, cosine-loss–aligned pseudo-measure. For that reason, we also report $\operatorname{CASE}_{\text{obs}}$, which is the primary loss-based goodness-of-fit summary used throughout the paper. Note that if $\text{SST}_{\mathrm{circ}}=0$, then $R^2_{\text{circ}}$ is undefined. In practice this occurs only when the sample directions are essentially identical. For alternative circular summaries (e.g., functionals of the mean resultant
length, circular correlations, or EDF-based discrepancies), see \citet{mardia2009directional}, \citet{fisher1992statistical},
and \citet{pewsey2013circular}.

Across bandwidth selectors, performance is very similar. For $\bm H_{\text{CV}}$, we obtained $\operatorname{CASE}_{\text{obs}}=0.182$ and $R^2_{\text{circ}}=0.140$; for $\bm H_{\text{boot}}$, $\operatorname{CASE}_{\text{obs}}=0.181$ and $R^2_{\text{circ}}=0.144$; and for $\bm H_{\text{RoT}}$, $\operatorname{CASE}_{\text{obs}}=0.182$ and $R^2_{\text{circ}}=0.139$. The bootstrap selector has a slight edge (lowest $\operatorname{CASE}_{\text{obs}}$, highest $R^2_{\text{circ}}$), but differences are minimal. The modest magnitude of $R^2_{\text{circ}}$ is expected here because the global metric pools trials across sensory conditions with markedly different dispersion. Moreover, under the most demanding conditions residual angles show weak concentration, which limits explainable cosine variability.

A condition-wise view of $\operatorname{CASE}_{\text{obs}}$ clarifies this heterogeneity (Table~\ref{tab:cmase_by_condition_sm}). Fit is tightest under \textit{Control}, followed by \textit{Forward Facing}, and weakest under \textit{Auditory}. The \textit{Preview} and \textit{Deprivation} conditions show intermediate but clearly higher loss than \textit{Control} and \textit{Forward Facing}, with very similar values to each other. Differences among selectors remain small within each condition, with no practically meaningful separation. These stratified results help interpret the pooled summaries: aggregation across heterogeneous sensory contexts attenuates the overall signal, while $\operatorname{CASE}_{\text{obs}}^{(z)}$ identifies where the model performs best.

\begin{table}[!htb]
\centering
\caption{$\operatorname{CASE}_{\text{obs}}^{(z)}$ by condition and bandwidth selector.}
\label{tab:cmase_by_condition_sm}
\begin{tabular}{lccccc}
\toprule
\textbf{Selector} & \textbf{Control} & \textbf{Preview} & \textbf{Forward Facing} & \textbf{Auditory} & \textbf{Deprivation} \\
\midrule
$\bm{H}_{\text{CV}}$  & 0.1504 & 0.2138 & 0.1717 & 0.2388 & 0.2118 \\
$\bm{H}_{\text{boot}}$        & 0.1497 & 0.2116 & 0.1715 & 0.2381 & 0.2110 \\
$\bm{H}_{\text{RoT}}$     & 0.1523 & 0.2110 & 0.1716 & 0.2381 & 0.2108 \\
\bottomrule
\end{tabular}
\end{table}

\subsection{Alternative model: distance error as continuous predictor}
\label{app:alt_predictor}

In the main study in Section~\ref{sec:regs}, the continuous predictor is the actual distance to the target (\textit{target distance}). Here we
consider an alternative specification based on \textit{distance error}, defined at each trial as the perceived
distance reported by the participant (\textit{target distance subject response}) minus the actual target distance.
This variable captures trial-by-trial perceptual misestimation and provides a complementary, cognitively meaningful
covariate. Whereas \textit{target distance} reflects the physical layout of the task, \textit{distance error} quantifies
the participant’s internal distortion of distance and allows us to assess whether directional performance is more
strongly related to objective distance or to subjective misjudgment.

We begin with a pooled (global) exploratory fit of angular error as a function of \textit{distance error}.
Figure~\ref{fig:global_fit_distance_error} shows a jittered scatterplot with a global Nadaraya--Watson circular
regression curve overlaid, using a Gaussian kernel and a rule-of-thumb bandwidth (as in Figure~\ref{fig:scatter_global}). In this
case, the pooled RoT choice is $h_{\mathrm{RoT}}=0.892$. Jitter is applied only for visualization to reduce
overplotting and does not affect the fitted curve.

\begin{figure}[!htb]
\centering
\includegraphics[width=0.85\textwidth]{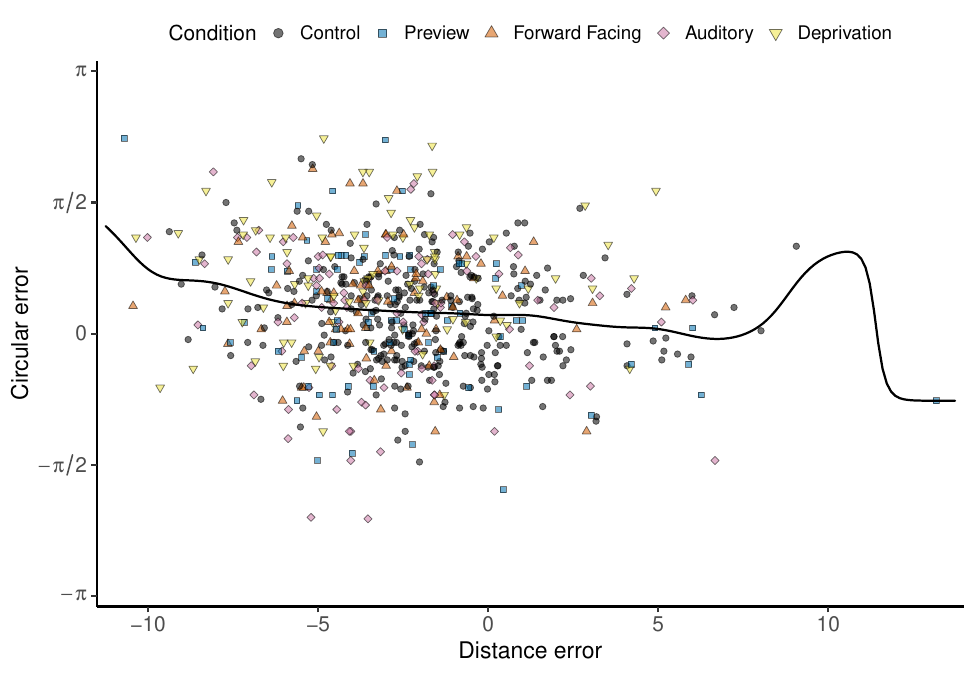}
\caption{Global circular regression curve using \textit{distance error} as predictor with rule-of-thumb bandwidth
$h_{\mathrm{RoT}}=0.892$. A small horizontal jitter was applied to the predictor values for visualization.}
\label{fig:global_fit_distance_error}
\end{figure}

Overall, the pooled fit suggests a mild decrease in the mean circular error as \textit{distance error} increases from
negative values towards zero, with the mean response remaining relatively stable across a broad central range. From a
behavioral standpoint, using \textit{distance error} re-expresses the covariate effect in terms of an internal metric
distortion. When participants \emph{underestimate} distance (negative values), the average directional error tends to
be larger, and it decreases as perceived and actual distances align. At the extremes of the predictor domain, most
noticeably for large positive \textit{distance error}, the pooled curve exhibits local fluctuations. These should be
interpreted cautiously, as they occur where the effective sample size is smaller and the Nadaraya--Watson estimator is
more sensitive to sparse design and to a potentially uneven composition of sensory conditions in that region.

We emphasize that apparent vertical spread can be visually amplified by the linear display of an inherently circular
response (angles near $-\pi$ and $\pi$ represent similar directions), reinforcing the need for circular regression and
diagnostics. Compared with the pooled curve obtained using \textit{target distance} (Figure~\ref{fig:scatter_global}), using \textit{distance error}
re-centers the continuous covariate around zero and expresses the effect in terms of signed perceptual misestimation.
In addition, the pooled curve shows a more visible edge behavior at large positive \textit{distance error}, likely driven
by sparse data and/or an uneven mix of sensory conditions, which further motivates the condition-specific analysis that follows.

To disentangle the effect of perceptual misestimation from sensory context, we fit the same circular regression model
using \textit{distance error} as the continuous covariate and sensory condition as the categorical
predictor. As in the analysis in Section~\ref{sec:regs}, bandwidth selection can be performed via cross-validation, bootstrap, or a
rule-of-thumb criterion. In this dataset (consistently with the \textit{target distance} analysis), the resulting
curves are very similar across selectors, so the choice of bandwidth selector does not affect the substantive
conclusions. Therefore, to keep this supplementary sensitivity check fully reproducible at low computational cost, we
report only the rule-of-thumb fit. The resulting bandwidth was
$\bm{H}_{\text{RoT}} = (0.892,\ 0.054)$, where the first component corresponds to Gaussian smoothing in the continuous
covariate and the second to the categorical smoothing parameter in the Aitchison--Aitken kernel. Figure~\ref{fig:confidence_bands_distance_error}
displays the corresponding condition-specific regression fits and their (calibrated) bootstrap confidence bands.

\begin{figure}[!htb]
\centering
\begin{subfigure}[b]{0.32\textwidth}
    \includegraphics[width=\textwidth]{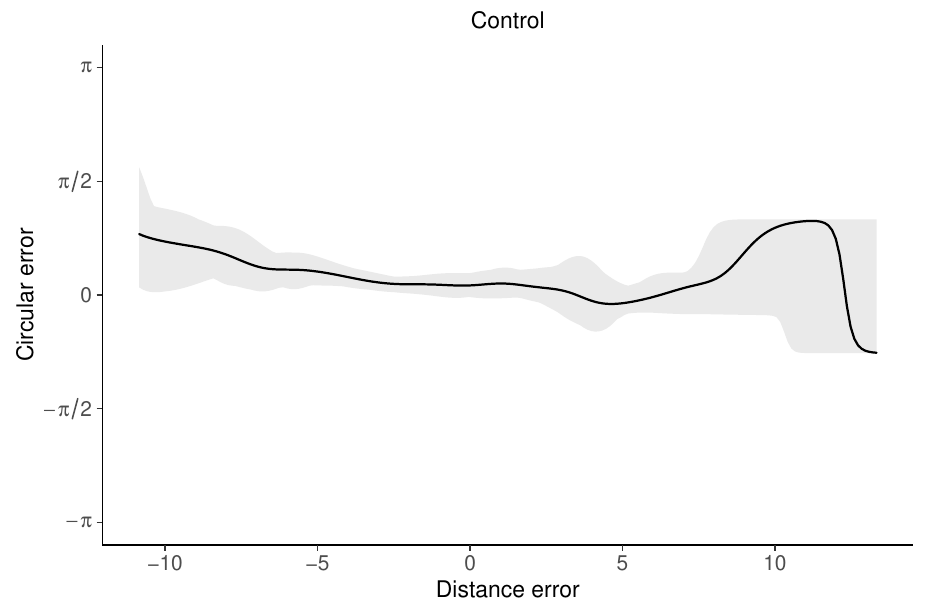}
\end{subfigure}
\begin{subfigure}[b]{0.32\textwidth}
    \includegraphics[width=\textwidth]{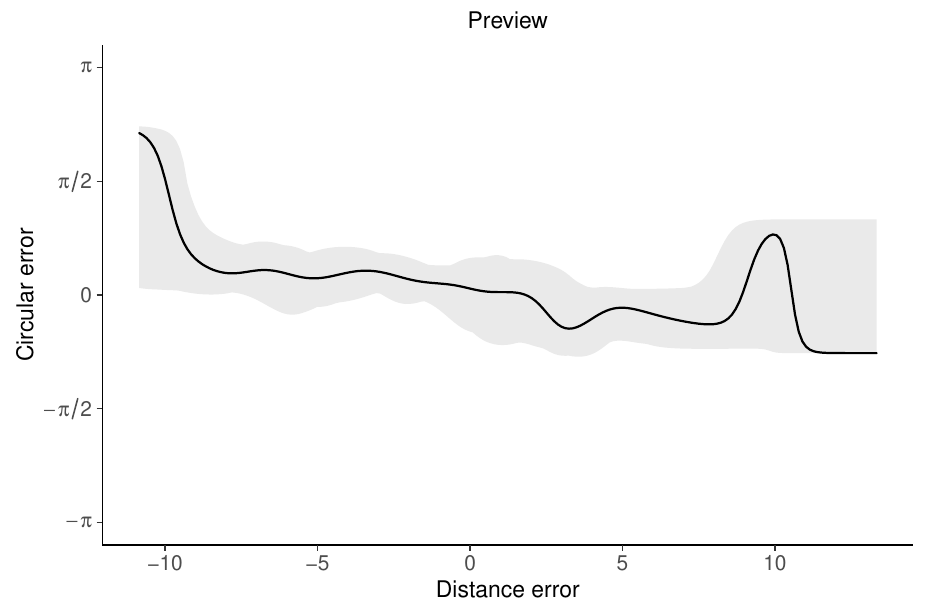}
\end{subfigure}
\begin{subfigure}[b]{0.32\textwidth}
    \includegraphics[width=\textwidth]{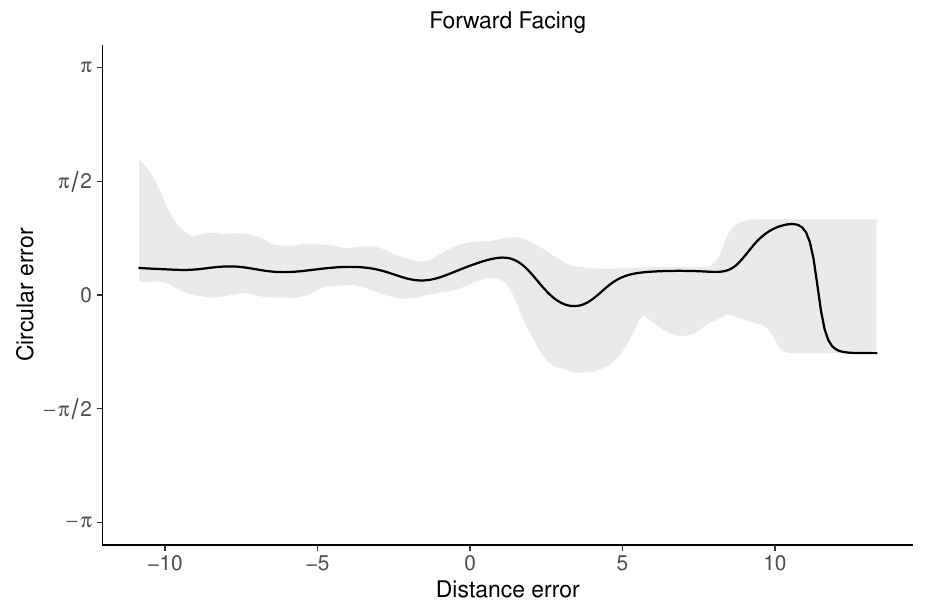}
\end{subfigure}

\vspace{0.6em}

\begin{minipage}{0.66\textwidth}
\centering
\begin{subfigure}[b]{0.48\textwidth}
  \includegraphics[width=\textwidth]{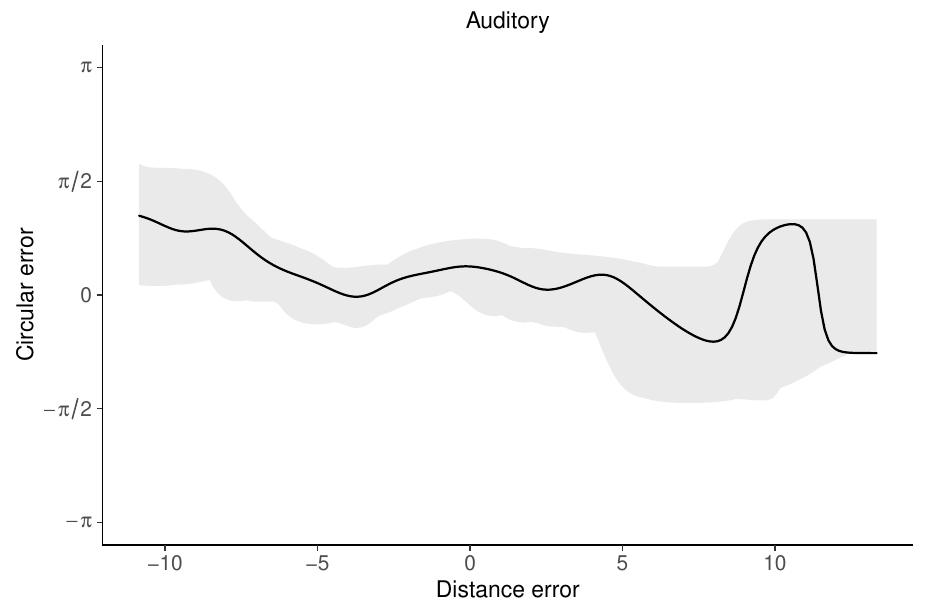}
\end{subfigure}
\hfill
\begin{subfigure}[b]{0.48\textwidth}
  \includegraphics[width=\textwidth]{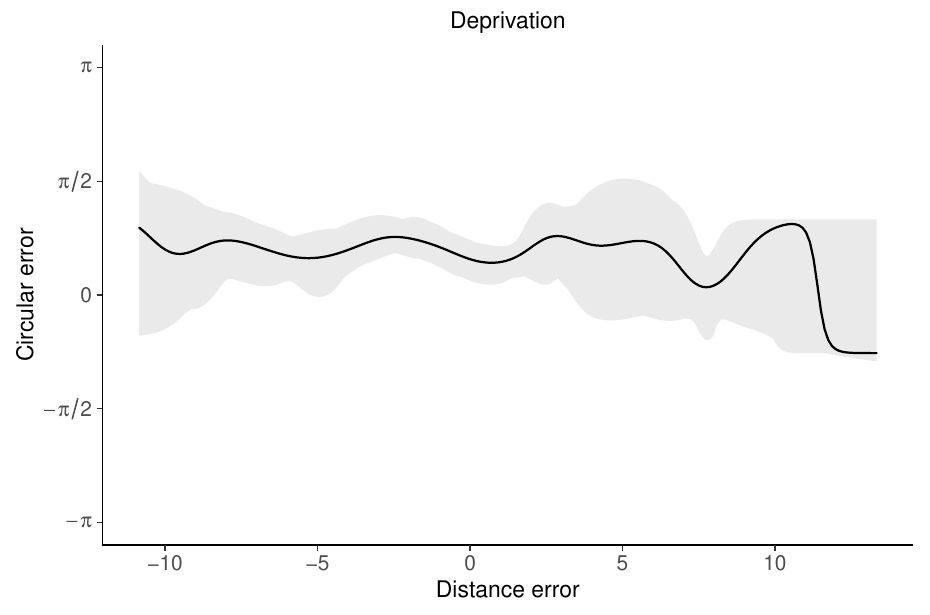}
\end{subfigure}
\end{minipage}

\caption{
Estimated circular regression curves and simultaneous confidence bands for the estimated circular regression functions under each sensory condition using \textit{distance error} as predictor. Bands were computed via bootstrap with \( B = 200 \) resamples and significance level \( \alpha = 0.05 \). 
%The bottom-right panel summarizes the mean angular width of each band.
}
\label{fig:confidence_bands_distance_error}
\end{figure}

As shown in Figure~\ref{fig:confidence_bands_distance_error}, under \textit{Control} and \textit{Preview} the estimated
mean direction is broadly stable over the central range of \textit{distance error}, with comparatively narrower bands.
In \textit{Forward Facing} the fit exhibits a modest dependence on \textit{distance error}, whereas under reduced
sensory input (most notably \textit{Auditory} and \textit{Deprivation}) the curves show a clearer association: mean
directional errors tend to be higher when \textit{distance error} is negative (targets perceived as nearer than they
truly are), and decrease as perceived and actual distances approach agreement. Across conditions, \textit{distance
error} provides an interpretable scale: negative values correspond to underestimation of distance and positive values
to overestimation, facilitating a direct link between changes in the mean direction and trial-by-trial perceptual
misestimation. Interpretation at the extremes should be made with caution: edge behavior is unstable where support is
very sparse (only a few trials). This is evident not only at large positive \textit{distance error} but also, for
\textit{Preview}, at the far negative end, where the fitted curve shows a sharp rise likely driven by limited support
rather than a robust trend.

Confidence bands in Figure~\ref{fig:confidence_bands_distance_error} are computed using the same bootstrap framework as
in Section~\ref{sec:regs}, applied within condition, with $B=200$ resamples and nominal target simultaneous level
$1-\alpha$ with $\alpha=0.05$. As in the that analysis, simultaneous coverage is obtained by calibrating the
pointwise level through the iterative procedure implemented in our software. Qualitatively, the band widths follow the expected pattern.
Uncertainty is smallest in full-vision conditions and increases as sensory information is degraded, with wider and more
irregular bands near the extremes of the predictor domain where data are sparse. 

Taken together, this alternative specification supports the main conclusions while offering a complementary
interpretation: directional updating errors depend both on sensory context and on participants’ internal misestimation
of distance. Using \textit{distance error} re-expresses the covariate effect in terms of a subjective distortion
(perceived minus actual distance), and the condition-specific fits indicate that the
association between directional error and distance misestimation is not uniform across sensory conditions.
In several settings, especially under reduced sensory input, larger directional errors tend to occur when
distance is underestimated (negative values), although the strength and shape of this pattern vary by condition.
Overall, this reinforces the motivation for modeling and inference \emph{by condition}, rather than relying on pooled
summaries that conflate heterogeneous sensory contexts. For completeness, we also ran the same residual diagnostics and
goodness-of-fit checks for this alternative predictor. The qualitative conclusions mirror those for
\textit{target distance}: residuals (wrapped to $(-\pi,\pi]$) show no systematic trend versus fitted values, dispersion
remains strongly condition-dependent, and pooled circular-uniformity tests continue to reject strict uniformity,
indicating that some residual circular structure persists in the most challenging settings and could be addressed by
future extensions (e.g., condition-specific dispersion or hierarchical components).

\end{document}